\documentclass[11pt]{article}
 
\usepackage{amssymb,amsmath,amsfonts}
\usepackage{graphicx,color,enumitem}
\usepackage{mathrsfs}
\usepackage{amsthm}
\usepackage[dvipsnames]{xcolor}
\usepackage{bm}
\usepackage{comment}
\usepackage{shuffle}
\usepackage[round]{natbib}
\usepackage[]{appendix}
\usepackage{xcolor}
\usepackage[a4paper, total={6in, 9in}]{geometry}
\usepackage{subcaption}
\usepackage{float}
\usepackage{subcaption}

\RequirePackage[colorlinks,citecolor=blue!70!green,urlcolor=blue, linkcolor=blue!70!green]{hyperref}

\usepackage{tikz}
\tikzstyle{vertex}=[circle, draw, inner sep=2pt, fill=white]

\newcommand{\Var}{{\text{Var}}}

\renewcommand{\d}{{\mathrm{d}}}

\newcommand{\e}{{\varepsilon}}

\newcommand{\E}{{\mathbb E}}

\newcommand{\Q}{{\mathbb Q}}

\newcommand{\R}{{\mathbb R}}
\renewcommand{\S}{{\mathbb S}}
\newcommand{\N}{{\mathbb N}}

\newcommand{\X}{{\mathbb X}}
\newcommand{\Y}{{\mathbb Y}}
\newcommand{\Z}{{\mathbb Z}}

\newcommand{\Fcal}{{\mathcal F}}

\newcommand{\Ical}{{\mathcal I}}

\newcommand{\Kcal}{{\mathcal K}}
\newcommand{\Lcal}{{\mathcal L}}
\newcommand{\Lvar}{{\mathscr L}}

\newcommand{\Rcal}{{\mathcal R}}
\newcommand{\Tcal}{{\mathcal T}}
\newcommand{\VIX}{{\textrm{VIX}}}
\newcommand{\SPX}{{\textrm{SPX}}}
\newcommand{\Scal}{{\mathcal S}}

\renewcommand{\a}{{\mathbf a}}
\renewcommand{\b}{{\mathbf b}}
\renewcommand{\u}{{\mathbf u}}
\newcommand{\vecsig}{{\textup{\bf vec}}}

\newcommand{\fdot}{{\,\cdot\,}}

\DeclareMathOperator{\argmin}{argmin}

\newtheorem{theorem}{Theorem}

\newtheorem{assumption}[theorem]{Assumption}
\newtheorem{corollary}[theorem]{Corollary}

\theoremstyle{definition}
\newtheorem{definition}[theorem]{Definition}
\newtheorem{remark}[theorem]{Remark}
\newtheorem{example}[theorem]{Example}

\newtheorem{lemma}[theorem]{Lemma}

\newtheorem{proposition}[theorem]{Proposition}

\numberwithin{equation}{section}
\numberwithin{theorem}{section}

\definecolor{brickred}{rgb}{0.8, 0.25, 0.33}
\definecolor{junglegreen}{rgb}{0.16, 0.67, 0.53}
\begin{document}

\title{Joint calibration to SPX and VIX options with signature-based models}
\author{Christa Cuchiero\thanks{Vienna University, Department of Statistics and Operations Research, Data Science Uni Vienna,
		Kolingasse 14-16 1, A-1090 Wien, Austria, christa.cuchiero@univie.ac.at}
	\and Guido Gazzani\thanks{University of Verona, Department of Economics,
		Via Cantarane 24, 37129 Verona, Italy, guido.gazzani@univr.it} \newline
	\and Janka M\"oller\thanks{Vienna University, Department of Statistics and Operations Research,
		Kolingasse 14-16 1, A-1090 Wien, Austria, janka.moeller@univie.ac.at} \newline
	\and Sara Svaluto-Ferro\thanks{University of Verona, Department of Economics,
		Via Cantarane 24, 37129 Verona, Italy, sara.svalutoferro@univr.it.\newline
		The first three authors gratefully acknowledge financial support through
		grant Y 1235 and grant I 3852 of the Austrian Science Fund. All authors acknowledge financial support through the OEAD WTZ project FR 02/2022. The present work was initiated when the second author was affiliated to the University of Vienna (ISOR)).
\newline
  { We also thank the anonymous referees and associated editor for their valuable comments and suggestions.}  
}
}

\maketitle
\begin{abstract}
We consider a stochastic volatility model where the dynamics of the volatility are described by a linear function of the (time extended) signature of a primary process which is supposed to be a polynomial diffusion.
We obtain closed form expressions for the VIX squared, exploiting the fact that the truncated signature of a polynomial diffusion is again a polynomial diffusion. Adding to such a primary process the Brownian motion driving the stock price, allows then to express both the log-price and the VIX squared as linear functions of the signature of the corresponding augmented process. This feature can then be efficiently used for pricing and calibration purposes.  Indeed, as the signature samples can be easily precomputed, the calibration task can be split into an offline sampling and a standard optimization. 
We also propose a Fourier pricing approach for both VIX and SPX options exploiting that the signature of the augmented primary process is an infinite dimensional affine process.
For both the SPX and VIX options we obtain highly accurate calibration results, showing that this model class allows to solve the joint calibration problem without adding jumps or rough volatility. 
\end{abstract}


\tableofcontents
\section{Introduction}
The joint calibration of option pricing models to SPX
 and VIX\footnote{With SPX and VIX we refer to the index tickers of the S\&P 500 and its volatility index, respectively. In the sequel we will use SPX and S\&P 500 interchangeably.} options is a problem that has gained a lot of attention in quantitative finance since several years.
One main reason for the increased interest is that the VIX index has become an important underlying for many derivatives.
In fact, futures and options written on it are extensively used to hedge the volatility exposure of option portfolios, see e.g., \cite{R:11}. 
We address the reader to the official website of CBOE\footnote{\href{https://www.cboe.com/tradable_products/vix/}{www.cboe.com/tradable$\_$products/vix/}} for details on the history of the VIX, how it is computed, traded  and used as underlying for derivatives, which emphasizes in particular the need to jointly calibrate to both, the prices of options on the  S\&P 500 index and to prices of VIX derivatives.   
In this respect the main challenge 
is to reconcile the large negative skew of SPX options' implied volatilities with relatively lower implied volatilities arising from the VIX options,  especially for short maturities (see e.g., \cite{G:20}.

Inspired by \cite{PSS:20} and \cite{CGS:23}, we consider here a  new type of stochastic volatility model for the discounted, dividend-adjusted price process $S=(S_{t})_{t\ge0}$. %with continuous trajectories. 
It is given by 
\begin{equation*}
    \d S_{t}(\ell) = S_{t}(\ell)\sigma_{t}^S(\ell) \d B_{t},
\end{equation*}
for an initial condition $S_0 \in \mathbb{R}_+$,  a standard Brownian motion $B$, and a volatility process $\sigma^S$ satisfying
\begin{equation}\label{eq:sigintro}
    \sigma_{t}^S(\ell) := \ell(\widehat{\X}_{t}),
\end{equation}
where $\ell$ is  a linear map 
of the signature $\widehat{\X}_{t}$   of a
process $\widehat{X}$.
Specifically, the main ingredient in this framework is
a
$d$-dimensional polynomial diffusion process $X=(X_{t}^{1},\dots,X_{t}^{d})_{t \geq 0}$ (see \cite{CKT:12,FL:16}), which we 
call here \emph{primary  process} and whose  augmentation with time $t$ is 
denoted by $(\widehat{X}_{t})_{t \geq 0}=(t,X_{t}^{1},\dots,X_{t}^{d})_{t \geq 0}$.
By modeling $\sigma^S$ via \eqref{eq:sigintro} we
assume that the signature of $\widehat{X}$, denoted by $\widehat{\mathbb{X}}$ (and rigorously introduced in Section~\ref{sec:sig}), serves as a linear regression basis for the volatility process, while the parameters of the linear map $\ell$ have to be learned from (option price) data.
Note that the parameters of $X$ are prespecified beforehand and can thus be seen -- in analogy to machine learning terminology -- as \emph{hyperparameters} (that of course can be optimized over some %training or
validation set). As outlined below this is one of the crucial features that allows for the split of the calibration task into precomputable samples and parameters $\ell$  to be optimized.
\\

Let us now highlight the implications of this modeling approach and the novelty of the present work.
\begin{itemize}
\item {\color{black} The current framework can be seen as \emph{universal} in a large class of
continuous (non-rough) stochastic volatility models in the following sense: for a stochastic volatility model, whose volatility is given by a continuous path-functional depending on a polynomial diffusion $(X_t)_{t\geq 0}$, it follows from Proposition~\ref {prop1} that this path-functional can be approximated by a linear function of the signature of $(\widehat{X}_t)_{t\geq 0}$. 
A concrete example for a volatility process of such a form is an Itô-diffusion with sufficiently regular coefficients, as in this case, the volatility is indeed given by a continuous path-functional of the time-augmented Brownian motion driving the diffusion.}

\item {\color{black} Additionally, our model }truly nests several classical models (see Remark~\ref{rem:classical}) and for instance also the
  `quintic Ornstein-Uhlenbeck volatility model',  recently proposed by \cite{AJIL:22b}, which -- with an additional input curve --  is shown to fit  SPX and VIX smiles well.

    \item By choosing the parameters of $\ell$
appropriately, the modeling framework incorporates both, %purely 
Markovian (in $(S,X)$) and path-dependent models.

\item Up to our knowledge, it is  the first signature-based model that is employed for pricing and calibration of VIX options as well as joint calibration, together with SPX options.

\item We illustrate that the joint calibration problem can be solved in this framework
without jumps and rough volatility (compare also \cite{R:22, AJIL:22b,GLJ:22}).

\item By using time-varying parameters we can go beyond short maturities both for SPX and VIX options (as classically tackled in the literature) and achieve a joint calibration also for longer maturities.

\end{itemize}

In order to achieve the highly accurate calibration results, illustrated in Section~\ref{sec:calibration-vix-options} and Section~\ref{sec:joint-calib}, we exploit the following 
mathematical and numerical properties.

\begin{itemize}

\item Defining $Z:= (X, B)$, then not only $\sigma^S(\ell)$ but also the log-price $\log(S(\ell))$ can be expressed  as a linear function of the signature of $\widehat{Z}$.
 The computational benefit is immediate, since no (Euler) simulation scheme is needed to sample from the marginals of the price process.
In terms of the parameters $\ell$, $\log(S(\ell))$ is the sum of a quadratic function and a linear one, see Proposition~\ref{sig-model-spx}.

\item Since $\widehat X$ is additionally assumed to be a \emph{polynomial diffusion} (see \cite{CKT:12, FL:16}), the VIX under our model can be computed analytically via  matrix exponentials. Indeed, in this case the forward variance can be represented by a quadratic form in the parameters $\ell$ and the corresponding matrix can be computed by polynomial technology, i.e.~via matrix exponentials, see Theorem~\ref{th:VIXclosed}.  This tractability property is a consequence of  the fact that \emph{the truncated signature of a polynomial diffusion is again a polynomial diffusion} (see Section~\ref{sec:expected_signature}).

\item We can efficiently apply a Monte Carlo approach (potentially with variance reduction) for option pricing and calibration, since  the signature samples of $\widehat{Z}$  can be computed offline and therefore the simulation and optimization step can be completely separated. Indeed, due to the representations of VIX and $\log(S(\ell))$ described above, the same samples can be used for \emph{every} linear map $\ell$. Therefore, the calibration task can be split into an offline sampling and a standard optimization, as no simulation is needed during the latter.  Moreover, due to the fact that we can obtain a closed-form expression for the VIX (thanks to the polynomial technology) we can avoid a nested Monte Carlo procedure to evaluate the conditional expectation.
    
\item Alternatively, a Fourier pricing approach for both VIX and SPX options can be used. Indeed, by building on the fact that the signature of $\widehat{Z}$ is an affine process (with values in the extended tensor algebra) as proved in \cite{CST:22}, its Fourier-Laplace transform can be computed by solving an (extended tensor algebra valued) Riccati equation, which in turn can be used for Fourier pricing as outlined in Section~\ref{sec:affine-nat}. 
\end{itemize}
The remainder of the paper is organized as follows.  Section~\ref{sec:state-of-art} gives a review over the different contributions in the literature concerning the joint calibration problem. In Section~\ref{sec:sig} we 
introduce the signature in the context of continuous semimartingales, its main properties as well as notation used throughout the paper.
Section~\ref{sec:model} 
is dedicated to the introduction of our signature-based model  and the connections to classical and also recent stochastic volatility models in the literature. Section~\ref{sec:expected_signature} is then devoted to the discussion and proof of the matrix exponential formula for the (conditional) truncated expected signature of a polynomial diffusion. This result is at the core of Section~\ref{sec:vix}, where  we derive  a tractable formula for the VIX, needed for pricing  VIX options and VIX futures. 
Building on these formulas, our
calibration results to VIX options are presented in Section~\ref{numerical_result_vix}. 
In Section~\ref{spx_sigsde}, we then prove, similarly as for the VIX, a tractable expression for $S$. Additionally,  we exploit in Section~\ref{sec:affine-nat} the affine nature of the signature process (as proved in \cite{CST:22}),  to obtain a Fourier pricing approach within our modeling choice for both VIX and SPX options. 
We finally present the numerical results of the joint calibration problem in Section~\ref{sec:joint-calib}, both in the case of constant parameters and with time-varying parameters, where the latter are introduced in Section~\ref{TV_VIX} and Section~\ref{TV_SPX}.

The data used in Section~\ref{numerical_result_vix} and Section~\ref{numerical_result_joint} were purchased from OptionMetrics\footnote{\href{https://optionmetrics.com}{https://optionmetrics.com/}}. An implementation of the model for the joint calibration can be found in \href{https://github.com/GuidoGazzani-ai/jointcalib_sigsde}{GuidoGazzani-ai/jointcalib$\_$sigsde} or \href{https://github.com/janka-moeller/joint_calib_SPX_VIX}{janka-moeller/joint$\_$calib$\_$SPX$\_$VIX}.

\subsection{State of the art}\label{sec:state-of-art}

This section is primarily dedicated to a literature review on the joint calibration problem and secondly, to a brief overview on signature methods in finance.

First attempts to solve the joint calibration problem appear in \cite{G:08}, with a double constant elasticity of variance model (CEV), which despite being rather flexible cannot fit accurately the implied volatilities of SPX and VIX options jointly. Later on, the introduction of models with jumps in the SPX (or additionally also in the volatility) led to different contributions, for instance  
the forward variance model of \cite{CK:13} described as an exponential of an affine process with L\'evy jumps,
the regime-switching enhancement of the classical Heston model by \cite{PS:14}, the 3/2 model with jumps in the asset price of \cite{BB:14}, in the volatility (\cite{KS:15}), or with co-jumps and idiosyncratic jumps in the volatility (\cite{PPR:18}).

Continuous  stochastic volatility models based on Markovian semimartingales have also been employed to solve the joint calibration problem. For instance, in \cite{FS:18} a Heston model with stochastic vol-of-vol has been calibrated, however only for maturities above 4 months where VIX options are less liquid. More recently, \cite{R:22} considered a model where the volatility is driven by two Ornstein-Uhlenbeck (OU) processes using a non-standard transformation function. This choice of two OU-processes  has been an inspiration for our concrete numerical implementations. 
We also point out that the (non-rough) model introduced in \cite{AJIL:22,AJIL:22b}, where  the volatility is described by a polynomial of order five in one single OU-process, falls (apart from the additional input curve) into this class of continuous Markovian models and is  a particular instance of our framework. 
 Let us also refer to the paper by \cite{GM:22}, where a neural SDE model 
has been successfully jointly calibrated.
Within the class of continuous, however not necessarily Markovian models,  \cite{GLJ:22} conduct an empirical and statistical analysis as well as a joint calibration for a  family of   models where the volatility depends on the paths of the asset. These models can be turned into Markovian ones by using exponential kernels instead of general ones, see also \cite{GG:23} for their joint calibration.

Two further distinct lines of research are worth being mentioned as well: first, martingale optimal transport and second rough volatility.

The martingale optimal transport approach is used to calibrate discrete-time models as proposed in \cite{G:20a,G:21}. These models are closely related to Schr\"odinger bridge problems, where the idea is to calibrate  only the drift of the volatility while keeping the volatility of volatility
unchanged, see e.g. \cite{GLOW:20} as well as the references therein regarding an optimal transport approach. Although the calibration within that setting is accurate, it is also computationally rather expensive and not amenable to calibrate to several maturities jointly. These computational challenges have been tackled recently in \cite{GB:22}.

In the area of rough volatility modeling, initiated by the seminal paper of \cite{GJR:18}, the main idea is to replace the standard Brownian motion in the volatility process  by a fractional Brownian motion. Even though the roughness of the trajectories found in \cite{GJR:18}, can also be related to the estimation procedure as discussed e.g.~in \cite{CD:23}, the non-Markovianity given by the fractional Brownian motion with Hurst parameter $H<0.5$, is well-suited to reproduce certain stylized facts arising in financial data, e.g. volatility persistence or multiple scales of mean reversion; see \cite{BFG:16}. Several classical models have been enhanced with rougher noise, but for simplicity we here only mention those employed in the SPX/VIX calibration. One example is the quadratic rough Heston model introduced in \cite{GJR:20}, which was in turn calibrated in \cite{RZ:21} by relying on neural networks approaches, also exploited in e.g.~\cite{BHMST:19}.
In \cite{R:22} an exhaustive study of the flexibility of different rough and non-rough volatility models for the joint SPX/VIX calibration is carried out, including the rough Bergomi and the rough Heston model.  Some of these, for instance the rough Heston model, have an affine structure i.e., can be embedded in the class of affine Volterra processes.
In particular they allow for Fourier pricing after solving the related fractional Riccati equations. This underlying structure is the building block of an extension with jumps investigated in \cite{BLP:22,BPS:22}.
{  We refer additionally to \cite{DKMY:23,GG:23} for a very recent literature review on volatility modeling.}

Concerning our framework, signature-based methods  provide a generic non-parametric way to extract characteristic features (linearly) and path-dependency from data, which is essential in (machine) learning and calibration tasks in finance. This explains why these techniques become more and more popular in  mathematical finance, see e.g., \cite{BHLPW:20, KLP:20, PSS:20, LNP:20, NSSBWL:21, BHRS:21,  MH:21, CPS:22,CM:22,AGTZ:23, NCL:23,WMK:23,C:23,LBW:23} and the references therein.

\section{Signature: definition and properties}\label{sec:sig}

We start by introducing basic notions related to the definition of the  signature of an $\mathbb{R}^{d}$-valued continuous semimartingale. This is similar as in \cite{CGS:23} or \cite{BHRS:21}, but to keep the paper self-contained we recall the essential definitions and properties.

For each $n \in \mathbb{N}_0$ we define recursively the $n$-fold tensor product of $\mathbb{R}^{d}$,
\begin{equation*}
(\mathbb{R}^{d})^{\otimes 0}:=\mathbb{R}, \qquad (\mathbb{R}^{d})^{\otimes n}:=\underbrace{\mathbb{R}^{d}\otimes\cdots\otimes\mathbb{R}^{d}}_{n}. 
\end{equation*}
	For $d\in \N$, we define the extended tensor algebra on $\mathbb{R}^{d}$ as 
	\begin{equation*}
		T((\mathbb{R}^{d})):=\{\textbf{a}:=(a_{0},\dots,a_{n},\dots) : a_{n}\in(\mathbb{R}^{d})^{\otimes n}\}.
	\end{equation*}
	Similarly we introduce the truncated tensor algebra of order $n \in \mathbb{N}$ 
		\begin{equation*}
		T^{(n)}(\mathbb{R}^{d}):=\{\textbf{a}\in T((\mathbb{R}^{d})) : a_{m}=0, \forall m>n\},
	\end{equation*}
	and the tensor algebra
$
		T(\mathbb{R}^{d}):=\bigcup_{n\in \N}T^{(n)}(\mathbb{R}^{d}).
$
Note that $T^{(n)}(\R^{d})$ has dimension \begin{equation}\label{eqn10}
    d_{n}:=(d^{n+1}-1)/(d-1).
\end{equation}

	For each $\textbf{a},\textbf{b}\in T((\mathbb{R}^{d}))$ and  $\lambda\in\R$ we set 
	\begin{align*}
		\textbf{a}+\textbf{b}&:=(a_{0}+b_{0},\dots,a_{n}+b_{n},\dots),\\
		\lambda\cdot \textbf{a}&:= (\lambda a_{0},\dots, \lambda a_{n},\dots),\\
		\textbf{a}\otimes \textbf{b}&:=(c_{0},\dots, c_{n},\dots),
	\end{align*}
where $c_{n}:=\sum_{k=0}^{n}a_{k}\otimes b_{n-k}$. Observe that $(T((\mathbb{R}^{d})),+,\cdot,\otimes)$ is a real non-commutative algebra.

	 For a multi-index $I:=(i_1,\ldots,i_n)$  we set $|I|:=n$. We also consider the empty index $I:=\emptyset$ and set $|I|:=0$. If $n\geq 1$ or $n\geq 2$ we set $I':=(i_1,\ldots,i_{n-1})$, and $I'':=(i_1,\ldots,i_{n-2})$, respectively. We also use the notation
	$$\{I\colon|I|=n\}:=\{1,\ldots,d\}^n,$$
	omitting the parameter $d$ whenever this does not introduce ambiguity. Observe that multi-indices can be identified with words, as it is done for instance in \cite{LNP:20}. 

Next, for each $|I|\geq 1$ we set
	$$e_I:=e_{i_1}\otimes\cdots\otimes e_{i_n}.$$
	Observe that the set $\{e_I\colon |I|=n\}$ is an orthonormal basis of $(\mathbb{R}^{d})^{\otimes n}$.
	Denoting by $e_\emptyset$ the basis element corresponding to $(\R^d)^{\otimes 0}$, each element of $\textbf{a}\in T((\R^d))$ can thus be written as
	$$\textbf{a}=\sum_{|I|\geq 0}\a_I e_I,$$
	for some $\a_I\in \R$. Note that if $a_n\in(\R^d)^{\otimes n}$ we use non-bold notation whereas for the components $\a_I\in \R$ we write them bold.
	Finally, for each $\textbf{a}\in T(\R^d)$ and each $\textbf{b}\in T((\R^d))$ we set
		$$\langle \textbf{a},\textbf{b}\rangle:=\sum_{|I|\geq 0}\langle \a_{I},\b_{I}\rangle.$$
Observe in particular that $\b_I=\langle e_I,\textbf{b}\rangle$.

In the present work it will be useful to enumerate the elements of the truncated tensor algebra. To this extent we introduce the isomorphism $\vecsig: T^{(n)}(\R^d)\to \R^{d_{n}}$ and an injective labeling function $\mathscr{L}:\{I: |I|\le n\}\longrightarrow \{1,\dots, d_{n}\}$, such that
	\begin{equation}\label{eqn8}
	    \vecsig(\u):=\sum_{|I|\le n}e_{\mathscr{L}(I)}\u_{I},
	\end{equation}
 where $d_n$ is as in \eqref{eqn10}.
 
Throughout the paper we fix a filtered probability space $(\Omega, \Fcal, (\Fcal_t)_{t\geq0},\Q)$ on which we consider the stochastic processes to be defined. We are now ready to introduce the signature of an $\mathbb{R}^d$-valued continuous semimartingale.

\begin{definition}
	Let $X$ be a continuous $\mathbb{R}^{d}$-valued semimartingale with $d\ge1$. The \emph{signature of $X$} is the $T((\R^d))$-valued process  $(s,t)\mapsto \X_{s,t}$ whose components are recursively defined as
\begin{equation*}
	\langle e_{\emptyset},\mathbb{X}_{s,t}\rangle:=\textup{1}, \qquad \langle e_{I}, \mathbb{X}_{s,t}\rangle:=\int_{s}^{t}\langle e_{I'},\mathbb{X}_{s,r}\rangle\circ \mathrm{d} X_{r}^{i_{n}},
\end{equation*}
for each $I=(i_1,\ldots, i_n)$ , $I'=(i_1,\ldots, i_{n-1})$ and $0\leq s\leq t$, where $\circ$ denotes the Stratonovich integral. 
Its projection $\X^n$ on $T^{(n)}(\mathbb{R}^{d})$ is given by
\begin{equation*}
	\mathbb{X}_{s,t}^{n}=\sum_{|I|\leq n} \langle e_{I}, \mathbb{X}_{s,t}\rangle e_{I}
\end{equation*}
and is called \emph{signature of $X$ truncated at level $n$}. If $s=0$, we use the notation $\X_t$ and $\X_t^n$, respectively.
\end{definition}
Observe that the signature of $X$ and the signature of $X-c$ coincide for each $c\in \R$. Moreover, with an equivalent notation we can write
\begin{align*}
    \X_t&=\bigg(1,\int_0^t1\circ \d X_s^1,\ldots, \int_0^t1\circ \d X_s^d,\int_0^t\bigg(\int_0^s 1 \circ \d X_r^1\bigg)\circ \d X_s^1,\\
&\qquad\int_0^t\bigg(\int_0^s 1 \circ \d X_r^1\bigg)\circ \d X_s^2,
\ldots,
\int_0^t\bigg(\int_0^s 1 \circ \d X_r^d\bigg)\circ \d X_s^d,\ldots\bigg).
\end{align*}

A well-known and extremely useful property of the signature is that every polynomial function in the signature has a linear representation. For the precise statement we first need to introduce the following concept  (see also Definition 2.4 in \cite{LNP:20} or Section~2.2. in \cite{BHRS:21}). 
\begin{definition}\label{shuffle-product}
	For every two multi-indices $I$ and $J$ the \emph{shuffle product} is defined recursively as
	\begin{align*}
		e_{I}\shuffle e_{J}:= (e_{I'}\shuffle e_{J})\otimes e_{i_{|I|}}+(e_{I}\shuffle e_{J'})\otimes e_{j_{|J|}},
	\end{align*}
	with $e_{I}\shuffle e_{\emptyset}:= e_{\emptyset}\shuffle e_{I}= e_{I}$. It extends to $\textbf{a},\textbf{b}\in T(\R^d)$ as
	$$\textbf{a}\shuffle\textbf{b}=\sum_{|I|,|J|\geq0}\a_I\b_J (e_I\shuffle e_J).$$
\end{definition}

Observe that $(T(\R^d),+,\shuffle)$ is a commutative algebra, which in particular means that the shuffle product is associative and commutative. 

In the following proposition we summarize some useful properties of the signature. 
These results have been developed in the rough paths literature (see for instance \cite{R:58} or \cite{LCL:07} for the shuffle property, \cite{BL:2016} for the uniqueness of the signature, and \cite{C:57, C:77} for Chen's identity) and have then been refined in the context of semimartingales (see e.g., \cite{BHRS:21,CM:22}).
For a more detailed exposition and proofs we refer to \cite{CGS:23}.

\begin{proposition}\label{prop1}
Let $X$ and $Y$ be two continuous $\R^{d}$-valued semimartingales with $X_0=Y_0=0$. Then the following properties hold.
\paragraph {Shuffle property} For each two multi-indices $I,J$ and each $0\leq s\leq t$ it holds
	\begin{equation}\label{eqn5}
		\langle e_{I},\mathbb{X}_{s,t}\rangle \langle e_{J}, \mathbb{X}_{s,t}\rangle =\langle e_{I}\shuffle e_{J}, \mathbb{X}_{s,t}\rangle.
	\end{equation}
 	\paragraph{Uniqueness of the signature} Set $\widehat X_t:=(t,X_t)$, $\widehat Y_t:=(t,Y_t)$ and let $\widehat \X$ and $\widehat \Y$ be the corresponding signature processes. Then the signature $\widehat\X_T=\widehat\Y_T$  if and only if $X_t=Y_t$ for each $t\in[0,T]$. 

     \paragraph{Chen's identity} For each $0\leq s\leq u\leq t$ it holds
    \begin{equation}\label{eqn7}
        \X_{s,t}=\X_{s,u}\otimes\X_{u,t}.
    \end{equation}
     This can equivalently be written as
    \begin{equation}\label{eqcen}
        \langle e_I,\X_{s,t}\rangle=\sum_{e_{I_1}\otimes e_{I_2}=e_I}\langle e_{I_1},\X_{s,u}\rangle\langle e_{I_2},\X_{u,t}\rangle,
    \end{equation}
    for each multi-index $I$.

    \paragraph{Universal approximation theorem}  For each $n\in \N$ consider the sets
$$\Scal^{(n)}:=\{(\widehat \X_t^n)_{t \in [0,T]}(\omega)\colon \omega\in \Omega\}$$
 and let $S^{(n)}:\Scal^{(2)}\to \Scal^{(n)}$ denote the corresponding Lyons lift. Then it holds that $S^{(n)}((\widehat \X^2_t)_{t\in [0,T]})=(\widehat \X^n_t)_{t\in [0,T]}$ almost surely.
Consider then a generic distance $d_{\Scal^{(2)}}$ on the set of trajectories given by $\Scal^{(2)}$,
 with respect to which the map from $\Scal^{(2)}$ to $\R$ given by
$$\hat{\textbf x}^2\mapsto\langle e_I,S^{(|I|)}(\hat{\textbf x}^2)_t\rangle$$
is continuous for each multi-index $I$ and every $t \in [0,T]$. 
Let $K$ be a compact subset of $\Scal^{(2)}$ and  consider a continuous map $f: K\to \R$. Then for every $\varepsilon>0$ there exists some $\ell\in T(\R^{d})$ such that
	\begin{equation*}
	\sup_{(\widehat\X_t^2)_{t \in [0,T]}\in K}	\lvert f((\widehat \X^2_t)_{t \in [0,T]})-\langle \ell, \widehat \X_{T}\rangle\lvert <\varepsilon,
	\end{equation*}
	almost surely.
\end{proposition}

\section{The model}\label{sec:model}
We start by introducing the concept of polynomial diffusions   (see \cite{CKT:12,FL:16})  which will play a key role for the computation of the conditional expected signature. Here we denote by $\sqrt{\fdot}$ the matrix square root.

\begin{definition}\label{def1}
Suppose that an $\R^d$-valued process  $X=(X_{t})_{t\ge0}$  is a weak solution of
$$
    \d X_{t} = b(X_t)\d t+ \sqrt{ a(X_t)}\d W_{t},\qquad X_0=x_0
$$
for some $d$-dimensional Brownian motion $W$ and some maps $a:\R^d\to \S^d_+$ and $b:\R^d\to\R^d$ such that
$a_{ij}$ is a polynomial of degree at most 2 and $b_j$ is a polynomial of degree at most 1 for each $i,j\in\{1,\ldots,d\}$. Then we call $X$ \emph{polynomial diffusion}.
\end{definition}
We are now ready to introduce the model $(S_{t})_{t\geq0}$ for the discounted, dividend-adjusted dynamics of the S\&P 500 index already outlined in the introduction. Its dynamics  under a risk-neutral probability measure $\Q$ are given by 
\begin{equation}\label{model}
    \mathrm{d}S_{t}=S_{t}\sigma_{t}^{S}\mathrm{d}B_{t},
\end{equation}
where $S_{0}\in\R^{+}$, $\sigma^{S}=(\sigma_{t}^{S})_{t\geq0}$ is the volatility process to be specified and $B=(B_{t})_{t\geq0}$ is a one-dimensional Brownian motion, correlated with $\sigma^S$. We define additionally the instantaneous variance via $V_t:=(\sigma_t^S)^2$ for every $t\geq0$. Our modeling choice  is to parametrize the volatility process $\sigma^S$ as a linear function of the time-extended signature of a primary process $X$, namely 
\begin{equation}\label{sig-vol2}
	\sigma_{t}^{S}(\ell):=\ell_{\emptyset}+\sum_{0<\lvert I \lvert \le n } \ell_{I} \langle e_{I},\widehat{\mathbb{X}}_{t}\rangle,
\end{equation}
where
\begin{itemize}
 \item $(X_t)_{t\geq0}$ and thus also
 $\widehat X=(t,X_{t})_{t\ge0}$ is a polynomial diffusion (with values in $\mathbb{R}^d$ and  $\mathbb{R}^{d+1}$ respectively) in the sense of Definition~\ref{def1}.
    \item $\ell:=\{\ell_{I}\in\mathbb{R}: |I|\le n\}$ denotes the collection of parameters of the model, i.e., $\ell\in \mathbb{R}^{(d+1)_{n}}$.
\end{itemize}

We then denote by $\rho$ the correlation matrix process between the components of $X$, i.e.
    \begin{equation*}
        \rho_{ij}=\frac{[X^{i},X^{j}]}{\sqrt{[X^{i}]}\sqrt{[X^{j}]}}\in [-1,1],
    \end{equation*}
    for all $i,j=1,\dots,d$, where $\lbrack\fdot,\fdot\rbrack$ denotes the quadratic covariation.

In order to simplify the notation   we will drop the dependence on $\ell$ for the  processes $S=(S_{t})_{t\geq0}$ and $(\sigma_{t}^{S})_{t\geq0}$ as in \eqref{model}, whenever this does not cause any confusion.

\begin{remark}\label{rem:altmod1}
As an alternative definition for the volatility process $(\sigma_{t}^{S})_{t\geq0}$ one can set
$$\sigma_{t}^{S}(\ell):=\ell_{\emptyset}+\sum_{0<\lvert I \lvert \le n } \ell_{I} \langle e_{I},\widehat{\mathbb{X}}_{t-\e,t}\rangle,$$
for some fixed $\e>0$.  In this case the value of the volatility process $\sigma^S$ at time $t$ does not depend on the whole trajectory of the primary process $X$, but just on its evolution from $t-\e$ to $t$. 
For an economically reasonable choice for $\e$ the lags used in Section~3.4 of \cite{GLJ:22} can be adapted to the current setting. 
\end{remark}

 \begin{remark}[Interest rates and dividends] \label{rem:interests}
In the model given by~\eqref{model} we describe the discounted, dividend-adjusted prices and construct the VIX from them, in line with the definition of the CBOE for the computation of the VIX. However, contingent claims are often expressed in terms of undiscounted, unadjusted prices. If the dynamics of the discounted, dividend-adjusted price process are given by \eqref{model}, the undiscounted, unadjusted one is denoted by $\tilde{S}$ and fulfills
$$ \d \tilde{S}_t = (r-q)\tilde{S}_t \d t + \tilde{S}_t \sigma_t^S(\ell) \d B_t,$$
where here $r,q \in \R$ denote the interest rate and the dividend, respectively.
Therefore $\tilde{S}_t(\ell) = e^{(r-q)t}S_t (\ell)$ and the price of a call option on the S\&P 500 index under our model, reads
\begin{equation*}
     C(T,K) = \mathbb{E}\left[e^{-rT}(\tilde{S}_T(\ell)- K)^+\right]=\E[e^{-rT} (e^{(r-q)T}S_T(\ell)- K)^+]
\end{equation*}
where $T>0$ denotes the maturity time  and $K\in\R$ the undiscounted strike price. 
\end{remark}

It is worth mentioning that the pool of eligible primary processes is rather wide, including for example correlated Brownian motions, geometric Brownian motions, OU processes, Cox-Ingersoll-Ross (CIR) processes, Jacobi processes, and all continuous affine processes. 

The reason why we require the primary process  to be a polynomial diffusion  is due to the tractability properties of the truncated signature  $\widehat{\X}^{n}$ under this assumption.
We will indeed see in Section~\ref{sec:expected_signature} that in this case the (conditional) truncated expected signature of $\widehat{X}$  can be computed by solving a finite-dimensional ODE, i.e., can be written in terms of a  matrix exponential.

\begin{remark}\label{rem:classical}
We illustrate here that several classical and also recently considered stochastic volatility models are nested within our modeling choice \eqref{sig-vol2}.
    \begin{itemize}
        \item Suppose that $(X_t)_{t\geq0}$ is a $1$-dimensional OU process and let the order of the signature be $n=1$, with $\ell_{\emptyset}=\ell_{(0)}=0$  and $\ell_{(1)}\neq0$. Then the process $S=(S_{t})_{t\ge0}$ coincides with the Stein-Stein model, as introduced in \cite{SS:91}.
      \item Suppose that $(X_t)_{t\geq0}$ is a $1$-dimensional geometric Brownian motion without drift and let the order of the signature be $n=1$, with $\ell_{\emptyset}=\ell_{(0)}=0$ and $\ell_{(1)}\neq0$. Then the process $S=(S_{t})_{t\ge0}$ coincides with the SABR model, as introduced in initially in \cite{HKLW:02} with $\beta=1$.
    \item Suppose that $(X_t)_{t\geq0}$ is a $1$-dimensional OU process and let the order of the signature be $n=5$, with $\ell_{\emptyset},\ell_{(1)},\ell_{(1,1,1)},\ell_{(1,1,1,1,1)}$ non-zero and $\ell_{I}=0$ otherwise. Then the process $S=(S_{t})_{t\ge0}$ coincides with the model considered in \cite{AJIL:22, AJIL:22b} with an exponential kernel (a part from the deterministic input curve considered there additionally).  Going beyond the assumption of $\hat{X}$ being a polynomial diffusion we may allow for $(X_t)_{t\geq0}$ to be a one-dimensional fractional Brownian motion, thus leaving the semimartingale setting. And if we do not consider the time augmentation, we can also include fractional kernels and therefore the whole class of Gaussian polynomial volatility models introduced in \cite{AJIL:22} within our framework.
    \end{itemize}
\end{remark}

\begin{remark}
As indicated in the last point of the previous remark, our framework can be extended beyond the semimartingale case as long as the trajectories of the corresponding process can be enhanced to be almost surely a weakly geometric $p$-rough path. This holds for instance true for the case of time-augmented multidimensional fractional Brownian motion when $H\in (1/4,1)$, since for any $p\in (1/H,4)$ there exists an almost surely weakly geometric $p$-rough path, such that the projection on the first component coincides with the process' increments. For this result we refer to \cite{CQ:02}, Theorem 2. Observe that the case considered in \cite{AJIL:22} is simpler since it is a one dimensional setting, meaning that the corresponding signature boils down to Taylor polynomials of fractional Brownian motion.

Note however, while our framework can be extended beyond the semimartingale case as long as signatures can be defined, our methodology to compute conditional truncated expected signatures via finite dimensional matrix exponentials only works in the polynomial diffusion setting. The same applies to the linear representation of the log-price provided in Section~\ref{spx_sigsde}.
\end{remark}

\begin{remark}
    Let $X$ be a 1-dimensional OU-process, such that without loss of generality $X_0=0$, i.e.,
\begin{equation*}
    \mathrm{d}X_{t}=\kappa(\theta-X_t)\mathrm{d}t+\sigma\mathrm{d}W_t.
\end{equation*}
Then, for $n=2$ the instantaneous dynamics of the volatility process are given by
\begin{align*}
    \mathrm{d}\sigma_{t}=& \ell_0\mathrm{d}t+\ell_1 \mathrm{d}X_t+\ell_{00}t\mathrm{d}t+\ell_{01}t\mathrm{d}X_t+\ell_{10}X_t\mathrm{d}t+\ell_{11}\mathrm{d}X_{t}^2\\
    =&(\ell_0+\ell_{00}t+\ell_{1}\kappa(\theta-X_{t})+\ell_{01}t\kappa(\theta -X_t)+\ell_{10}X_t+2\ell_{11}X_{t}(\kappa(\theta-X_{t})+\sigma^2))\mathrm{d}t\\
    & +(\ell_{1}\sigma+\ell_{01}t\sigma+2\ell_{11}\sigma X_{t})\mathrm{d}W_t,
\end{align*}
which can be rewritten as
\begin{equation*}
    \mathrm{d}\sigma_t=(f_{1}(t)+cX_{t}(g(t)-X_{t}))\mathrm{d}t+(f_{2}(t)+\tilde{c}X_{t})\mathrm{d}W_t,
\end{equation*}
where  $f_{1},f_{2},g$ are affine functions of time and $c,\tilde{c}\in\mathbb{R}$,  all depending on the model parameters $\{\ell_{I}, |I| \le n
\}$.
The previous simple derivation implies:
\begin{itemize}
    \item If $n=1$ the instantaneous vol of vol is constant and given by $|\ell_{1}\sigma|$.
    \item If $n\ge2$ the instantaneous vol of vol is stochastic, depending explicitly on $X_t$.
    \item For $n=2$, the instantaneous volatility exhibits a stochastic  mean reversion rate given by the term $cX_{t}$, with a time-dependent long-run mean by the affine function $g(t)$. We will see in the subsequent sections that this type of model with a 3-dimensional OU-process is flexible enough to solve the joint calibration problem.
    \item  Notice that even for $n=2$, the choice $X=W$, i.e.~choosing just a Brownian motion (as for instance in \cite{PSS:20,CGS:23} for the price process), would lead to restrictive dynamics of the instantaneous volatility.
\end{itemize}
\end{remark}

\section{Expected signature of polynomial diffusion}\label{sec:expected_signature}

Let $(Y_t)_{t\geq0}$ be a polynomial diffusion in sense of Definition~\ref{def1} whose dynamics are given by
\begin{equation}\label{eqn2}
      \d Y_{t} = b(Y_t)\d t+ \sigma(Y_t)\d W_{t},\qquad Y_0=y_0,  
\end{equation}
where $\sigma(Y_t)$ denotes the matrix square root of $a(Y_t)$. Recall that in this case
 the components of $a:\R^d\to \S^d_+$ are polynomials of degree at most 2,
the components of $b:\R^d\to\R^d$ are polynomials of degree at most 1, and
 $W=(W_{t})_{t\ge0}$ is a $d$-dimensional Brownian motion. Denote then by $\Y$  the corresponding signature.

 We now explain how to employ the polynomial technology to compute the conditional truncated expected signature of $(Y_t)_{t\geq0}$. The corresponding code is available at \href{https://github.com/sarasvaluto/AffPolySig}{sarasvaluto/AffPolySig.} Several representations of related quantities in particular for the Brownian case can be found in the literature, see for instance \cite{FW:03}, \cite{LV:04}, \cite{LN:15}, \cite{BDZN:21}, \cite{RFC:22}. Our approach follows \cite{CST:22} and is based on the classical theory of polynomial processes (see \cite{CKT:12} and \cite{FL:16}). Even though results for the corresponding infinite dimensional stochastic processes (see for instance \cite{CS:21, CGDS:21}) are needed in the case of general signature SDEs considered in \cite{CST:22}, the polynomial property of $(Y_t)_{t\geq0}$ here permits to stay in  the finite dimensional setting.

\begin{lemma}\label{lem:lem2}
Let $(Y_t)_{t\geq 0}$ be the polynomial diffusion given by \eqref{eqn2} and $b$ and $a$ be the corresponding drift and diffusion coefficients.
Then
$$b_j(y)=b_j^c+\sum_{k=1}^db_j^{k}y_k
\qquad\text{ and }\qquad a_{ij}(y)=a_{ij}^c+\sum_{k=1}^da_{ij}^ky_k+\sum_{k,h=1}^da_{ij}^{kh}y_ky_h,$$
for some $b_j^c$,$b_j^{k}$,$a_{ij}^c$, $a_{ij}^k$, $a_{ij}^{kh}=a_{ij}^{hk}\in \R$.
Moreover,
$b_j(Y_t)=\langle \b_j, \Y_t^{1}\rangle
$ and $a_{ij}(Y_t)=\langle \a_{ij}, \Y_t^{2}\rangle
$
for
\begin{align*}
\b_j
&=\Big(b_j^c+\sum_{k=1}^db_j^{k}Y_0^k\Big)e_\emptyset
+\sum_{k=1}^db_j^{k}e_k\qquad\text{and}\\
\a_{ij}&=\Big(a_{ij}^c+\sum_{k=1}^da_{ij}^kY_0^k+\sum_{k,h=1}^da_{ij}^{kh}Y_0^kY_0^h\Big)e_\emptyset
+\sum_{k=1}^d\Big(a_{ij}^k+2\sum_{h=1}^da_{ij}^{kh}Y_0^h\Big)e_k
+\sum_{k,h=1}^d a_{ij}^{kh}e_k\shuffle e_h.
\end{align*}
Observe that the upper index on $Y_0^k$ and $Y_0^h$ refers to $Y$'s components and not to powers.
\end{lemma}
\begin{proof}
The first part follows by the observation that by definition of polynomial diffusion $b$ and $a$ are polynomials of degree at most 1 and 2, respectively. For the second part it then suffices to note that
$\langle e_\emptyset,\Y_t^1\rangle =\langle e_\emptyset,\Y_t^2\rangle=1,\ 
\langle e_k,\Y_t^1\rangle=
\langle e_k,\Y_t^2\rangle=(Y_t^k-Y_0^k),$ and $
\langle e_k\shuffle e_h,\Y_t^2\rangle=(Y_t^k-Y_0^k)(Y_t^h-Y_0^h).
$
\end{proof}

\begin{lemma}\label{lem1}
Let $(Y_t)_{t\geq 0}$ be the polynomial diffusion given by \eqref{eqn2} and let $\b$ and $\a$ as in Lemma~\ref{lem:lem2}.
The truncated signature $( \Y_t^n)_{t\geq0}$ is a polynomial diffusion in the sense of Definition~\ref{def1} and for each $|I|\leq n$ it holds that
$$\langle e_I, \Y_t^n\rangle
=\int_0^t\langle Le_I, \Y_s^n\rangle\d s
+ \int_0^t \langle e_{I'}, \Y_s^n\rangle \sigma_{i_{|I|}}(Y_s) \d W_s,
$$
where the operator $L:T((\R^d))\to T((\R^d))$ satisfies $L(T^{(n)}(\R^d))\subseteq T^{(n)}(\R^d)$ and is given by 
\begin{equation}\label{eq:opL}
Le_{I}= e_{I'}\shuffle \b_{i_{|I|}}+\frac{1}{2}e_{I''}\shuffle \a_{i_{|I|-1}i_{|I|}}.
\end{equation}
\end{lemma}
\begin{proof}
Let $\sigma_j(Y_t)$ denote the $j$-th row of $\sigma(Y_t)$.
By  definition of the signature,  Stratonovich integral and by the shuffle property we can compute
\begin{align*}
	\langle e_{I},  {\Y}_t \rangle
	&=\int_0^t\langle e_{I'},  {\Y}_s \rangle\circ {\mathrm{d}}\langle e_{i_{|I|}},  Y_s\rangle\\
 &= { \int_0^t \langle e_{I'}, \Y_s \rangle \mathrm{d}\langle e_{i_{|I|}}, {Y}_s \rangle + \frac{1}{2} \left[ \langle e_{I'}, \Y \rangle, \langle e_{i_{|I|}}, {Y} \rangle\right]_t}\\
	&=\int_0^t\langle e_{I'},  {\Y}_s \rangle {\d}\langle e_{i_{|I|}},  Y_{s}\rangle+\frac{1}{2}\int_0^t\langle e_{I''},  {\Y}_s \rangle {\d}[\langle e_{i_{|I|-1}},  Y\rangle,\langle e_{i_{|I|}},  Y\rangle]_s\\
	&=\int_0^t
	\langle e_{I'}, \Y_s\rangle\langle \b_{i_{|I|}}, \Y_s\rangle \d s
	+\int_0^t \langle e_{I'}, \Y_s\rangle\sigma_{i_{|I|}}(Y_s) \d W_s\\
	&\qquad+\frac{1}{2}\int_0^t \langle e_{I''}, \Y_s\rangle\langle \a_{i_{|I|-1},i_{|I|}}, \Y_s\rangle\d s\\
	&=\int_0^t
	\langle e_{I'}\shuffle \b_{i_{|I|}}, \Y_s\rangle \d s
	+\int_0^t \langle e_{I'}, \Y_s\rangle \sigma_{i_{|I|}}(Y_s) \d W_s+\frac{1}{2}\int_0^t \langle e_{I''}\shuffle \a_{i_{|I|-1}i_{|I|}}, \Y_s\rangle\d s\\
	&=\int_0^t\langle Le_I, \Y_s\rangle\d s
+\int_0^t \langle e_{I'}, \Y_s\rangle \sigma_{i_{|I|}}(Y_s) \d W_s,
\end{align*}
for each $|I|\geq0$. 
Since $|I\shuffle J|=|I|+|J|$ it holds  that $L(T^{(n)}(\R^d))\subseteq T^{(n)}(\R^d)$. For $|I|\leq n$ we thus get that the corresponding drift's components are linear maps in $\Y^n$. Similarly, since  $\a_{ij}=\sum_{|I|,|J|\leq 1}\lambda_{ij}^{IJ}e_I\shuffle e_J$  for some $\lambda_{ij}^{IJ}\in \R$ and for $|I|\leq n$ we can compute
\begin{align*}
 \langle e_{I'}, \Y_s\rangle \sigma_{i_{|I|}}(Y_t)
\big(\langle e_{J'}, \Y_s\rangle \sigma_{i_{|J|}}(Y_t)\big)^\top
&=
\langle e_{I'}, \Y_s\rangle
\langle e_{J'}, \Y_s\rangle
\langle \a_{i_{|I|}j_{|J|}}, \Y_s\rangle\\
&=
\sum_{|H_1|, |H_2|\leq 1}\lambda_{i_{|I|}j_{|J|}}^{H_1H_2}
\langle e_{I'}\shuffle e_{H_1}, \Y_s\rangle
\langle e_{J'}\shuffle e_{H_2}, \Y_s\rangle,
\end{align*}
we also have that the components of the corresponding diffusion matrix are polynomials of degree 2 in $\Y^n$.  
Lemma 2.2 in \cite{FL:16} yields  the polynomial property.
\end{proof}

Since the linear operator $L$ maps the finite dimensional vector space $T^{(n)}(\R^d)$ to itself, it admits a matrix representation. 
\begin{definition}\label{def:matrix}
We call the operator $L$ defined in \eqref{eq:opL} \emph{dual operator corresponding to $\Y$}.
For each $|I|\leq n$ set then $\eta_{IJ}\in \R$ such that
\begin{equation*}
    Le_I=\sum_{|J|\leq n}\eta_{IJ}e_J,
\end{equation*}
and fix a labelling injective function $\Lvar:\{I\colon |I|\leq n\}\to\{1,\ldots,d_n\}$  as introduced before \eqref{eqn8}. We then call the matrix  $G\in \R^{d_n\times d_n}$ given by 
\begin{equation}\label{eqn3}
    G_{\Lvar(I)\Lvar(J)}:=\eta_{IJ},
\end{equation}
the \emph{$d_n$-dimensional matrix representative of $L$}.
\end{definition}
Observe that using the notation of \eqref{eqn8}, for each $\u\in T^{(n)}(\R^d)$  the matrix representative $G$ of $L$ satisfies
$$\vecsig(L\u)=G\vecsig(\u).$$
\begin{theorem}\label{thm:poly}
Let $(Y_t)_{t\geq 0}$ be the polynomial diffusion given by \eqref{eqn2},
 $(\Fcal_t)_{t\geq0}$ be the filtration generated by $(Y_t)_{t\geq0}$ and let $G$ be the $d_n$-dimensional matrix representative of the dual operator corresponding to $\Y$. Then for each $T,t\geq0$ and each $|I|\leq n$ it holds
$$\mathbb{E}[ \vecsig(\Y_{T+t}^n) |\mathcal{F}_{T}]= e^{tG^\top}\vecsig(\Y_T^n),$$
or equivalently,
$$\mathbb{E}[ \langle e_I,\Y_{T+t}^n\rangle |\mathcal{F}_{T}]= \sum_{|J|\leq n} (e^{tG^\top})_{\Lvar(I)\Lvar(J)}\langle e_J,\Y_T^n\rangle,$$
where $e^{(\fdot)}$ denotes the matrix exponential.
\end{theorem}
\begin{proof}
By Lemma~\ref{lem1} we know that $\vecsig(\Y^n)$ is a polynomial diffusion and Theorem~3.1 in \cite{FL:16} for polynomials of degree 1 yields the claim.
\end{proof}

\begin{example}\label{example_OU}
For the present paper a crucial role is played by the polynomial diffusion given by time, a $d$-dimensional OU process, and a Brownian motion. Specifically, we consider the process $\widehat Z_t:=(\widehat X_t,B_t)$ where $B$ is a Brownian motion and $\widehat X_t=(t,X_t)$ with
\begin{align*}
    \d X_t^j=\kappa^j(\theta^j-X_t^j) \d t+\sqrt{a(X_{t})}\d W_t,\qquad X_0=x_0,
\end{align*}
for $a_{ij}(X_{t})=\sigma^i\sigma^j\rho_{ij}$, and $W$ being a $d$-dimensional Brownian motion. We denote by $\rho_{j(d+1)}$ the correlation between $X^j$ and $B$. Setting $\kappa^{d+1}:=0$ and $\sigma^{d+1}:=1$ we can see that $\widehat Z$ satisfies \eqref{eqn2} in $d+2$ dimensions for 
$$b_j(\widehat Z_t)=1_{\{j=0\}}+\kappa^j(\theta^j-\widehat Z_t^j)1_{\{j\neq0\}}
\qquad\text{and}\qquad
a_{ij}(\widehat Z_t)=\sigma^i\sigma^j\rho_{ij}1_{\{i,j\neq0\}}.$$
The corresponding $\b$ and $\a$ are given by
  $
	\b_{j}=e_{\emptyset} (1_{\{j=0\}}+\kappa^{j}(\theta^{j}-\widehat Z_{0}^{j})1_{\{j\neq0\}})-e_{j}\kappa^{j}1_{\{j\neq0\}}$
	and 
	$
	\a_{ij}=e_{\emptyset}\sigma^i\sigma^j\rho_{ij}1_{\{i,j\neq0\}}
$  and we thus get
\begin{align*}
    Le_I&=e_{I'}(1_{\{{i_{|I|}}=0\}}+\kappa^{i_{|I|}}(\theta^{i_{|I|}}-\widehat Z_{0}^{i_{|I|}})1_{\{{i_{|I|}}\neq0\}})
-(e_{I'}\shuffle e_{i_{|I|}})\kappa^{i_{|I|}}1_{\{{i_{|I|}}\neq0\}}\\
&\qquad+\frac 1 2e_{I''}\sigma^{i_{|I|-1}}\sigma^{i_{|I|}}\rho_{{i_{|I|-1}}{i_{|I|}}}1_{\{{i_{|I|-1}},{i_{|I|}}\neq0\}}.
\end{align*}
	An application of $L$ to the first basis elements yields the following results:
	\begin{itemize}
		\item $L(e_{1})=e_{\emptyset}\kappa^{1}(\theta^{1}-X_{0}^{1})-e_{1}\kappa^{1}$;
		\item $L(e_{I}\otimes e_{0})= e_{I}\shuffle \b_{0} +  \frac{1}{2} e_{I'}\shuffle \a_{i_{|I|}0}= e_{I}$;
  
		\item $L(e_{0}\otimes e_{1}\otimes e_{2})= e_{0}\otimes e_{1}\kappa^{2}(\theta^{2}-X_{0}^{2})-(e_{0}\otimes e_{1})\shuffle e_{2}\kappa^{2}+\frac{1}{2}e_{0}\sigma^{1}\sigma^{2}\rho_{12}$.
	\end{itemize}
	Letting $(\Fcal_t)_{t\geq0}$ be the filtration generated by $(\widehat Z_t)_{t\geq 0}$ by Theorem~\ref{thm:poly} we can conclude that
\begin{equation}\label{eqn4}
    \mathbb{E}[ \vecsig(\widehat \Z_{T+t}^n) |\mathcal{F}_{T}]= e^{tG^\top}\vecsig(\widehat \Z_T^n),
\end{equation}
or equivalently,
\begin{equation}\label{eqn4_alternative}
    \mathbb{E}[\langle e_{I}, \widehat{\Z}_{T+t}^n \rangle|\mathcal{F}_{T}]=\sum_{|J|\leq n} (e^{tG^\top})_{\Lvar(I)\Lvar(J)}\langle e_J,\widehat{\Z}_T^n\rangle,
\end{equation}
where $G$ denotes the $(d+2)_n$-dimensional matrix representative of $L$.
In order to work with the VIX it will be convenient to restrict our attention to the signature components of $(\widehat \Z_t)_{t\geq0}$ not involving $B$. The following remark will be useful.
\end{example}
\begin{remark}\label{restriction_L}
	    Observe that given a subset $E\subseteq \{0,\ldots,d+1\}$, setting $\Ical_E:=\{I\colon i_j\in E\}$ it holds $L(\Ical_E)\subseteq \Ical_E$.  This in particular implies that 
	    $$Le_I=\sum_{I\in \Ical_E}\eta_{IJ} e_J$$
	    for each  $I\in \Ical_E$. Choosing $E=\{0,\ldots,d\}$, letting $\Lvar_E:\Ical_E\to\{1,\ldots,(d+1)_n\}$ be a labelling injective function, and setting
	    $ G_{\Lvar_E(I)\Lvar_E(J)}^E:=\eta_{IJ}$ we can see that \eqref{eqn4_alternative} reduces to 
	   	   \begin{equation*}
	          \mathbb{E}[\langle e_{I}, \widehat{\X}_{T+t}^n \rangle|\mathcal{F}_{T}]=\sum_{|J|\leq n}(e^{t(G^E)^\top})_{\Lvar_{E}(I)\Lvar_E(J)}\langle e_J,\widehat{\X}_T^n\rangle.
	   \end{equation*}
	    To simplify the notation we often drop the $E$ from $G^E$ whenever this does not introduce any confusion.
	\end{remark}

 \begin{remark}\label{rem:altmod2}
 Let  $(Y_t)_{t\geq0}$ be a polynomial diffusion  and let $ \Y^{-1}$ be defined via $e_\emptyset=\Y_s^{-1}\otimes \Y_s$, i.e.~$\langle e_\emptyset,\Y_s^{-1}\rangle=1$ and 
$$\sum_{e_{I_1}\otimes e_{I_2}=e_I}\langle e_{I_1},\Y_s^{-1}\rangle\langle e_{I_2},\Y_s\rangle=0,$$
for each $|I|>0$. Observe that it can be defined recursively on $|I|$ and each component of $\Y_s^{-1}$ corresponds to a linear combination of components of $\Y_s$ of the same length or shorter.

Since by Chen's identity (see \eqref{eqn7} or \eqref{eqcen}) we have $\Y_s\otimes \Y_{s,t}=\Y_{t}$, for each $s\leq u\leq t$ and $|I|\leq n$ we then get
\begin{align*}
    \E[\langle e_I,\Y_{s,t}\rangle|\Fcal_u]
&=\E[\langle e_I,\Y_s^{-1}\otimes\Y_{t}\rangle|\Fcal_u]
=\sum_{e_{I_1}\otimes e_{I_2}=e_I}\langle e_{I_1},\Y_s^{-1}\rangle\E[\langle e_{I_2},\Y_{t}\rangle|\Fcal_u]\\
&=\sum_{e_{I_1}\otimes e_{I_2}=e_I}\langle e_{I_1},\Y_s^{-1}\rangle
\vecsig(e_{I_2})^{\top }e^{(t-u) G^\top }\vecsig({\Y}_{u}^n),
\end{align*}
where $G$ denotes the $d_n$-dimensional matrix representative of the dual operator of $\Y$.
\end{remark}

\section{VIX options with signatures}\label{sec:vix}
In this section we discuss the implication on pricing VIX options under the model \eqref{model}-\eqref{sig-vol2}. %The implications of these on the log-price will be investigated in Section~\ref{spx_sigsde}.
The VIX index is a popular measure of the market's expected volatility of the S$\&$P 500, calculated and published by the Chicago Board Options Exchange (CBOE). The current VIX value quotes the expected annualized change in the S\&P 500 over the following 30 days, based on options-based theory and current options-market data. As stylized definition we consider  
\begin{equation}\label{initial_formula_vix_mkt}
    \VIX_{T}^2=\text{Price}_T\left\lbrack -\frac{2}{\Delta}\log\left(\frac{S_{T+\Delta}}{F_{T}^{T+\Delta}}\right)\right\rbrack,
\end{equation}
where $\Delta=30$ days, $F_{T}^{T+\Delta}$ denotes the price at time $T$ of the SPX future with maturity $T+\Delta$ and with $\text{Price}_{T}$ we refer to the market price at time $T$ of the log-contract, i.e.~the payoff in \eqref{initial_formula_vix_mkt}. Hence, under a given model we define the VIX,
\begin{equation}\label{initial_formula_vix}
    \VIX_{T}=\sqrt{\mathbb{E}\left\lbrack -\frac{2}{\Delta}\log\left(\frac{S_{T+\Delta}}{S_{T}}\right)\lvert \mathcal{F}_{T}\right\rbrack},
\end{equation}
where $\Delta=30$ days and $S_{T}$ denotes the price process at time $T>0$. 
 Recall that under a diffusion model,  the previous expression is equivalent to
\begin{equation}\label{initial_formula_vix2}
    \VIX_{T}=\sqrt{\frac{1}{\Delta}\mathbb{E}\left\lbrack\int_{T}^{T+\Delta}V_{t}\mathrm{d}t\lvert \mathcal{F}_{T}\right \rbrack},
\end{equation}
as long as $\mathbb{E}[\int_0^t V_s ds] < \infty$ for all $t \geq 0$, 
see e.g. \cite{N:94,G:11}. With  VIX options we here usually refer to either put or calls written on $\VIX$. In the present work we will consider without loss of generality only call options.

\subsection{Explicit formulas for the VIX}

This section is dedicated to one of the main implication of our modeling framework, namely an explicit formula for the VIX expression \eqref{initial_formula_vix} for $S$ following \eqref{model}-\eqref{sig-vol2}. In particular we show in the next theorem that the computation of the VIX squared reduces to a quadratic form in the parameters $\ell$. The entries of the corresponding positive semidefinite matrix can be computed by polynomial technology, i.e.~by matrix exponential as proved in Section~\ref{sec:expected_signature}.

\begin{theorem}\label{th:VIXclosed}
Let $S=(S_{t})_{t\ge0}$ be a price process described by 
\begin{equation*}
    \d S_{t}=S_{t}\sigma_t^{S}\d B_{t},
\end{equation*}
where $\sigma^{S}=(\sigma_{t}^S)_{t\ge0}$ denotes the volatility process, $B=(B_{t})_{t\ge0}$ a one-dimensional Brownian motion. Assume that  $\sigma^S$ satisfies \eqref{sig-vol2}. Following \eqref{eqn8}, fix an injective labeling function $\mathscr{L}:\{I: |I|\le n\}\to \{1,\dots, (d+1)_{2n+1}\}$ and  let $G$ be the $(d+1)_{(2n+1)}$- dimensional matrix representative of the dual operator corresponding to $\widehat \X$.
Then,
\begin{align}\label{eq:integrability}
\mathbb{E}\left\lbrack\int_0^t V_s \d s\right\rbrack< \infty
\end{align}
holds for every $t \geq 0$ and
\begin{equation}\label{matrix_vix}
    \VIX_{T}(\ell)=\sqrt{\frac{1}{\Delta}\ell^\top Q(T,\Delta)\ell},
\end{equation}
where  
\begin{equation}\label{matrix_Q}
    Q_{\mathscr{L}(I)\mathscr{L}(J)}(T,\Delta)=\vecsig((e_I\shuffle e_J)\otimes e_0)^{\top} ( e^{\Delta G^\top}-\operatorname{Id})\vecsig(\widehat{\X}_{T}^{2n+1}),
\end{equation}
and $\operatorname{Id}\in\R^{(d+1)_{2n+1} \times(d+1)_{2n+1}}$ denotes the identity matrix. More explicitly without the vectorisation this reads 
$$
 Q_{\mathscr{L}(I)\mathscr{L}(J)}(T,\Delta)= \sum_{e_K=(e_I\shuffle e_J)\otimes e_0}  \sum_{|H|\leq 2n+1}   (e^{\Delta G^\top}-\operatorname{Id})_{\Lvar(K)\Lvar(H)}  \langle e_H,\widehat{\X}_T\rangle.
$$
\end{theorem}

\begin{proof}
Observe that
\begin{align*}
    V_{t}(\ell)&=\biggl(\sum_{|I|\le n} \ell_{I}\langle e_{I},\widehat{\mathbb{X}}_{t}\rangle\biggr)^2=\sum_{|I|, |J|\le n}\ell_{I}\ell_{J} \langle e_{I}\shuffle e_{J},\widehat{\mathbb{X}}_{t}\rangle.
\end{align*}
Since continuous polynomial diffusions have finite moments of every degree,  \eqref{eq:integrability} is satisfied due to Lemma~\ref{lem1}.
Under \eqref{initial_formula_vix2}, the expression for $V_t(\ell)$ yields then
\begin{align*}\label{matrix_vix}
    \VIX_{T}^{2}(\ell)&=\frac{1}{\Delta}\sum_{\lvert I \lvert , |J| \le n } \ell_{I}\ell_{J}
    \E\biggl[\int_{T}^{T+\Delta} \langle e_{I}\shuffle e_{I},\widehat{\mathbb{X}}_{t}\rangle\mathrm{d}t\lvert \mathcal{F}_{T}\biggr]\\ \notag
    &=\frac{1}{\Delta}\ell^{\top}Q(T,\Delta)\ell,
\end{align*}
where for each $T>0$ the matrix $Q$ is given by
\begin{align*}
    Q_{\mathscr{L}(I)\mathscr{L}(J)}(T,\Delta):&=\mathbb{E}\biggl[\int_{T}^{T+\Delta}\langle e_{I}\shuffle e_{J},\widehat{\mathbb{X}}_{t}\rangle\mathrm{d}t\lvert \mathcal{F}_{T}\biggr]\\
&=\mathbb{E}\biggl[\int_{0}^{T+\Delta}\langle e_{I}\shuffle e_{J},\widehat{\mathbb{X}}_{t}\rangle\d t-\int_{0}^{T}\langle e_{I}\shuffle e_{J},\widehat{\mathbb{X}}_{t}\rangle \d t \lvert \mathcal{F}_{T}\biggr]\\
&=\mathbb{E}\bigl[\langle (e_{I}\shuffle e_{J})\otimes e_{0},\widehat{\mathbb{X}}_{T+\Delta}\rangle-\langle (e_{I}\shuffle e_{J})\otimes e_{0},\widehat{\mathbb{X}}_{T}\rangle  \lvert \mathcal{F}_{T}\bigr]\\
&=\mathbb{E}\bigl[\langle (e_{I}\shuffle e_{J})\otimes e_{0},\widehat{\mathbb{X}}_{T+\Delta}\rangle \lvert \mathcal{F}_{T}\bigr]-\langle (e_{I}\shuffle e_{J})\otimes e_{0},\widehat{\mathbb{X}}_{T}\rangle.
\end{align*}
By Theorem~\ref{thm:poly} we can rewrite the matrix $Q$ as
\begin{align*}
    Q_{\mathscr{L}(I)\mathscr{L}(J)}(T,\Delta)&=\vecsig((e_I\shuffle e_J)\otimes e_0)^{\top}e^{\Delta G^\top }\vecsig(\widehat{\X}_{T}^{2n+1})\\
    &\qquad-\vecsig((e_I\shuffle e_J)\otimes e_0)^{\top}\vecsig(\widehat{\X}_{T}^{2n+1})\\
    &=\vecsig((e_I\shuffle e_J)\otimes e_0)^{\top} ( e^{\Delta G^\top}-\operatorname{Id})\vecsig(\widehat{\X}_{T}^{2n+1}),
\end{align*}
and the claim follows.
\end{proof}

\begin{remark}
Consider now the model described in Remark~\ref{rem:altmod1} and set for simplicity  $\e\geq\Delta$. Then the results of Theorem~\ref{th:VIXclosed} still hold however with 
$$Q_{\Lvar(I)\Lvar(J)}(T,\Delta)=\sum_{e_{I_1}\otimes e_{I_2}=e_{I}\shuffle e_{J}}\int_T^{T+\Delta}\langle e_{I_1},\widehat\X_{t-\e}^{-1}\rangle
\vecsig(e_{I_2})^{\top }e^{(t-T) G^\top }\vecsig(\widehat{\X}_{T})\d t,$$
where $G$ denotes the $(d+1)_{2n+1}$-dimensional matrix representative of the dual operator corresponding to $\widehat \X$.
To adapt the proof we just need to note that for each $t\in [T,T+\Delta]$ Remark~\ref{rem:altmod2} yields
$$\mathbb{E}[\langle e_{I}\shuffle e_{J},\widehat{\mathbb{X}}_{t-\e,t}\rangle \lvert \mathcal{F}_{T}]
=\sum_{e_{I_1}\otimes e_{I_2}=e_{I}\shuffle e_{J}}\langle e_{I_1},\widehat\X_{t-\e}^{-1}\rangle
\vecsig(e_{I_2})^{\top }e^{(t-T) G^\top }\vecsig(\widehat{\X}_{T}).
$$
Note that since the integration's variable $t$ appears twice in this expression the time integral cannot be incorporated in the signature. 
 \end{remark}

\begin{remark}
     Observe that accounting for the scaling factor of 100, conventionally introduced by \href{https://www.cboe.com/tradable_products/vix}{CBOE}, the VIX index squared can equivalently be redefined (see e.g., \cite{RZ:21,R:22}) as
\begin{equation}\label{vix1}
	\VIX_{T}^{2}=\frac{100^2}{\Delta}    \mathbb{E}\left\lbrack\int_{T}^{T+\Delta}V_{t}\mathrm{d}t\lvert \mathcal{F}_{T}\right \rbrack,
\end{equation}
where $T,t>0$ and $\Delta=\frac{1}{12}$, i.e., approximately 30 days. 
Notice that since the expressions \eqref{initial_formula_vix2} and \eqref{vix1} differ only by a scaling factor, all the theoretical results of the present work hold true disregarding this scaling. For sake of simplicity we will always use \eqref{initial_formula_vix2}. 
We address the reader to Chapter 11 in \cite{G:11} for further details about the conventions of CBOE and its link with \eqref{initial_formula_vix}.
\end{remark}
We observe that the expression \eqref{matrix_Q} is computationally appealing as we can unpack the computation in three parts: compute the coordinate vector $\vecsig((e_I\shuffle e_J)\otimes e_0)$, which depends just on $d>0$ and $n>0$, calculate the matrix exponential of $G^\top$ which depends on the choice of the primary process $X$, and finally sample $\widehat{\X}_{T}^{2n+1}$ which is the only part that depends on the chosen maturity time $T$. {  In order to compute the matrix exponential we rely on \cite{BBC:19} who developed a Padé-insipired approximation to reduce the matrix multiplications, see also \cite{MVL:03} for further possible methods.  For  the implementation of the signature samples and its computational complexity  we refer to \cite{RG:18,KL:20}.}

\begin{remark}
    In general the computation of $G\in \R^{(d+1)_{2n+1}\times (d+1)_{2n+1}}$, even if done only once, can be costly. For this reason it can sometimes be interesting to avoid the last time integral and to consider the following equivalent expression of the matrix $Q$, for $|I|,|J|\le n$:
    \begin{equation}\label{Q_old}
            Q_{\mathscr{L}(I)\mathscr{L}(J)}(T,\Delta)=\bigg(\int_{T}^{T+\Delta}\vecsig(e_I\shuffle e_J)^\top e^{(t-T) G^\top}\d t\bigg)  \vecsig(\widehat{\X}_{T}^{2n}),
    \end{equation}
    where now $G\in \R^{(d+1)_{2n}\times (d+1)_{2n}}$ and where we use the fact that we can interchange the conditional expectation with the time-integral by dominated convergence.  As $G$ is singular, this time integral has to be computed numerically, in general. We propose here two possible methods that can be used in order to compute it efficiently.
\begin{enumerate}
    \item[(i)]\label{trapezoidal} \textbf{Approximation of the time integral}: e.g., via the trapezoidal rule also applied for $\VIX^2$ in \cite{BDM:21}. Hence if we consider the shuffled coordinates $\vecsig(e_I\shuffle e_J)$ of the exponential matrix we can use the symmetry of the shuffle to reduce the number of integrals to be solved from $((d+1)_{2n})^{2}$ to $\frac{(d+1)_{n}((d+1)_{n}+1)}{2}\cdot (d+1)_{2n}$, instead of $(d_{2n})^{2}$. Observe that for our integral the error of such an approximation is given by
    \begin{equation*}
        \textrm{Err}(N)=-\frac{\Delta^2}{12 N^2} G^\top( e^{ G^\top\Delta}-I)+\mathcal{O}(N^{-3}),
    \end{equation*}
    as $N\to+\infty$.
    As a further dimension reduction one can exploit the polynomial nature of $\widehat \X^n$ to obtain a matrix representation of its second order moments. Without entering into details, the matrix $G$ would then be the matrix corresponding to the linear operator acting on coefficients of polynomials of degree 2 in $\widehat \X^n$. Its dimension would thus be $\frac{(d+1)_{n}((d+1)_{n}+1)}{2}$.

    \item[(ii)]\label{taylor} \textbf{Approximation of the matrix exponential}: we can avoid to approximate the integral  by approximating the matrix exponential. Assuming that 
    \begin{equation}\label{spectral_radius_condition}
    \lim_{N\to+\infty} (G^{\top}\Delta)^{N}=0,
\end{equation}
this can for instance be done via its Taylor expansion:
\begin{equation*}
     \int_T^{T+\Delta}e^{tG^\top} \d t
     =\Delta\left(I+\frac{G^{\top}\Delta}{2\text{!}}+\cdots+\frac{(G^{\top}\Delta)^{N}}{N+1\text{!}}+ \mathcal{O}((G^{\top}\Delta)^{N+1})\right).
\end{equation*}
Observe that \eqref{spectral_radius_condition} holds true whenever the spectral radius, i.e., the maximal eigenvalue in absolute value, of the matrix $G^{\top}\Delta$ is less than 1 (see for instance Theorem~1.5 in \cite{QSS:10}). This requirement suggests that for numerical purposes the parameters of the primary process have to be chosen accordingly.

An interesting example is given by the case where $X$ is a $d$-dimensional correlated Brownian motion, as considered for instance in \cite{CGS:23}. In this case the process has no linear drift and the corresponding matrix $G$ is nilpotent, meaning that 
$G^{n}=0,
$ for each $n$ big enough.

In general, this Taylor approach permits to avoid a numerical integration and produces an   accurate approximation, allocating as few memory as possible.
\end{enumerate}
\end{remark}
\begin{remark}\label{remark:cholesky}
    A further step in the direction of a fast evaluation of $\VIX_T(\ell)$ can be taken by noticing that the matrix $Q$ in \eqref{matrix_Q} admits a Cholesky decomposition. Indeed since $Q$ is positive semidefinite and symmetric by the shuffle property, we know that there exists an upper triangular matrix $U_{T}\in \R^{(d+1)_{n}\times (d+1)_{n}}$, with possible zero elements on the diagonal, such that
    \begin{equation*}
        Q(T,\Delta)=U_{T} U_{T}^\top,
    \end{equation*}
    where for sake of simplicity we drop the dependence on $\Delta$ of $U_{T}$.
    Hence the evaluation of the $\VIX_{T}(\ell)$ reduces to 
    \begin{equation*}
        \VIX_{T}(\ell)=\sqrt{\frac{1}{\Delta}\ell^\top U_T U_T^\top \ell}=\frac{1}{\sqrt{\Delta}}\sqrt{(U_T^{\top}\ell)^2}=\frac{1}{\sqrt{\Delta}}\lVert U_T^\top \ell \lVert,
    \end{equation*}
    where here $\lVert\fdot\lVert$ denotes the Euclidean norm. We stress the fact that the Cholesky decomposition can be carried out offline, and the computational benefit  is immediate if several samples of the signature are considered.
\end{remark}
In the following remark we discuss a possible dimension reduction technique from which one can benefit computationally. {  Inspired by the approach of \cite{CGGOT:21a, CSBOHT:23}, we employ the Johnson-Lindenstrauss Lemma and consider a random projection of the signature.}  A first way to use this tool is the following.

\begin{remark}\label{remark:randomized_sig}
    Let $d_{<}\in\mathbb{N}$ be the dimension of the space to which we would like to project the signature of order $n>0$, such that $d_{<}\ll (d+1)_{n}$. Consider $A=(\alpha_{ij})\in \mathbb{R}^{d_{<}\times (d+1)_{n}}$, such that $\alpha_{ij}\sim \mathcal{N}(0,1/d_{<})$. 
    Then a possible way to employ the randomised signature is to parametrize the volatility process as follows,
    \begin{equation*}
        \sigma_{t}^{S}(\ell):= \tilde{\ell}^\top A\cdot\vecsig(\widehat{\X}_{t}^{n})
    \end{equation*}
    where with $\tilde{\ell}=\ell\cdot A^{\top}\in \R^{d_{<}}$ we denote the randomised parameters.
    Due to the linearity of integral and conditional expectation in \eqref{initial_formula_vix2} this modeling choice is equivalent to consider the randomised matrix $\widetilde{Q}\in \mathbb{R}^{d_{<}\times d_{<}}$ given by
    \begin{equation*}
        \widetilde{Q}_{\mathscr{L}(I)\mathscr{L}(J)}(T,\Delta):=A Q_{\mathscr{L}(I)\mathscr{L}(J)}(T,\Delta) A^{\top},
    \end{equation*}
   which leads to the following  representation of $\VIX_{T}(\ell)$:
    \begin{equation*}
        \VIX_{T}(\ell)=\sqrt{\frac{1}{\Delta}\tilde{\ell}^{\top}\tilde{Q}(T,\Delta)\tilde{\ell}}.
    \end{equation*}
    Observe that even if this procedure does not reduce the number iterated integrals to be computed  offline, it reduces the number of parameters to calibrate,  yielding in general to a faster evaluation of $\VIX_{T}(\ell)$. 
\end{remark}

\begin{remark}[Options on VIX]
Note that VIX options are written on VIX futures. The price process of a VIX future contract with maturity $T>0$, is given by 
\begin{equation}\label{eq:underlying-vix}
    F_t(T) := \mathbb{E}\left[ \VIX_T \lvert \mathcal{F}_t\right],
\end{equation}
and we write in particular $F(T):=F_0(T)$ to simplify notation. We point out that the VIX index does not pay dividends. The correct implied volatility for VIX options can then be obtained by inverting the Black-Scholes formula with interest rate $r>0$ and $e^{-r(T-t)}F_t(T)$ as underlying. When calibrating to VIX options, we stress that we additionally calibrate to VIX futures' prices, see Section~\ref{sec:calibration-vix-options}. This is important since   futures prices under the calibrated model are employed to compute its implied volatility surface. Including VIX futures in the calibration leads to a consistent model, both for VIX options and VIX futures, see e.g.~\cite{PPR:18,GLOW:20,G:20,G:21}.   Using  market prices of the VIX futures to invert the implied volatility surface could lead to inconsistencies if one would like to price further derivatives with the calibrated model.
\end{remark}

\subsection{Variance reduction for pricing VIX options}\label{sec:variance_reduction}
We here discuss variance reduction techniques (see e.g.~\cite{G:04}) that can speed up the calibration in the subsequently applied Monte Carlo approach further. The key idea is to introduce a control variate, namely an easy to evaluate random variable  $\Phi^{cv}$ such that given $T>0$ and $K>0$,
\begin{equation*}
    \mathbb{E}[\Phi^{cv}]=0, \qquad \qquad \Var \big((\VIX_{T}(\ell)-K)^{+}-\Phi^{cv}\big)< \Var \big((\VIX_{T}(\ell)-K)^{+}\big).
\end{equation*}
A well-working example of control variates used for pricing and calibrating neural SDE models can be found in  \cite{GSSSZ:20}, where $\Phi^{cv}$ is constructed from hedging strategies.

In the following we describe two possible choices of control variates, which consist of polynomials on $\VIX$ futures. We stress the fact that these can be seen as linear functions of the signature of the primary process $\widehat{X}$, hence they belong to the class of sig-payoffs, see \cite{LNP:20,PSS:20} and Section~4.2.2 in \cite{CGS:23}.
\begin{itemize}
    \item The first example is to employ the $\VIX$ squared as main ingredient, see for instance \cite{BDM:21, GG:22} for a similar choice within a rough Bergomi model for pricing $\VIX$ options. This is particularly easy to treat in our set up, as for any given maturity $T>0$ we  have
\begin{equation*}
    \E[\VIX_{T}^2(\ell)]=\frac 1 \Delta\ell^\top Q^{cv}(T,\Delta)\ell,
\end{equation*}
with $Q^{cv}(T,\Delta):=\E[Q_{\Lvar(I)\Lvar(J)}(T,\Delta)]$. By Theorem~\ref{th:VIXclosed} and Theorem~\ref{thm:poly} we indeed have
\begin{align*}
Q_{\Lvar(I)\Lvar(J)}^{cv}(T,\Delta)&=\vecsig((e_I\shuffle e_J)\otimes e_0)^{\top} (e^{\Delta G^\top }-\operatorname{Id})   \,\E[\vecsig(\widehat{\X}_{T}^{2n+1})] \\
&=\vecsig((e_I\shuffle e_J)\otimes e_0)^{\top} (e^{\Delta G^\top }-\operatorname{Id})e^{T G^\top}\vecsig(\widehat{\X}_{0}^{2n+1})\\
&=\vecsig((e_I\shuffle e_J)\otimes e_0)^{\top} (e^{(T+\Delta) G^\top }-e^{T G^\top})\vecsig(\widehat{\X}_{0}^{2n+1})
\end{align*}
where $G$ denotes the $(d+1)_{2n+1}$-dimensional matrix representative of the dual operator corresponding to $\widehat \X$ and  $\vecsig(\widehat{\X}_{0}^{2n+1})=e_\emptyset\in \R^{(d+1)_{2n+1}}$.

Observe that $Q^{cv}$ can again be computed  offline similarly to the matrix $Q$. Thus to compute the expectation of $\VIX_{T}^2(\ell)$ we only have to evaluate the previous quadratic form. To apply this now for pricing a call option with maturity $T>0$ and strike $K>0$,
we set
\begin{align*}
    \Phi^{cv}(\ell,T,K):&=c_{T,K}(\Delta\VIX_{T}^{2}(\ell)-\ell^{\top}Q^{cv}(T,\Delta)\ell),\\
    &=c_{T,K}(\ell^{\top}( Q(T,\Delta)-Q^{cv}(T,\Delta))\ell),
\end{align*}
where the constant $c_{T,K}$ maximizing the variance reduction is given by:
\begin{equation*}
    c_{T,K}^{\ast}=\frac{\text{Cov}((\VIX_{T}(\ell)-K)^{+},\ell^{\top}Q(T,\Delta)\ell)}{\text{Var}(\ell^{\top}Q(T,\Delta)\ell)}.
\end{equation*}
Notice that also in this case both $Q$ and $Q^{cv}$ satisfy the condition for applying the Cholesky decomposition, leading to a faster evaluation of the control variate as discussed in Remark~\ref{remark:cholesky}. Note that the Cholesky decomposition cannot be applied to $Q-Q^{cv}$, as this is in general an indefinite matrix.
\item As a second example we consider a generic polynomial in $\VIX^2$ as control variate by defining 
\begin{equation}\label{control-variate2}
    Y_{m}^{cv}(\ell,T,K) = \sum_{i=0}^{m} \alpha_i(T,K) (\VIX_{T}^2(\ell))^i
\end{equation}
where $\alpha_i(T,K)$ are chosen to approximate the payoff $(\VIX_T -K)^{+}$ with  strike price $K$ for some $m\geq 1$. The corresponding control-variate is then defined as $\Phi^{cv}(\ell, T, K):= c_{T,K} \left(Y_{m}^{cv}(\ell,T,K) - \mathbb{E}[Y_{m}^{cv}(\ell,T,K)]\right)$. Regarding the computational effort, let us remark the following.
\begin{enumerate}
\item $\VIX^2_T$ is computed anyway for every realisation and is hence already available, therefore the computation of $Y_{m}^{cv}(\ell,T,K)$ is not expensive.
\item It is possible to calculate $\mathbb{E}[Y_{m}^{cv}(\ell,T,K)]$ analytically relying on the moment formula, see Theorem~\ref{thm:poly}. 
\item The choice of $c_{T,K}\in\R$ is important and the optimal one, i.e., the one leading the highest variance reduction, is given by the following expression
\begin{equation*}
    c_{T,K}^{\ast}=\frac{\text{Cov}((\VIX_{T}(\ell)-K)^{+},Y_m^{cv}(\ell,T,K))}{\text{Var}(Y_m^{cv}(\ell,T,K))},
\end{equation*}
see for instance Section 4.1.1 in \cite{G:04}.
\end{enumerate}
We stress the fact that for $m=1$ the two control variates introduced coincide.
\end{itemize}

\subsection{Calibration to VIX options}\label{sec:calibration-vix-options}
In this section we focus on the calibration to VIX options only. 
Let $\mathcal{T}$ be a set of maturities and $\mathcal{K}$ a collection of strikes. Consider the model given by \eqref{model} and \eqref{sig-vol2}. 

Using Monte Carlo compute an approximation of option and futures' prices with $N_{MC}>0$ samples, i.e.
\begin{equation}\label{model_prices}
    \pi_{\VIX}^{\text{model}}(\ell,T,K)\approx\frac{e^{-rT}}{N_{MC}}\sum_{i=1}^{N_{MC}}(\VIX_{T}(\ell,\omega_{i})-K)^{+}, \qquad F_{\VIX}^{\text{model}}(\ell,T)\approx\frac{1}{N_{MC}}\sum_{i=1}^{N_{MC}}\VIX_{T}(\ell,\omega_{i}),
\end{equation}
where 
\begin{equation*}
    \VIX_{T}(\ell,\omega)=\sqrt{\frac{1}{\Delta}\ell^{\top}Q(T,\Delta)(\omega)\ell}=\frac{1}{\sqrt{\Delta}}\Vert U_{T}^{\top}(\omega)\ell\lVert.
\end{equation*}

It is crucial to note that in this framework a Monte Carlo approach is tractable since   for \emph{every} $\ell$ the same samples can be used. This means that we do not need to carry out any simulation during the optimization task. Indeed, 
 the matrix $Q$ can be simulated offline while only the products with $\ell\in \mathbb{R}^{(d+1)_{n}}$ enter in the calibration step.
 
Observe that an auxiliary  randomization can be employed in every optimisation step as discussed in Remark~\ref{remark:randomized_sig}. Moreover, if we want to use control variates to reduce the variance of the Monte Carlo estimator as described in the previous section, we would consider 
\begin{equation*}
\pi_{\VIX}^{\text{model}}(\ell,T,K)\approx\frac{e^{-rT}}{N_{VR}}\sum_{i=1}^{N_{VR}}(\VIX_{T}(\ell,\omega_{i})-K)^{+}-\Phi^{cv}(\ell, T,K)(\omega_{i}).
\end{equation*}
Due to the variance reduction the number of samples needed is $N_{VR}\ll N_{MC}$ and $\Phi^{cv}$ is as in Section~\ref{sec:variance_reduction}

  The calibration to VIX call options and the corresponding futures on $\mathcal{T}$ and $\mathcal{K}$ consists in minimizing the  functional 
\begin{align}\label{L_cv_vix}
        L_{\VIX}(\ell):=\sum_{T\in \Tcal, K\in\Kcal}&\mathcal{L}\left(\pi_{\VIX}^{\text{model}}(\ell,T,K),\pi_{\text{\VIX}}^{b,a}(T,K),\sigma_{\text{\VIX}}^{b,a}(T,K), F_{\VIX}^{\text{model}}(\ell,T),F_{\VIX}^{mkt}(T)\right),\end{align}
where $\mathcal{L}$ denotes a real-valued loss function,
$F_{\VIX}^{mkt}(T)$ the market's futures' prices and
\begin{equation*}
    \pi_{\text{\VIX}}^{b,a}(T,K):= \{\pi_{\text{\VIX}}^{mkt,b}(T,K),\pi_{\text{\VIX}}^{mkt,a}(T,K)\}, \quad 
    \sigma_{\text{\VIX}}^{b,a}(T,K):= \{\sigma_{\text{\VIX}}^{mkt,b}(T,K), \sigma_{\text{\VIX}}^{mkt,a}(T,K)\},
\end{equation*}
the market's option bid/ask prices $\pi_{\text{\VIX}}^{mkt,b}(T,K),\pi_{\text{\VIX}}^{mkt,a}(T,K)$, and  bid/ask implied volatilities $\sigma_{\text{\VIX}}^{mkt,b}(T,K), \sigma_{\text{\VIX}}^{mkt,a}(T,K)$, respectively. We will specify the choice of the function $\mathcal{L}$ in Section~\ref{numerical_result_vix} and Section~\ref{numerical_result_joint}. {  In both sections we employ the same optimizer, i.e.~BFGS with default parameters in ${\tt scipy.optimize}$.}

\begin{remark}[Initial guess search]\label{remark:initial_guess}
    Since within our model choice we are given  a quadratic function in $\ell$ to be minimized, a stochastic optimization with an initial guess is employed. In order to achieve faster convergence we consider an hyperparameter search to choose the starting parameters. The steps are outlined as follows.
    \begin{itemize}
        \item Find the magnitude of the coefficients returning Monte Carlo prices of the VIX options \emph{close} to the one observable on the market. To this extent we sample $N_{\ell}>0$ times parameters $\ell\in J_{i}=[-10^{-i},10^{-i}]^{(d+1)_{n}}$, for $i=1,\dots,m$ with $m>0$.
        \item Select $J^{\ast}\in (J_{i})_{i=1}^{m}$ such that
        \begin{equation*}
            J^{\ast}\in \argmin_{i:\ \ell\in J_{i}} L_{\VIX}(\ell).
        \end{equation*}
            \item Choose the initial guess to be
        \begin{equation*}
            \ell_{\text{initial}}\in \argmin_{ \ell \in J^{\ast}}L_{\VIX}(\ell).
        \end{equation*}
    \end{itemize}
\end{remark}

\subsubsection{Numerical results}\label{numerical_result_vix}
In the present section we report the results of the calibration to $\VIX$ options only. Here we consider call options written on the $\VIX$ on the trading day 02/06/2021, the same as in \cite{GLJ:22}. We stress that for such recent dates the bid-ask spreads for $\VIX$ options are rather tight with respect to older dated options as considered for instance in \cite{GJR:20,BPS:22}. The maturities are reported in the following table with the corresponding range of strikes (in percentage) with respect to the market's futures prices.
\vspace{0.2cm}
\begin{center}
\begin{tabular}{||c |c|c| c|c|c ||} 
 \hline
 $T_{1}=0.0383$ & $T_{2}=0.0767$ & $T_{3}=0.1342$ & $T_{4}=0.2108$ & $T_{5}=0.2875$ & $T_{6}=0.3833$\\ 
 \hline\hline
 (90$\%$,250$\%$) & (90$\%$,250$\%$) & (80$\%$,310$\%$) &  (80$\%$,300$\%$) & (75$\%$,395$\%$) & (80$\%$,405$\%$) \\ 
 \hline
\end{tabular}
\end{center}
\vspace{0.2cm}
We underline that the shortest maturity considered is 14 days. Regarding our modeling choice we fix $d=2$, $n=3$, which means to calibrate $40$ parameters. 
For $X$ we choose a $2$-dimensional Ornstein-Uhlenbeck processes, see Example~\ref{example_OU}, with the following (hyperparameter) configuration:
\begin{equation*}
    \kappa=(0.1,25)^\top, \qquad \theta=(0.1,4)^\top,\qquad \sigma=(0.7, 10)^\top, \qquad  \rho=	\begin{pmatrix}
	1 & -0.577 & 0.3 \\
	\cdot & 1 & -0.6\\
	\cdot & \cdot & 1 \\
	\end{pmatrix},
\end{equation*}
where we slightly abuse notation and denote by $\rho$ the correlation matrix of $(X, B)$. This implies that its last column describes the correlations of $X$ with the Brownian motion $B$ driving the price process $S$.

These hyperparameters are chosen randomly. Indeed, in spirit of reservoir computing, the idea is to view the OU-process' signature as (randomly chosen) reservoir, while a simple readout mechanism is trained, i.e. the linear function defined by $\{\ell_{I}:|I|\le n \}$, to map the state of the reservoir to the desired output (in our case instantaneous volatilities). However, it is of course possible to perform a hyperparameter optimization or to add expert knowledge, e.g.~that a high mean reversion rate is important.
We tried the latter by mimicking a rough or strong mean-reverting model as suggested in \cite{R:19,R:22}.

We also refer to Appendix~\ref{appendix} for  numerical results where we use only a correlated $2$-dimensional Brownian motion as primary process, which yields significantly worse results. 
 Note that the second simplest choice after Brownian motion within the family of polynomial diffusions (also with exact simulation) is the Ornstein-Uhlenbeck process which we thus applied. \\

Before stating the loss function $\mathcal{L}$ that we employed in the calibration task, let us make the following remark.

\begin{remark}\label{remark:greeks}
    Let $f:\R^{+}\times\R^{+}\to \R^+$ be the call pricing functional in the Black-Scholes model, depending on the volatility $\sigma^{\textrm{BS}}$ and the spot price $\xi$, i.e., $f: (\sigma^{\textrm{BS}},\xi)\mapsto f(\sigma^{\textrm{BS}},\xi)$. By Taylor expansion in an appropriate neighbourhood of $(\sigma^{mkt},\xi^{mkt})$ we obtain
    \begin{align*}
        f(\sigma^{\textrm{BS}},\xi)\approx& f(\sigma^{mkt},\xi^{mkt})+\frac{\partial f}{\partial \sigma }(\sigma^{mkt},\xi^{mkt})(\sigma^{\textrm{BS}}-\sigma^{mkt})+\frac{\partial f}{\partial \xi}(\sigma^{mkt},\xi^{mkt})(\xi-\xi^{mkt}),
    \end{align*}
    which equivalently gives
    \begin{align} \label{eq:taylor}
      (\sigma^{\textrm{BS}}-\sigma^{mkt})\approx&\frac{1}{\frac{\partial f}{\partial \sigma }(\sigma^{mkt},\xi^{mkt})}\bigl( f(\sigma^{\textrm{BS}},\xi)- f(\sigma^{mkt},\xi^{mkt})\bigr)-\frac{\frac{\partial f}{\partial \xi}(\sigma^{mkt},\xi^{mkt})}{\frac{\partial f}{\partial \sigma }(\sigma^{mkt},\xi^{mkt})}(\xi-\xi^{mkt}),
    \end{align}
    where we recognize for the derivatives with respect to $\sigma$ and $\xi$, the Greeks Vega and Delta, respectively. 
    \end{remark}
 Motivated by Remark~\ref{remark:greeks} we propose, for a fixed maturity and strike price, the following loss-function for $\beta\in \{0,1\}$
\begin{align}\label{greek_loss}
    &\mathcal{L}^{\beta}(\pi, \pi^{mkt,b,a}, \sigma^{mkt,b,a}, F, F^{mkt})=\\ \notag &\quad\left(\frac{\big(\beta\tilde{1}_{\{\pi \notin [\pi^{mkt,b},\pi^{mkt,a}]\}}+ (1-\beta)\big)\big|\pi-(\pi^{mkt,a}+\pi^{mkt,b})/2\big|+\big|\delta^{mkt}e^{-rT}(F-F^{mkt})\big|}{{\upsilon^{mkt} (\sigma^{mkt,a}-\sigma^{mkt,b})}}\right)^2,
\end{align}
where 
\begin{itemize}
    \item $\upsilon^{mkt}$ and $\delta^{mkt}$ denote the Vega and Delta of the option under the Black-Scholes model which depend on the maturity and on the strike price;
    \item $F$ and $F^{mkt}$ denote futures with maturity $T$ such that the variables $\xi, \xi^{mkt}$ appearing in Remark~\ref{remark:greeks} are  $\xi=e^{-rT}F$ and $\xi^{mkt}=e^{-rT}F^{mkt}$, respectively;
    \item $\tilde{1}_{\{x \notin [y^{b},y^{a}]\}}:=s(y^b-x)+s(x-y^a)$ for $s(x):=\frac{1}{2}\tanh(100x)+\frac{1}{2}$ a smooth version of the indicator function.
\end{itemize}
\begin{remark}
\begin{enumerate}
\item 
We observe that by Remark~\ref{remark:greeks} minimizing $\mathcal{L}^0$ is equivalent to minimizing an upper bound of the square of the right-hand side of \eqref{eq:taylor} normalized by the bid-ask spread of the implied volatilities. Note that we slightly abused notation, since $\upsilon^{mkt}$ and $\delta^{mkt}$ of course depend on the strike and the maturity.

\item Note that as $\ell \mapsto \VIX_T(\ell, \omega)=\frac{1}{\sqrt{\Delta}} \| U_T^{\top}\ell\|$ is convex and the call payoff is convex and increasing, the model option and futures prices are convex in $\ell$. If $\beta=0$ and the initialization of $\ell$ is such that both the model and futures prices are higher than the market ones, then we actually deal with a convex optimization problem.

\item If our aim does not consist in calibrating to the mid-price or mid-implied-volatility precisely, but we merely want to be within the bid-ask spreads we can set $\beta=1$

\end{enumerate}
\end{remark}

For the next calibration result we minimize $\Lcal^{1}$ as introduced above   with  $N_{MC}=80000$ Monte Carlo samples for the previous maturities and strikes.

\begin{center}
 \begin{figure}[H]
        \centering
		\includegraphics[scale=0.13]{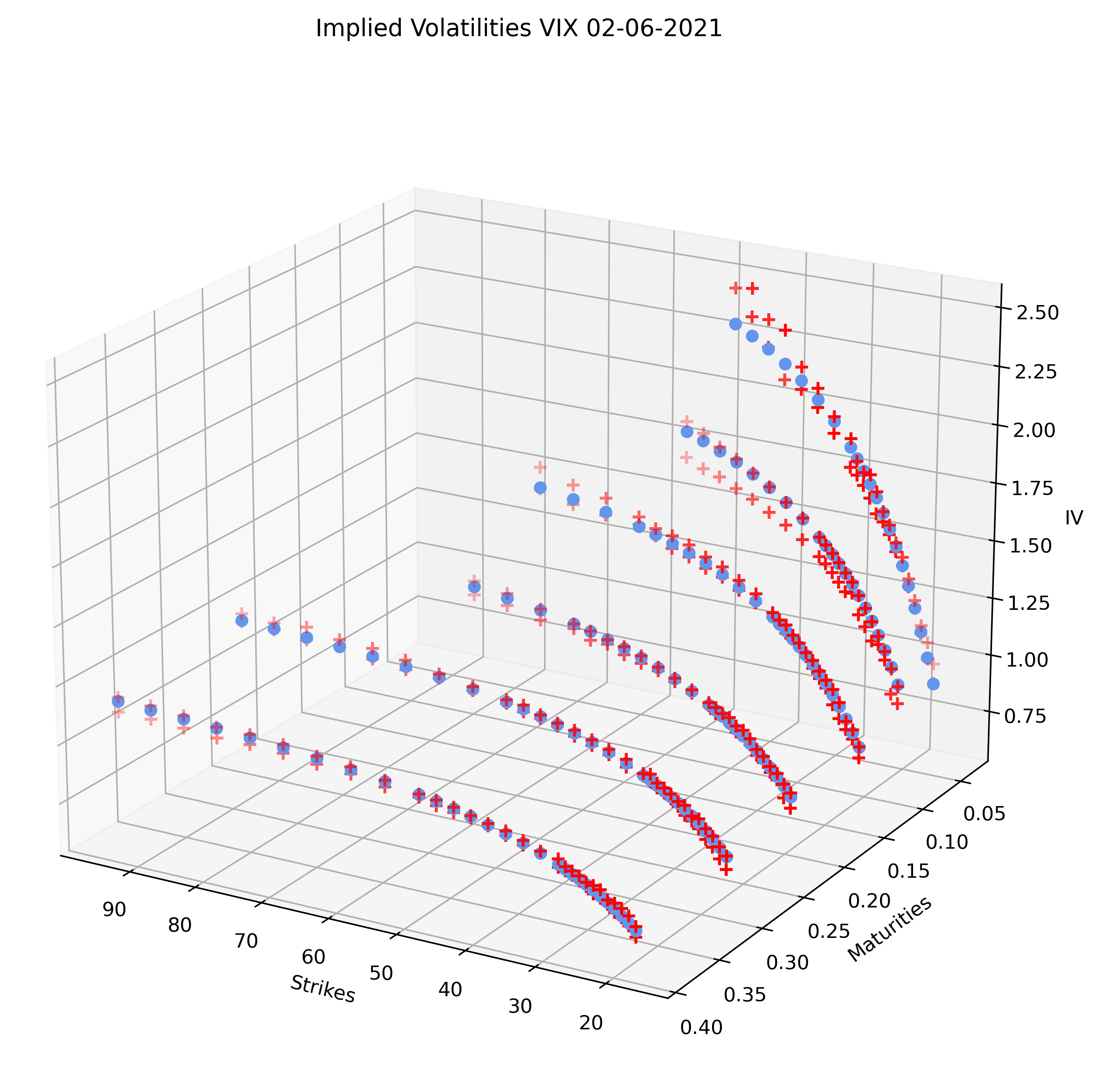}
		\caption{The red crosses denote the bid-ask spreads (of the implied volatilities) for each maturity, while the azure dots denote the calibrated implied volatilities of the model. On the $x$-axis we find the strikes and on the $y$-axis we find the maturities.}
		\label{fig:iv_calib1}
\end{figure}
\end{center}
We observe that the calibrated VIX smiles fall systematically in the bid-ask interval for all the maturities considered. We report additionally in the next tables the relative error between the market futures prices and the calibrated ones for each maturity, i.e.,
\begin{equation}\label{relative_AE}
    \varepsilon_{T}:=\frac{|F^{mkt}(T)-F_{\VIX}^{\text{model}}(\ell^{\ast},T)|}{F^{mkt}(T)},
\end{equation}
where $\ell^\ast\in \R^{40}$ denotes the calibrated parameters and here $F_{\VIX}^{\text{model}}(\ell^{\ast},T)$ stands for the calibrated future model price. In Figure~\ref{fig:FTS} we can find an illustration of the calibrated and the market futures' term structure.
\begin{center}
\begin{tabular}{||c |c|c ||} 
 \hline
 $T_{1}=0.0383$ & $T_{2}=0.0767$ & $T_{3}=0.1342$\\ 
 \hline\hline
 $\varepsilon_{T_{1}}=7.0 \times 10^{-6}$ & $\varepsilon_{T_{2}}=2.1 \times 10^{-3}$ & $\varepsilon_{T_{3}}=1.3 \times 10^{-5}$\\ 
 \hline
\end{tabular}
\end{center}
\begin{center}
\begin{tabular}{||c |c|c||} 
 \hline
 $T_{4}=0.2108$ & $T_{5}=0.2875$ & $T_{6}=0.3833$\\ 
 \hline\hline
  $\varepsilon_{T_{4}}=1.5  \times 10^{-4}$ & $\varepsilon_{T_{5}}=1.9  \times 10^{-6}$ & $\varepsilon_{T_{6}}=1.3  \times 10^{-6}$ \\ 
 \hline
\end{tabular}
\end{center}

\begin{center}
 \begin{figure}[H]
        \centering
		\includegraphics[scale=0.45]{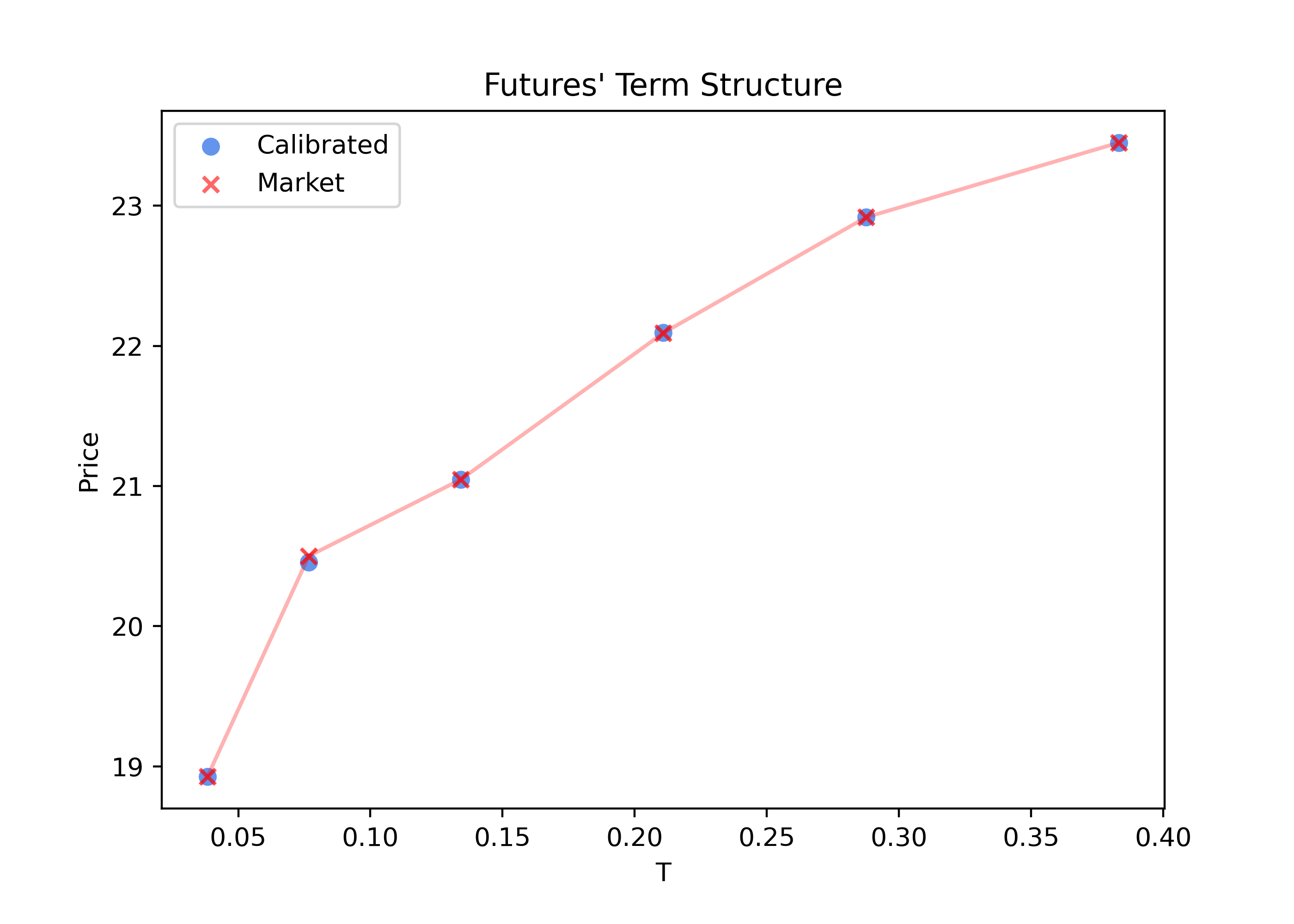}
		
		\caption{The blue circles denote the calibrated futures prices and the red crosses the futures prices on the market, in between a linear interpolation is reported.}
		\label{fig:FTS}
\end{figure}
\end{center}

\subsection{The case of time-varying parameters}\label{TV_VIX}
We now consider the case of maturity dependent parameters as for instance employed in \cite{GSSSZ:20,CGS:23}. Since it will be important later on to distinguish maturities of options written on the VIX index from  maturities of options written on the SPX index, we introduce the two sets $\mathcal{T}^{\VIX}$ and $\mathcal{T}^{\SPX}$. Let us fix here  $\mathcal{T}^{\VIX}=\{T_{1},\dots,T_{N}\}$, where $T_{i}<T_{i+1}$ for any in $i=1,\dots,N-1$ and denote by $\ell(T_{i})\in \mathbb{R}^{d_{n}}$ the parameters depending on the maturity $T_{i}>0$.  We set $T_0=0$ and $T_{N+1}=+ \infty$.
Then, we consider for any $t\geq 0$ the volatility process to be
\begin{equation}\label{variance_TV}
    \sigma_{t}^{S}(\ell)
    = \sum_{i=0}^{N} \sum_{|I|\le n}\ell_{I}(T_{i})1_{[T_{i}, T_{i+1})}(t) \langle e_{I},\widehat{\mathbb{X}}_{t}\rangle.
\end{equation}
Therefore the variance process reads as follows,
\begin{equation}\label{variance_process_tv}
    {V_{t}}(\ell)= \sum_{i=0}^{N} \sum_{|J|,|I|\le n}\ell_{I}(T_{i})\ell_{J}(T_{i})1_{[T_{i}, T_{i+1})}(t) \langle e_{I}\shuffle e_J,\widehat{\mathbb{X}}_{t}\rangle.
\end{equation}

\begin{assumption}\label{ass:distance_mat}
    Assume that for a set of maturities $\mathcal{T}^{\VIX}$ it holds that $|T_i-T_j|\geq \Delta$ for all $i\neq j$. 
\end{assumption}

\begin{proposition}\label{prop:vix_tv}
    Let $\mathcal{T}^{\VIX}$ be a set of maturities on the VIX index and let $Q(T,\tau)$ be the matrix  as defined in \eqref{matrix_Q} (here for general $\tau >0$ instead of $\Delta$).  
    Then, under \eqref{variance_process_tv} the VIX squared at time $T_i\in \mathcal{T}^{\VIX}$ is given by
        \begin{align*}
        \textrm{VIX}_{T_i}^2(\ell)&= \frac{1}{\Delta} \Big( \sum_{j=i}^{N} \ell(T_j)^{\top} \left({Q}(T_i, (T_{j+1}-T_i)\land \Delta) -  {Q}(T_i, (T_j-T_i)\land \Delta)\right)\ell(T_j) \Big).
    \end{align*}
    Note that, if $T_{i+1}-T_i > \Delta$ (which is in particular holds under Assumption~\ref{ass:distance_mat}) then,
    \begin{equation*}
        \textrm{VIX}_{T_i}^2(\ell)= \frac{1}{\Delta} \ell(T_i)^{\top} Q(T_i,\Delta) \ell(T_i).
    \end{equation*}
\end{proposition}

\begin{proof}
By the definition of the VIX, it holds that
\begin{align*}
\VIX_{T_i}^2(\ell)
= &\frac{1}{\Delta} \mathbb{E}\biggl[ \int_{T_i}^{T_i+ \Delta} \sum_{j=i}^N \sum_{|J|,|I|\le n}\ell_{I}(T_{j})\ell_{J}(T_{j})1_{[T_{j}, T_{j+1})}(t)  \langle e_{I}\shuffle e_J,\widehat{\mathbb{X}}_{t}\rangle  \d t\bigg\lvert \mathcal{F}_{T_i}\biggr] \\
=&
\frac{1}{\Delta} \sum_{j=i}^N\sum_{ |J|, |I|\le n}\ell_{I}(T_{j})\ell_{J}(T_{j})\mathbb{E}\biggl[ \int_{T_{j}\land (T_i+\Delta) }^{T_{j+1}\land (T_i+\Delta)} \langle e_{I}\shuffle e_J,\widehat{\mathbb{X}}_{t}\rangle  \d t\bigg\lvert \mathcal{F}_{T_i}\biggr]\\
=& \frac1{\Delta}  \sum_{j=i}^N\sum_{ |J|, |I|\le n}\ell_{I}(T_{j})\ell_{J}(T_{j})  \Bigg(\mathbb{E}\biggl[\int_{T_i}^{T_{j+1}\land (T_i+\Delta)}\langle e_{I}\shuffle e_J,\widehat{\mathbb{X}}_{t}\rangle\d t\bigg\lvert \mathcal{F}_{T_i}\biggr]\\
&\qquad-\mathbb{E}\biggl[\int_{T_i}^{T_{j}\land (T_i+\Delta)}\langle e_{I}\shuffle e_J,\widehat{\mathbb{X}}_{t}\rangle\d t\bigg\lvert \mathcal{F}_{T_i}\biggr]\Bigg)
\end{align*}
and hence the first statement follows by the definition of $Q$ in \eqref{matrix_Q}.
\end{proof}

Notice that also in the case of Proposition~\ref{prop:vix_tv}, Remark~\ref{remark:cholesky} applies.

\section{SPX as a signature-based model}\label{spx_sigsde}

The goal of this section is to express the discounted, dividend-adjusted price of the SPX, modeled via \eqref{model}-\eqref{sig-vol2}
$$
\mathrm{d}S_{t}(\ell)=S_{t}(\ell)\sigma_t^S(\ell)\mathrm{d}B_{t}, 
$$
 in terms of the
signature of $(t, X_t , B_t)_{t \geq 0}$, allowing again to precompute  its samples and use the same ones  for every $\ell$.
This is in the same spirit as in \cite{CGS:23}, even though there the asset price was directly modeled as  linear function of the signature of some primary process.

Recall that by \eqref{sig-vol2} 
 $\sigma^S$ is parametrized as follows
$$
	\sigma_t^S(\ell):=\ell_{\emptyset}+\sum_{0<\lvert I \lvert \le n } \ell_{I} \langle e_{I},\widehat{\mathbb{X}}_{t}\rangle,
$$
where $\widehat{X}_{t}=(t,X_t)$ with $X$ a $d$-dimensional polynomial diffusion $X$ in the sense of Definition~\ref{def1}. Before addressing a more tractable expression for $S$, that allows to avoid (Euler) simulation schemes, we recall the following well-known integrability result.
\begin{lemma}
Assume that  $\mathbb{E}[S_0]< \infty$. Then, the process $(S_t)_{t\geq0}$ is a (non-negative) supermartingale and in particular $\E[S_t]<\infty$ for each $t\geq 0$. 
\end{lemma}

\begin{proof}
    Note that $S_t= S_0\mathcal{E}\left( \int_0^{\cdot} \sigma_s^S \d B_s\right)_{t}$ for all $t\ge0$. Moreover $(\int_0^{t} \sigma_s^S \d B_s)_{t\ge0}$ is a local martingale and hence, by the properties of the stochastic exponential, $S_t$ is a non-negative local martingale. 
    It follows from Fatou's Lemma that non-negative local martingales are supermartingales. 
\end{proof}
In the following we suppose without loss of generality that $S_{0}=1$.
\begin{remark}
    Recall that if Novikov's condition is satisfied, then a stochastic exponential of the form $S_t= \mathcal{E}\left( \int_0^{\cdot} \sigma_s^S \d B_s\right)_{t}$ for $t\in [0,T]$ is a true martingale. For $\sigma_s^S$ as in \eqref{sig-vol2},  such condition reads 
    \begin{equation*}
        \E\biggl[\exp\biggl\{\frac{1}{2}\int_{0}^{T}V_{t}(\ell)\d t\biggr\}\biggr]<+\infty.
    \end{equation*}
    Observe that
    \begin{align}\label{novikov}\nonumber
         \E\biggl[\exp\biggl\{\frac{1}{2}\int_{0}^{T}V_{t}(\ell)\d t  \biggr\}\biggr]&=\E\biggl[\exp\biggl\{\frac{1}{2}\sum_{|I|,|J|\le n}\ell_{I}\ell_{J}\int_0^T\langle e_{I}\shuffle e_J,\widehat{\X}_{t}\rangle \d t\biggr\}\biggr]\\
         &=\E\biggl[\exp\biggl\{\frac{1}{2}\ell^\top Q^{0}(T)\ell\biggr\}\biggr],
    \end{align}
   where for $\mathscr{L}:\{I: |I|\le n\}\to \{1,\dots, (d+1)_{n}\},$
    \begin{equation*}
        Q_{\mathscr{L}(I)\mathscr{L}(J)}^{0}(T):=\langle (e_{I}\shuffle e_{J})\otimes e_{0},\widehat{\X}_{T}\rangle.
    \end{equation*}
    We point out that the previous condition is not necessarily satisfied for all $\ell\in \R^{(d+1)_{n}}$. Indeed, let us consider $X$ to be a one-dimensional Brownian motion and choose $\ell$ such that the only non trivial component is the last one, i.e., 
    \begin{align}
        \ell_{I}:=\begin{cases}
        c \in \R, \quad &\text{if}\  I=(1,\dots,1),\ |I|=n,\\
        0, \quad &\text{otherwise}.
        \end{cases}
    \end{align}
    Then, \eqref{novikov} translates into 
    \begin{equation*}
        \E\biggl[\exp\biggl\{\frac{c^2}{2}\int_0^T \frac{2n\text{!}}{n\text{!}n\text{!}} \frac{X_{t}^{2n}}{2n \text{!}}\mathrm{d}t\biggr\}\biggr]=\E\biggl[\exp\biggl \{\frac{c^2}{2(n\text{!})^2}\int_0^T   X_{t}^{2n} \mathrm{d}t\biggr \}\biggr],
    \end{equation*}
    which is not finite in general, e.g. if $n=2$, $c=\sqrt{2}(n\text{!})$ then by Jensen's inequality it follows
    \begin{equation*}
        \E\biggl[\exp\biggl \{\int_0^T   X_{t}^{4} \mathrm{d}t\biggr \}\biggr]\ge \E\biggl[  \frac{1}{T} \int_0^T e^{ TX_{t}^{4}}\mathrm{d}t\biggr]=\frac{1}{T}\int_{0}^{T}\E[e^{TX_{t}^{4}}]\d t =+\infty.
    \end{equation*}
\end{remark}

{ 
\begin{remark}
As well known from the results of \cite{DS:94}
 the existence of an equivalent local martingale measure is sufficient for NFLVR, and  risk neutral pricing works, too. This is important when the process $S(\ell)$ is a non-negative true \emph{local} martingale such that $\E[S_T(\ell)]<S_{0}$. If one could go short 
in the asset, and thus gets $S_0(\ell)$,
and long in the `call option with strike 0' (corresponding to the payoff $S_T(\ell)$) with price $\E[S_T(\ell)] < S_0(\ell)$, an arbitrage would be created. But the latter is simply not allowed as trading strategy under NFLVR.
We address the reader to \cite{KKN:15} for further details.
\end{remark}
}

The key idea is to rewrite \eqref{model} as a type of signature-based model in sense of \cite{CGS:23} including $B=(B_{t})_{t\ge0}$ as part of the primary process. This is possible since It\^o integrals with respect to primary process' components can be rewritten as linear functions of the signature of the primary process itself. Before stating the result we need to introduce some auxiliary notation. We denote by $(Z_t)_{t\geq0}$ the $(d+1)$-dimensional process given by 
\begin{equation}\label{eqZ}
Z_t=(X_t,B_t),
\end{equation}
by $(\widehat Z_t)_{t\geq0}$ its time extension, and by $(\widehat \Z_t)_{t\geq0}$ the signature of $(\widehat Z_t)_{t\geq0}$. With a slight abuse of notation we again denote by $\rho$ the correlation matrix process between the components of $Z$. Observe that $\rho$ encodes in particular the correlation between $X$ and $B$. Finally, we let $a_{ij}^{J}\in\R$ denote the coefficients satisfying 
$$\d[Z^i,Z^j]_t=\sum_{|J|\leq 2} a_{ij}^J\langle e_J,\widehat \Z_t\rangle \d t,$$
for each $i,j\in\{1,\ldots,d+1\}$.

\begin{proposition}\label{sig-model-spx}
    Let $S=(S_{t})_{t\ge0}$ satisfy \eqref{model} with $S_0=1$, and $\sigma^S=(\sigma_{t}^S)_{t\ge0}$ satisfy \eqref{sig-vol2}.  Then,
    \begin{equation}\label{logsig-model-representation-SPX}
        \log(S_{t}(\ell))=-\frac{1}{2}\ell^\top Q^{0}(t)\ell+\sum_{|I|\le n}\ell_{I}\langle \tilde{e}_{I}^{B},\widehat{\mathbb{Z}}_{t}\rangle,
    \end{equation}
    where 
 \begin{equation*}
     \tilde e_\emptyset^B:=e_{d+1}, \qquad \tilde e_I^B:=e_{I}\otimes e_{d+1}- \sum_{|{ J}|\leq 2}\frac {a_{i_{|{ I}|}(d+1)}^J} 2(e_{{I}'}\shuffle e_{ J})\otimes e_0 ,
 \end{equation*} 
 for each $|I|>0$, and the components of the matrix $Q^{0}(t)\in \mathbb{R}^{(d+1)_{n}\times (d+1)_{n}}$ are given by
\begin{equation*}
    Q^{0}_{\mathscr{L}(I)\mathscr{L}(J)}(t)=\langle (e_{I}\shuffle e_{J})\otimes e_{0},\widehat{\mathbb{X}}_{t}\rangle,
\end{equation*}
for a labeling function $\mathscr{L}:\{I: |I|\le n\}\to \{1,\dots,(d+1)_{n}\}$.
\end{proposition}
\begin{proof}
   We can compute
    \begin{align*}
         \log(S_{t}(\ell))
         &=-\frac{1}{2}  \int_0^tV_s(\ell)\d s + \int_0^t\sigma_s^S(\ell)\d B_{s}\\
         &=-\frac{1}{2}\sum_{|I|,|J|\le n}\ell_{I}\ell_{J}\int_{0}^{t}\langle e_{I}\shuffle e_{J},\widehat{\mathbb{X}}_{s}\rangle\d s
         + \sum_{|I|\le n}\ell_{I}\int_{0}^{t}\langle e_{I},\widehat{\mathbb{X}}_{s}\rangle\mathrm{d}B_{s}\\
         &\stackrel{(\ast)}=-\frac{1}{2}\sum_{|I|,|J|\le n}\ell_{I}\ell_{J}\langle( e_{I}\shuffle e_{J})\otimes e_{0},\widehat{\mathbb{X}}_{t}\rangle +\sum_{|I|\le n}\ell_{I}\langle \tilde{e}_{I}^{B},\widehat{\mathbb{Z}}_{t}\rangle\\
          &=-\frac{1}{2}\ell^\top Q^{0}(t)\ell+\sum_{|I|\le n}\ell_{I}\langle \tilde{e}_{I}^{B},\widehat{\mathbb{Z}}_{t}\rangle,
    \end{align*}
where for $(\ast)$ we used that $\int_{0}^{t}\langle e_{I},\widehat{\mathbb{X}}_{s}\rangle\mathrm{d}B_{s}=\langle \tilde{e}_{I}^{B}, \widehat{\mathbb{Z}}_{t}\rangle$ by Lemma~3.10 in \cite{CGS:23}. 
\end{proof}

\begin{remark}
Consider again the model described in Remark~\ref{rem:altmod1}. Then the results of Proposition~\ref{sig-model-spx} still hold with 
$$Q^{0}_{\mathscr{L}(I)\mathscr{L}(J)}(t):=\int_0^t\langle e_{I}\shuffle e_{J},\widehat\X_{s-\e,s}\rangle
\d s,$$
and $\int_0^t \langle e_I,\widehat \X_{s-\e,s}\rangle \d B_s$ instead of $\langle \tilde{e}_{I}^{B},\widehat{\mathbb{Z}}_{t}\rangle$. Since the proof follows closely the proof of the original result, we omit it.
 \end{remark}

\begin{remark}\label{remark_Q0}
    \begin{itemize}
    \item Observe that since the matrix $(\langle e_{I}\shuffle e_{J},\widehat{\X}_{t}\rangle)_{|I|,|J|\le n}$ is positive semidefinite, by monotonicity of the time integral on $[0,t]$ for some $t>0$, we also have
\begin{equation*}
    \ell^{\top}Q^{0}(t)\ell \ge 0,
\end{equation*}
for all $\ell \in \mathbb{R}^{(d+1)_{n}}$. This means that for any $t>0$, we can rewrite the log-price as 
\begin{equation*}
    \log(S_{t})=-\frac{1}{2}\lVert (U_{t}^{0})^ {\top}\ell\lVert^2+\sum_{|I|\le n}\ell_{I}\langle \tilde{e}_{I}^{B},\widehat{\Z}_{t}\rangle,
\end{equation*}
where $U^0_{t}$ is the upper-triangular matrix of the Cholesky decomposition of $Q^{0}(t)$.
\item Notice that the log-price model in \eqref{logsig-model-representation-SPX}, it is not exactly a signature-based model in the sense of \cite{CGS:23}, as here it is given by a linear part in the parameters $\ell$ and an additional quadratic part. It can also be rewritten as
\begin{equation*}
    \d\log(S_t)=-\frac{1}{2}\ell^\top \tilde{Q}(t)\ell \d t +\ell^\top \vecsig(\widehat{\X}_{t}^{n})\d B_{t},
\end{equation*}
where $\tilde{Q}$ is given by 
\begin{equation*}
    \tilde{Q}_{\mathscr{L}(I)\mathscr{L}(J)}(t):=\langle e_{I}\shuffle e_{J},\widehat{\mathbb{X}}_{t}\rangle.
\end{equation*}
 Hence, the relevant factors entering in the dynamics of $\log S$ and $\sigma^S$ are the components of $\widehat{\mathbb{X}}^{2n}$. Note also that $(\log S, \widehat{\mathbb{X}}^{2n})$ is a $1+(d+1)_{2n}$ dimensional polynomial diffusion (see \eqref{eqn10}), whence in particular Markovian. This is in spirit of path-dependent factor model, for instance also considered in \cite{GLJ:22}, with the additional tractability feature that $(\log S, \widehat{\mathbb{X}}^{2n})$ is a polynomial diffusion. Therefore all techniques for polynomial processes in view of pricing and hedging can be applied.

\item In order to sample the log-price at  maturity, consistently with the $\VIX$, we follow the following road map. We simulate $\widehat{\Z}$ and  compute $\langle\tilde{e}_{I}^{B}, \widehat{\Z}\rangle$ for each $I$ as specified above. Next, we drop from the samples of
$\widehat{\Z}$ the terms where  $B$ appears, i.e.~the components corresponding to indices containing the letter $d+1$.  The result coincides with a sampling of $\widehat \X$ and is then used to work with both $Q$ and $Q^0$.

This is equivalent to sampling $\widehat{\X}$ for the variance process and to compute an additional It\^o integral as in \eqref{model}.
    \end{itemize}
\end{remark}

In the following corollary we  state the form of  $\tilde{e}_I^B$ when $X$ is 
$d$-dimensional OU-process. We omit the proof for sake of brevity.

 \begin{corollary}\label{corollary:Z_OU}
    Let $X$ be a $d$-dimensional OU-process as in Example \ref{example_OU} driven by a $d$-dimensional Brownian motion with correlation matrix $\rho$. Then $\tilde{e}_I^B$ is given by
    \begin{equation*}
    \tilde{e}_I^B = e_I\otimes e_{d+1} -\frac{1}{2} 1_{\{i_{|I|}\neq 0\}} (\sigma^{i_{|I|}} \rho_{i_{|I|} d+1}) e_{I'} \otimes e_0,
    \end{equation*}
    for any multi-index $I \neq \emptyset$.
\end{corollary}

\begin{remark}[Variance reduction for pricing SPX options]\label{var_red:spx}
    Observe that a possible control variate for reducing the variance of the Monte Carlo estimator for pricing $\SPX$ options is the value at maturity of the log-price process. This  means, 
    \begin{align*}
       \Phi^{cv}(\ell,T,K):
       &=c_{T,K}\Big(\log(S_{T}(\ell))+\frac{1}{2}\ell^\top Q^{0,cv}(T)\ell\Big),
    \end{align*}
    where, using that the linear part (in $\ell$) of $\log(S_{T}(\ell))$ vanishes under the risk-neutral expectation, we have
    \begin{equation*}
        Q_{\mathscr{L}(I)\mathscr{L}(J)}^{0,cv}(T)= \vecsig((e_I\shuffle e_J)\otimes e_0)^{\top}  e^{TG^\top}\vecsig(\widehat{\X}_0^{2n+1}),
    \end{equation*}
    for  $G\in \R^{(d+1)_{2n+1}\times(d+1)_{2n+1}}$ denoting the $(d+1)_{2n+1}$-dimensional matrix representative of the dual operator corresponding to $\widehat \X$. We choose the optimal $c_{T,K}^{\ast}\in\R$ as
    \begin{equation*}
    c_{T,K}^{\ast}=\frac{\text{Cov}((S_{T}(\ell)-K)^{+},\log(S_{T}(\ell)))}{\text{Var}(\log(S_{T}(\ell)))}.
\end{equation*}
\end{remark}

\subsection{Exploiting the affine nature of the signature: Fourier pricing of SPX and VIX options}\label{sec:affine-nat}

This section is dedicated to outline how the linear parametrizations of the log-price and the volatility process in $\widehat{\mathbb{Z}}$ can be used for Fourier pricing.
Assume that 
\begin{equation*}
	\mathrm{d} Z_{t}^{j}=\kappa^{j}(\theta^{j}-Z_{t}^{j})\mathrm{d}t + \sigma^{j}\mathrm{d} W_{t}^{j},\qquad Z_0^j=0,
\end{equation*}
 for each $j=1,\dots,d+1$, where $W$ denotes a $(d+1)$-dimensional Brownian motion with $W^{d+1}=B$.
  All parameters $\kappa^j, \theta^j, \sigma^j $ are in $ \mathbb{R}$ with $\kappa^{d+1}=\theta^{d+1}=0$ and $\sigma^{d+1}=1$ so that $Z^{d+1}=W^{d+1}=B$. Note that we do not account for correlations.
 
 We illustrate now how to apply the results of \cite{CST:22} in the present setting. Since $(\widehat Z_t)_{t\geq0}$ is a polynomial diffusion, by Lemma~\ref{lem:lem2} there are $\b\in (T((\R^{d+2})))^{d+2}$ and $\a\in (T((\R^{d+2})))^{(d+2)\times (d+2)}$ such that
\begin{equation*}
	\mathrm{d} \widehat Z_{t}^{j}=\langle \b_{j},\widehat \Z_t\rangle \mathrm{d}t + \sqrt{\langle \a_{jj},\widehat\Z_t\rangle}\mathrm{d} W_{t}^j,
\end{equation*}
where  $\b_j=\kappa^j\theta^je_\emptyset-\kappa_je_j$ and $\a_{jj}=(\sigma^j)^2e_\emptyset$, using that (with a small abuse of notation) $\kappa^0\theta^0:=1$, $\kappa^j:=0$ and $\sigma^0:=0$.
Consider then the  Riccati operator $\Rcal$  given by 
\begin{align*}
    \Rcal(\mathbf u)&=\sum_{j=0}^{d+1} \sum_{|I|\geq0}
\Big(\kappa^j\theta^j\u_{(Ij)}e_{I}
+ \kappa^j \u_{(Ij)}e_j\shuffle e_{I}+\frac 1 2(\sigma^j)^2(\u_{(Ijj)}e_{I}+\u_{(Ij)}^2e_I\shuffle e_I)
\Big).
\end{align*}

By Theorem~4.23 in \cite{CST:22}, we expect that
\[
\mathbb{E}[\exp(\langle \mathbf{u}, \widehat{\Z}_T  \rangle)]= \exp( {\bm \psi}(T)_\emptyset  ),
\]
where ${\bm \psi}$ is a solution of  the extended tensor algebra valued Riccati equation\footnote{We refer to  \cite{CST:22} for the appropriate solution concept and to a numerical treatment  in the one dimensional case where \eqref{eq:Ric} reduces to a sequence-valued Riccati equation.} 
\begin{align}\label{eq:Ric}
\partial_t {\bm \psi}(t)=\Rcal({ \bm\psi}(t)), \quad {\bm \psi}(0) =\mathbf{u}.
\end{align}
Choosing $\u$ as
$$\u(\ell):=-\frac 1 2 ({ \ell}\shuffle { \ell})\otimes e_0+\tilde\ell$$
where  $\tilde \ell:=\sum_{|I|\leq n}\ell_I\tilde e_I^B$, by Proposition~\ref{sig-model-spx} we get 
$$\log(S_t(\ell))=\langle \u(\ell),\widehat\Z_t\rangle.$$
The representation of the Fourier-Laplace transform described above can then be used for Fourier pricing. We dedicate the remaining part of this section to illustrate how this can be done. 

From Fourier analysis we know that for $K>0$ and $C<0$ it holds
$$(K-e^y)^+=\frac 1 {2\pi}\int_\R e^{(i\lambda+C)y}\frac{K^{-C+1-i\lambda}}{(i\lambda+C)(i\lambda+C-1)}\d\lambda.$$
This in particular implies that
\begin{align*}
    \E[(K-S_T(\ell))^+]
    &=\frac 1 {2\pi}\int_\R \E[e^{(i\lambda+C)\log(S_T(\ell))}]\frac{K^{-C+1-i\lambda}}{(i\lambda+C)(i\lambda+C-1)}\d\lambda\\
    &=\frac 1 {2\pi}\int_\R \E[e^{\langle \u_\lambda,\widehat \Z_T\rangle}]\frac{K^{-C+1-i\lambda}}{(i\lambda+C)(i\lambda+C-1)}\d\lambda\\
    &=\frac 1 {2\pi}\int_\R e^{ \bm\psi_\lambda(T)_\emptyset}\frac{K^{-C+1-i\lambda}}{(i\lambda+C)(i\lambda+C-1)}\d\lambda,
\end{align*}
where $\u_\lambda:=(i\lambda+C)\u(\ell)$ and $\bm\psi_\lambda$ is a solution of the Riccati equation with initial condition $\bm\psi_\lambda(0)=\u_\lambda$. 

Let us now consider the case of $\VIX$ options where Fourier pricing can be applied by computing the Fourier-Laplace transform of VIX squared, see also \cite{S:08,PS:14,BPS:22} and references therein for a Fourier-based approach to pricing VIX options. 
Fix a labelling injective function $\Lvar:\{I\colon |I|\leq n\}\to\{1,\ldots,(d+1)_{(2n+1)}\}$  as introduced before \eqref{eqn8} and recall that by Theorem~\ref{th:VIXclosed} it holds
\begin{equation*}
    \VIX_T^2(\ell)=\frac 1 \Delta \ell^\top Q(T,\Delta)\ell
\end{equation*}
for 
$$
 Q_{\mathscr{L}(I)\mathscr{L}(J)}(T,\Delta)= \sum_{e_K=(e_I\shuffle e_J)\otimes e_0}  \sum_{|H|\leq 2n+1}   (e^{\Delta G^\top}-\operatorname{Id})_{\Lvar(K)\Lvar(H)}  \langle e_H,\widehat{\X}_T\rangle.
$$
where $G$ denotes the $(d+1)_{(2n+1)}$-dimensional matrix representative of the dual operator corresponding to $\widehat \X$.

Setting for $|I|,|J|\le n$
$$I(\Delta)\u:=  \sum_{|K|,|H|\leq 2n+1}   (e^{\Delta G^\top}-\operatorname{Id})_{\Lvar(K)\Lvar(H)} \u_K e_H$$
we can write
\begin{align*}
    \VIX_T^2(\ell)
    &=\frac 1 \Delta\sum_{|I|,|J|\leq n} \ell_I\ell_J
    \sum_{e_K=(e_I\shuffle e_J)\otimes e_0}  \sum_{|H|\leq 2n+1}   (e^{\Delta G^\top}-\operatorname{Id})_{\Lvar(K)\Lvar(H)}  \langle e_H,\widehat{\X}_T\rangle\\
    &=\frac 1 \Delta \sum_{|K|,|H|\leq 2n+1}   (e^{\Delta G^\top}-\operatorname{Id})_{\Lvar(K)\Lvar(H)} \langle e_K,(\ell\shuffle \ell)\otimes e_{0}\rangle \langle e_H,\widehat\X_t\rangle\\
    &=\frac 1 \Delta \langle  I(\Delta)((\ell\shuffle \ell)\otimes e_{0}),\widehat\X_T\rangle.
\end{align*}
Also in this case since 
$\frac 1 {\sqrt{2\pi}}\int_\R e^{i\lambda y}(K-\sqrt{|y|})^+ \d y=\frac{\sqrt{\frac 2{\pi}}F_S(K\sqrt{|\lambda|})}{|\lambda|^{3/2}}$ is integrable for $F_S(u)=\int_0^{u} \sin( z^2)\d z$,
Fourier analysis yields 
$$(K-\sqrt{y})^+=\frac 1 {\pi}\int_\R e^{-i\lambda y}\frac{F_S(K\sqrt{|\lambda|})}{|\lambda|^{3/2}}\d \lambda,$$
for each $y\geq 0$. This in particular implies that
$$\E[(K-\VIX_T(\ell))^+]
=\frac 1 {\pi}\int_\R \E[e^{-i\lambda \VIX_T^2(\ell)}]\frac{F_S(K\sqrt{|\lambda|})}{|\lambda|^{3/2}}\d\lambda
=\frac 1 {\pi}\int_\R e^{ {\bm\psi}_\lambda(T)_\emptyset}\frac{F_S(K\sqrt{|\lambda|})}{|\lambda|^{3/2}}\d\lambda,$$
where $\bm \psi_\lambda$ is a solution of the Riccati equation with initial condition $\bm\psi_\lambda(0)=-i\lambda {\mathbf{v}}(\ell)$ for ${\bf v}(\ell):= \frac 1 \Delta I(\Delta)((\ell\shuffle \ell)\otimes e_{0})$. 

Analogous formulations in terms of the error function are also possible, see for instance \cite{BPS:22}. In the same spirit one can also obtain a representation of futures prices. 
We  here do not provide an implementation of this Fourier pricing approach but numerical experiments can be found in \cite{CST:22}.

\subsection{The case of time-varying parameters}\label{TV_SPX}
Analogously to Section~\ref{TV_VIX}, we now further enhance Proposition~\ref{sig-model-spx} by allowing the parameters $\ell$ to depend on the maturity.

\begin{proposition}\label{prop:price_tv}   Let $S=(S_{t})_{t\ge0}$ satisfy \eqref{model} with $S_0=1$, and $(\sigma_t^S)_{t\geq0}$ satisfy \eqref{variance_TV}  for a set of maturities $\Tcal^\VIX=\{T_{1},\dots,T_N\}$. Recall that in this case $V=(V_{t})_{t\ge0}$ satisfy
$${V_{t}}(\ell)= \sum_{i=0}^{N} \sum_{|J|,|I|\le n}\ell_{I}(T_{i})\ell_{J}(T_{i})1_{[T_{i}, T_{i+1})}(t) \langle e_{I}\shuffle e_J,\widehat{\mathbb{X}}_{t}\rangle.$$
Then, with the notation of Proposition~\ref{sig-model-spx} we write the following recursion for the log-price process 
\begin{align*}
    \log(S_t(\ell^{<m+1}))
    & =
    \sum_{i=0}^N\bigg[-\frac{1}{2} \ell(T_i)^{\top} \big( Q^0({t\land T_{i+1}})-Q^0(t\land T_i) \big) \ell(T_i)\\
    &\qquad+ \sum_{ |I|\leq n} \ell_I(T_i) \langle \tilde{e}_I^B, \widehat{\mathbb{Z}}_{{t\land T_{i+1}}} -\widehat{\mathbb{Z}}_{t\land T_{i}} \rangle\bigg] 
\end{align*}
for each $t\geq 0$, where $T_0:=0$, $\ell^{< m+1}:=\{\ell(0),\dots, \ell(T_{m})\}$,  
$m=\max\{j\colon  T_{j} < t\}$, and 
\begin{equation*}
    Q^{0}_{\mathscr{L}(I)\mathscr{L}(J)}(t)=\langle (e_{I}\shuffle e_{J})\otimes e_{0},\widehat{\mathbb{X}}_{t}\rangle,
\end{equation*}
for a labeling function $\mathscr{L}:\{I: |I|\le n\}\to \{1,\dots,(d+1)_{2n+1}\}$. 
\end{proposition}

\begin{proof}
 We know that 
\begin{equation*}
    \log(S_t(\ell)) = -\frac{1}{2} \int_0^t  V_s(\ell)\d s + \int_0^t \sigma_s^S(\ell)\d B_s
\end{equation*}
and we will calculate each integral separately. We start with the first one.
\begin{align*}
 \int_0^t V_s(\ell) \d s 
=&  \int_0^t \sum_{i=0}^N \sum_{ |I|, |J| \leq n} \ell_I(T_i) \ell_J(T_i) 1_{[T_i, T_{i+1})} \langle e_I\shuffle e_J, \widehat{\mathbb{X}}_s \rangle \d s \\
= &  \sum_{i=0}^{N}  \sum_{ |I|, |J| \leq n} \ell_I(T_i) \ell_J(T_i) \int_{t\land T_i}^{t\land T_{i+1}} \langle e_I\shuffle e_J, \widehat{\mathbb{X}}_s \rangle \d s  \\
= &   \sum_{i=0}^{N}  \sum_{ |I|, |J| \leq n} \ell_I(T_i) \ell_J(T_i) \left( \int_{0}^{t\land T_{i+1}} \langle e_I\shuffle e_J, \widehat{\mathbb{X}}_s \rangle \d s 
- \int_{0}^{t\land T_{i}} \langle e_I\shuffle e_J, \widehat{\mathbb{X}}_s \rangle \d s \right) \\  
=& \sum_{i=0}^{N}  \sum_{ |I|, |J| \leq n} \ell_I(T_i) \ell_J(T_i) \left( \langle (e_I\shuffle e_J)\otimes e_0, \widehat{\mathbb{X}}_{t\land T_{i+1}} \rangle -  \langle (e_I\shuffle e_J)\otimes e_0, \widehat{\mathbb{X}}_{t\land T_{i}} \rangle \right) \\
=& \sum_{i=0}^{N} \left[ \ell(T_i)^{\top} \left( Q^0(t\land T_{i+1}) - Q^0(t\land T_i)\right) \ell(T_i)\right].
\end{align*}
Using similar arguments and Lemma 3.10 in \cite{CGS:23}, the second integral yields
\begin{align*}
\int_0^t \sigma_s^S(\ell)\d B_s 
=& \int_0^t \sum_{i=0}^N \sum_{ |I| \leq n} \ell_I(T_i)  1_{[T_i, T_{i+1})} \langle e_I, \widehat{\mathbb{X}}_s \rangle \d B_s \\
=& \sum_{i=0}^{N}  \sum_{ |I| \leq n} \ell_I(T_i) \int_{t\land T_i}^{t\land T_{i+1}} \langle e_I, \widehat{\mathbb{X}}_s \rangle \d B_s  \\ 
=& \sum_{i=0}^N  \sum_{ |I| \leq n} \ell_I(T_i) \left( \int_{0}^{t\land T_{i+1}} \langle e_I, \widehat{\mathbb{X}}_s \rangle \d B_s -\int_{0}^{t\land T_{i}} \langle e_I, \widehat{\mathbb{X}}_s \rangle \d B_s\right) \\ 
=& \sum_{i=0}^N  \sum_{ |I| \leq n} \ell_I(T_i)   \langle \tilde{e}^B_I, \widehat{\mathbb{Z}}_{t\land T_{i+1}}- \widehat{\mathbb{Z}}_{t\land T_i} \rangle,
\end{align*}
and the claim follows.
\end{proof}

\section{Joint calibration of SPX and VIX options}\label{sec:joint-calib}
We here  consider again the model introduced in \eqref{model}-\eqref{sig-vol2}. Note that we just work with
call options, but the setup can easily be extended also to other liquid options on the market.  
Again we denote by $\Tcal^{\VIX}$ and $\Tcal^{\SPX}$ the maturities set for options written on $\VIX$ and $\SPX$, respectively. Similarly we use the notation $\Kcal^{\VIX}$ and $\Kcal^{\SPX}$ for the corresponding strikes.
The functional to be minimized in order to achieve a joint calibration of the SPX/VIX options reads as follows:
\begin{equation}\label{L_joint}
    L_{\text{joint}}(\ell,\lambda):=\lambda L_{\SPX}(\ell)+(1-\lambda) L_{\VIX}(\ell),
\end{equation}
where $\lambda\in(0,1)$ and
\begin{itemize}
    \item $L_{\text{VIX}}(\ell)$ is as in \eqref{L_cv_vix}, i.e.
    \begin{align*}
    \sum_{T\in \Tcal^{\VIX}, K\in\Kcal^{\VIX}}\!\!\!\mathcal{L}\left(\pi_{\VIX}^{\text{model}}(\ell,T,K),\pi_{\textrm{\VIX}}^{b,a}(T,K),\sigma_{\text{\VIX}}^{b,a}(T,K), F_{\VIX}^{\text{model}}(\ell,T),F_{\VIX}^{mkt}(T)\right),
    \end{align*}
    with $\pi_{\VIX}^{model}$ and $F_{\VIX}^{model}$ as in \eqref{model_prices} for
    $\VIX_T(\ell,\omega_i)$ defined as in \eqref{matrix_vix};
    \item $L_{\SPX}(\ell)$ is the SPX loss function given by
\begin{align*}
   L_{\SPX}(\ell):=\!\!\!\sum_{T\in\Tcal^{\SPX},K\in\Kcal^{\SPX}}\!\!\!\mathcal{L}(\pi_{\SPX}^{\text{model}}(\ell,T,K),\pi_{\textrm{\SPX}}^{b,a}(T,K),\sigma_{\text{\SPX}}^{b,a}(T,K)),
     \end{align*}
    for a real-valued function $\mathcal{L}$. Observe that with a slight abuse of notation we denote this function as the one for $L_{\VIX}$, but for $\SPX$ options we do not have to calibrate to futures, hence the last term of \eqref{greek_loss} vanishes. 
\end{itemize}
By Proposition~\ref{sig-model-spx} the SPX call option payoff with  maturity $T>0$ and a strike price $K>0$ reads in our model as follows
\begin{equation*}
    e^{-rT}(\tilde{S}_{T}(\ell)-K)^{+}=e^{-rT}\biggl(\exp\biggl\{(r-q)T-\frac{1}{2}\ell^\top Q^{0}(t)\ell +\sum_{|I|\le n}\ell_{I}\langle \tilde{e}_{I}^{B},\widehat{\mathbb{Z}}_{T}\rangle\biggr\}-K\biggr)^{+},
\end{equation*}
where $\tilde{S}$ denotes the undiscounted, unadjusted process as discussed in Remark~\ref{rem:interests} and $r,q>0$ the interest rate and the dividends, respectively. Recall also that the call option payoff written on the $\VIX$ is given by
\begin{equation*}
    e^{-rT}(\VIX_{T}(\ell)-K)^{+}= e^{-rT}\biggl(\sqrt{\frac{1}{\Delta}\ell^{\top}Q(T,\Delta)\ell}-K\biggr)^{+}= e^{-rT}\left(\frac{1}{\sqrt{\Delta}}\lVert U_{T}^{\top}\ell \lVert- K\right)^{+},
\end{equation*}
where $U_{T}$ denotes the upper-triangular matrix of the Cholesky decomposition of the symmetric positive semidefinite matrix $Q(T,\Delta)$.

{ 
\begin{remark}

We report 
in the table below 
the   (average over $10^3$ trials) timings of evaluating $\VIX_{T}(\ell)$ and $S_{T}(\ell)$ for $\ell\in\mathbb{R}^{85}$, a fixed $T>0$ and $N_{MC}=8 \cdot 10^{5}$  samples on both CPU (on the left) and GPU (on the right) with PyTorch, respectively:

\begin{center}
\begin{tabular}{||c| c ||} 
 \hline
 $(\VIX_{T}(\ell)(\omega_{i}))_{i=1}^{N_{\text{MC}}}$ & $(S_{T}(\ell)(\omega_{i}))_{i=1}^{N_{\text{MC}}}$\\
 \hline\hline
 $0.771s$ &  $0.702s$\\
 \hline
\end{tabular}
\quad
\begin{tabular}{||c| c ||} 
 \hline
 $(\VIX_{T}(\ell)(\omega_{i}))_{i=1}^{N_{\text{MC}}}$ & $(S_{T}(\ell)(\omega_{i}))_{i=1}^{N_{\text{MC}}}$\\
 \hline\hline
 $0.04s$ &  $0.014s$\\
 \hline
\end{tabular}
\end{center}

This evaluations are the relevant operations in the Monte Carlo pricing and in turn in the calibration procedure. Note again that  both the sampling of the signature components and the matrix exponential, are  achieved offline, as they do not depend on $\ell$. 
\end{remark}}

\subsection{Numerical results}\label{numerical_result_joint}

Before presenting our numerical results, let us discuss two different ways of approaching the joint calibration problem that can be found in the recent literature. 
\begin{enumerate}
    \item\label{approach_1} The first approach consists in choosing for instance the first maturity of SPX and VIX to coincide (or differ up to two days), i.e., $T_{1}^{\SPX}=T_{1}^{\VIX}$ and then for $j\ge2$, $T_{j}^{\SPX}=T_{j-1}^{\VIX}+\Delta$, see for instance \cite{G:21,GLOW:21,GLJ:22}.
 \item\label{approach_2} The second approach is to consider $\Tcal^{\SPX}=\Tcal^{\VIX}$, i.e., to choose the same (or  close together) maturities both for SPX and VIX options. This perspective has been adopted for instance by \cite{GJR:18,RZ:21,GL:22,BPS:22}.
\end{enumerate}

\begin{figure}[H]
    \centering
        \begin{tikzpicture}
        \draw (0,0)-- (11,0);
        
        \foreach \x in {0} {
            \draw (\x,0.1) -- (\x,-0.1) node[below] {$0$}; 
        }
        
        \foreach \x in {1,...,2} {
            \draw (\x*2-0.7,0.1) -- (\x*2-0.7,-0.1) node[below] {$T_{\x}$}; 
        }
        
        \draw (6.3,0.1) -- (6.3,-0.1) node[below] {$T_{1}+\Delta$};
        
        \draw (8.3,0.1) -- (8.3,-0.1) node[below] {$T_{2}+\Delta$};

         \draw[blue] (0,1.5) -- (1.3,1.5); 
        \draw[fill,blue] (0,1.5) circle (0.05); 
        \draw[fill,blue] (1.3,1.5) circle (0.05);

        \draw[blue] (0,1.3) -- (6.3,1.3);
        \draw[fill,blue] (0,1.3) circle (0.05); 
        \draw[fill,blue] (6.3,1.3) circle (0.05);
        
        \draw[blue] (0,1.1) -- (8.3,1.1);
        \draw[fill,blue] (0,1.1) circle (0.05);
        \draw[fill,blue] (8.3,1.1) circle (0.05);

        \draw[red] (1.3,1.9) --  (6.3,1.9);
        \draw[fill,red] (1.3,1.9) circle (0.05);
        \draw[fill,red] (6.3,1.9) circle (0.05);
        
        \draw[red] (3.3,2.1) --  (8.3,2.1);
        \draw[fill,red] (3.3,2.1) circle (0.05);
        \draw[fill,red] (8.3,2.1) circle (0.05);

    \end{tikzpicture}
    \vspace{1cm}\\
 \begin{tikzpicture}
        \draw (0,0)-- (11,0); 
        
        \foreach \x in {0} {
            \draw (\x,0.1) -- (\x,-0.1) node[below] {$0$}; 
        }
        
        \foreach \x in {1,...,3} {
            \draw (\x*2-0.7,0.1) -- (\x*2-0.7,-0.1) node[below] {$T_{\x}$}; 
        }
        
        \draw (6.3,0.1) -- (6.3,-0.1) node[below] {$T_{1}+\Delta$};
        
        \draw (8.3,0.1) -- (8.3,-0.1) node[below] {$T_{2}+\Delta$};
        
        \draw (10.3,0.1) -- (10.3,-0.1) node[below] {$T_{3}+\Delta$};

        \draw[blue] (0,1.5) -- (1.3,1.5); 
        \draw[fill,blue] (0,1.5) circle (0.05); 
        \draw[fill,blue] (1.3,1.5) circle (0.05);
        \draw[blue] (0,1.3) -- (3.3,1.3); 
        \draw[fill,blue] (0,1.3) circle (0.05); 
        \draw[fill,blue] (3.3,1.3) circle (0.05);
        
        \draw[blue] (0,1.1) -- (5.3,1.1);
        \draw[fill,blue] (0,1.1) circle (0.05); 
        \draw[fill,blue] (5.3,1.1) circle (0.05);

        \draw[red] (1.3,1.9) --  (1.3+5,1.9);
        \draw[fill,red] (1.3,1.9) circle (0.05);
        \draw[fill,red] (1.3+5,1.9) circle (0.05);
        
        \draw[red] (3.3,2.1) --  (3.3+5,2.1);
        \draw[fill,red] (3.3,2.1) circle (0.05);
        \draw[fill,red] (3.3+5,2.1) circle (0.05);
        
        \draw[red] (5.3,2.3) --  (5.3+5,2.3);
        \draw[fill,red] (5.3,2.3) circle (0.05);
        \draw[fill,red] (5.3+5,2.3) circle (0.05);

    \end{tikzpicture}
    \caption{The blue lines denote the time interval where the dynamics of the variance process influence the SPX option up to the maturity time. For instance the shortest blue line denotes the time interval where the dynamics of the variance process enter up to maturity $T_{1}$. Similarly the red lines denote the corresponding ones for the VIX, as for instance the variance process enters here in the time integral on $\lbrack T_{1},T_{1}+\Delta\rbrack$, see \eqref{initial_formula_vix2}.
    On the upper graph a representation of the joint calibration approach~\ref{approach_1} is given where we notice that the maturities of the VIX are chosen so that there is a maximal overlap with the ones of the SPX. On the lower graph a representation of approach~\ref{approach_2} is given where the maturities $\Tcal=\{T_{1},T_{2},T_{3}\}$ are considered. We observe that there is less overlap between the maturities of the SPX and VIX  than in approach \ref{approach_1}.}
\end{figure}
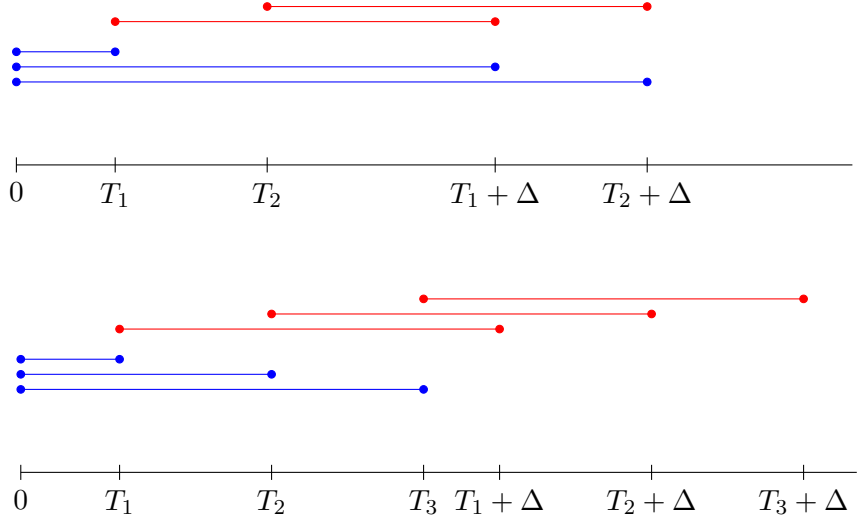

Both approaches deal with the joint modeling of SPX and VIX options and in order to be consistent with both viewpoints taken in the literature, we show how our signature-based model solves the joint calibration within both settings. For this reason we split the rest of the section in two subsection and discuss them separately.

\subsubsection{First approach}
Here we consider call options for both indices on the trading day 02/06/2021, as in \cite{GLJ:22}. Maturities are reported in the following tables with the corresponding range of strikes (in percentage) with respect to the spot and the market's futures prices.
\vspace{0.2cm}
\begin{center}
\begin{tabular}{||c | c ||} 
 \hline
 $T_{1}^{\VIX}=0.0383$ & $T_{2}^{\VIX}=0.0767$\\
 \hline\hline
 (90$\%$,220$\%$) & (90$\%$,220$\%$)\\
 \hline
\end{tabular}
\end{center}
\begin{center}
\begin{tabular}{||c |c| c||} 
 \hline
 $T_{1}^{\SPX}=0.0383$ & $T_{2}^{\SPX}=0.1205$ & $T_{3}^{\SPX}=0.1588$ \\
 \hline\hline
 (92$\%$,105$\%$) & (70$\%$,105$\%$) & (80$\%$,120$\%$)\\
 \hline
\end{tabular}
\end{center}
\vspace{0.2cm}

We stress that the shortest maturity considered is of 14 days for both SPX and VIX, then the second and third maturity of the SPX are 44 days and 58 days, respectively, and the second one for the VIX is 28 days. Moreover, we consider a high moneyness level (up to 220$\%$) for $\VIX$ options, usually rather difficult to fit. Regarding our modeling choice we fix $d=3$, $n=3$ and choose the primary process $X$ to be a three dimensional Ornstein-Uhlenbeck process (see Example~\ref{example_OU}) with parameters
\begin{equation*}
    \kappa=(0.1,25,10)^\top, \qquad \theta=(0.1,4,0.08)^\top,\qquad \sigma=(0.7, 10,5)^\top,
\end{equation*}
\begin{equation*}
   \rho=	\begin{pmatrix}
	1 & 0.213 & -0.576 & 0.329 \\
	\cdot & 1 & -0.044 & -0.549 \\
	\cdot & \cdot & 1 & -0.539 \\
	\cdot & \cdot & \cdot & 1 \\
	\end{pmatrix}, \qquad X_{0}=(1,0.08,2)^\top.
\end{equation*}
Note that with this configuration we need to calibrate $85$ parameters, i.e., $\ell\in \R^{85}$.
Concerning the calibration task, we solve \eqref{L_joint} with $\lambda=0.35$  with $N_{MC}=80000$ Monte Carlo samples for the previous maturities and strikes. Furthermore we specify the loss function $\mathcal{L}$ as in \eqref{greek_loss} for $\beta=1$ both for $\SPX$ and $\VIX$.
\begin{figure}[H]
    \centering
    \includegraphics[scale=0.44]{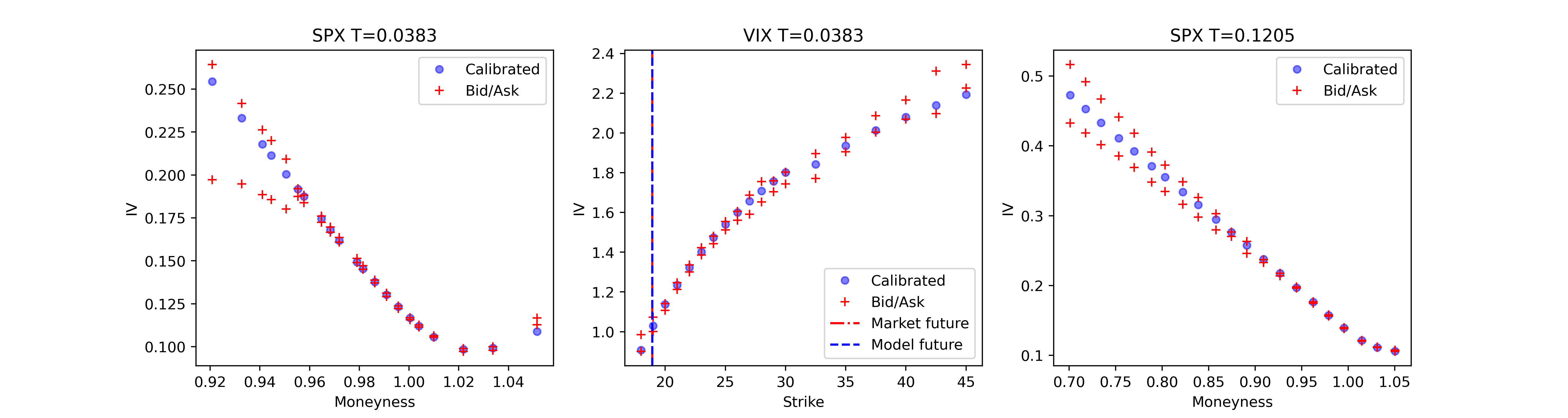}
\end{figure}
\begin{figure}[H]
    \centering
    \includegraphics[scale=0.44]{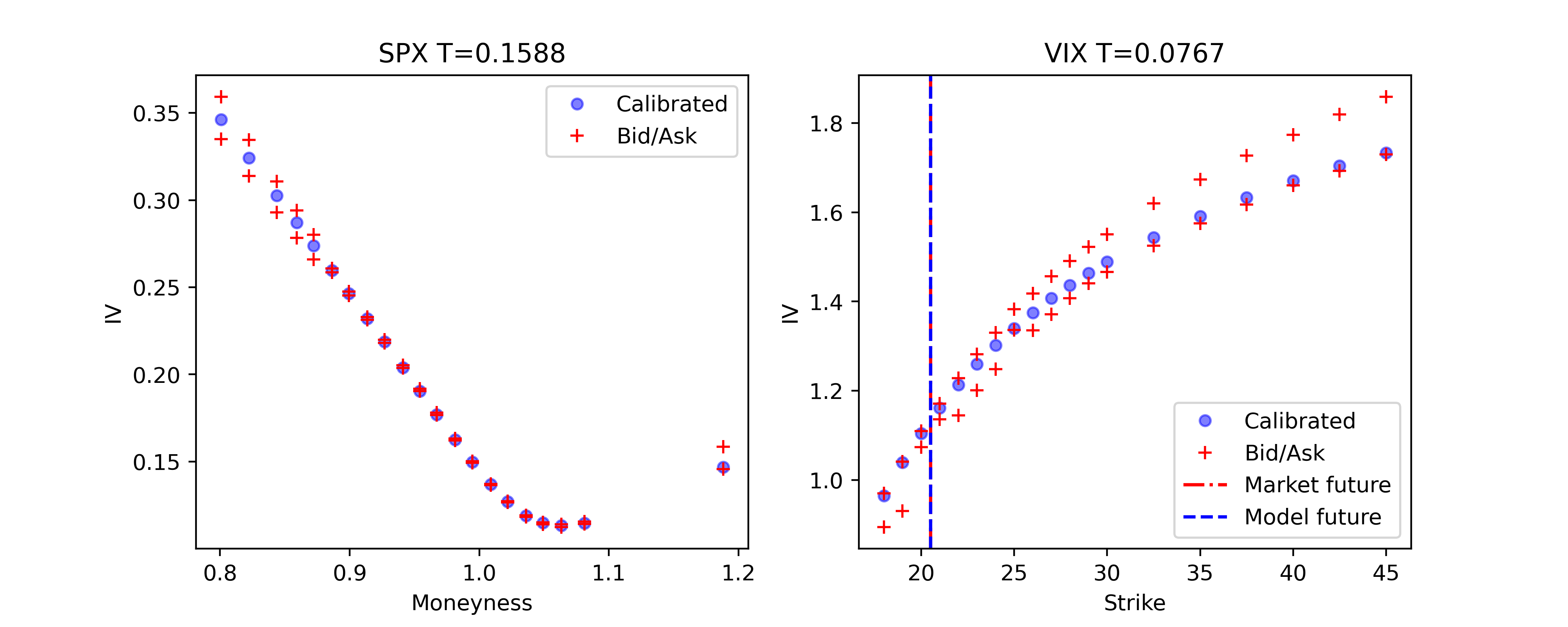}
    \caption[]{In blue the calibrated implied volatility smiles from top-left at maturities $T_{1}^{\SPX}, T_{1}^{\VIX}, T_{2}^\SPX, T_{3}^{\SPX}, T_{2}^{\VIX}$. In red the corresponding bid-ask spreads. In the graphs of the VIX smiles the red dashed line indicates the market future price at maturity and the blue dashed line the calibrated one.}
    \label{fig:joint_constant_ell}
\end{figure}
We report also the relative error between the market futures prices and the calibrated ones
as defined in \eqref{relative_AE}:
\vspace{0.2cm}
\begin{center}
\begin{tabular}{||c| c ||} 
 \hline
 $T_{1}^{\VIX}=0.0383$ & $T_{2}^{\VIX}=0.0767$\\
 \hline\hline
 $\varepsilon_{T_{1}^\VIX}=9.8\cdot 10^{-6}$ &  $\varepsilon_{T_{2}^\VIX}=6.6\cdot 10^{-8}$\\
 \hline
\end{tabular}
\end{center}
\vspace{0.2cm}

\paragraph{Simulation of time-series of SPX and VIX}
Let $\ell^\star\in \R^{85}$ be the  calibrated parameters already used for Figure~\ref{fig:joint_constant_ell}. We then fix $T=60$ days the longest considered maturity for the SPX and  sample a trajectory for $(V_{t}(\ell^\star))_{t\in[0,T]}$, $(\VIX_{t}(\ell^\star))_{t\in[0,T]}$, $(S_{t}(\ell^\star))_{t\in[0,T]}$. Precisely, we sample 12 grid points per day, i.e.~we consider a 2 hours sampling per calendar day, for a total of $N=720$ grid points. The results of this simulation are reported in Figure~\ref{fig:timeseries}.

\begin{figure}[H]
    \centering
    \includegraphics[scale=0.6]{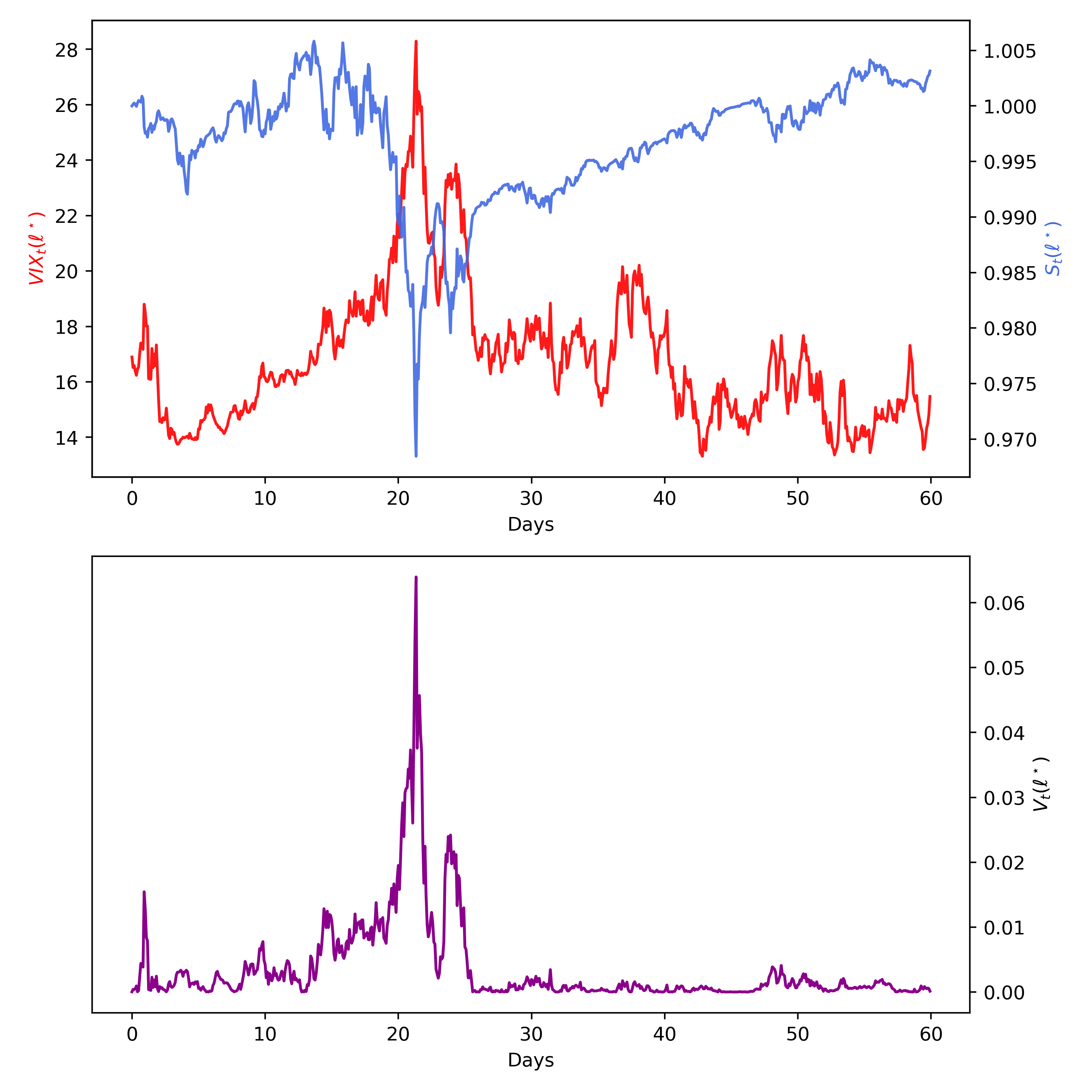}
    \caption{On the top: one realization of the calibrated model $S(\ell^\star)$ for the SPX (in blue) and the corresponding calibrated VIX (in red). On the bottom: the corresponding realization of the calibrated variance process $V(\ell^\star)$.}
\label{fig:timeseries}
\end{figure}

Observe that even though $\ell^*$ was only calibrated to option prices, 
the trajectories produced by the model are economically reasonable and also in line with several stylized facts, such as negative correlation between SPX and VIX or volatility clustering. To obtain the dynamics under the physical measure, these trajectories could still be adjusted by an appropriate market price of risk, but the  quantities which are invariant under equivalent measure changes like the volatility of volatility or the correlation stay the same.

\paragraph{The case of time-varying parameters}

Next, we consider again the case of time-varying parameters as introduced in Section~\ref{TV_VIX} and Section~\ref{TV_SPX} for $\VIX$ and $\SPX$, respectively.
Although the joint calibration is mainly considered for short-dated options in the literature as $\VIX$ options are then more liquid, it is even more challenging to provide an accurate fit for both, short and long maturities. 
Allowing the parameters $\ell$ of our model to depend on time, in particular on the maturities, we are able to calibrate additionally to longer maturities than the ones considered in Figure~\ref{fig:joint_constant_ell}.
We consider for the choice of the primary process the same configuration as we used for Figure~\ref{fig:joint_constant_ell}. The procedure of our time-varying calibration routine is as follows:
\begin{enumerate}[label*=\arabic*.]
    \item Calibrate jointly $T_{1}^{\SPX}, T_{1}^\VIX$ and $T_{2}^\SPX$.
   \item Use the parameters from the calibration of $T_{j}^\SPX$ and $T_{j-1}^\VIX$ to fit jointly the maturities $T_{j+1}^\SPX$ and $T_{j}^\VIX$ for $j=2,\dots,J$.
\end{enumerate}
We consider $J=4$, where the last maturity for the SPX is 170 days, and the last maturity for the VIX is 77 days. For the first two maturities of the SPX and the first of the VIX we consider the same moneyness ranges as in Figure~\ref{fig:joint_constant_ell}, hence we specify here only the ranges for the longer maturities:
\vspace{0.2cm}
\begin{center}
\begin{tabular}{||c | c| c ||} 
 \hline
 $T_{2}^{\VIX}=0.1342$ & $T_{3}^{\VIX}=0.2875$& $T_{4}^{\VIX}=0.3833$ \\
 \hline\hline
(90\%,330\%) & (78\%,395\%) & (80\%,405\%)\\
 \hline
\end{tabular}
\end{center}
\begin{center}
\begin{tabular}{||c |c| c ||} 
 \hline
 $T_{3}^{\SPX}=0.2163$ & $T_{4}^{\SPX}=0.3696$ & $T_{5}^{\SPX}=0.4654$  \\
 \hline\hline
 (75\%,125\%)& (60\%,135\%) & (50\%,145\%)\\
 \hline
\end{tabular}
\end{center}
\vspace{0.2cm}
We observe that for this choice of maturities Assumption~\ref{ass:distance_mat} is satisfied. Hence the second expression for the time-varying $\VIX$ is used from Proposition \eqref{prop:vix_tv}. On the other hand in order to compute the price of the $\SPX$ options in the time-varying case we use the representation of the log-price provided in Proposition~\ref{prop:price_tv}. In \eqref{L_joint}, we employ  $\lambda=0.25$ for each calibration within the rolling procedure and we consider always as loss function $\Lcal^{\beta}$ as introduced in \eqref{greek_loss} for $\beta=0$. It is worth  mentioning that the initial parameter search discussed in Remark \ref{remark:initial_guess}, has been employed for calibrating jointly $T_{1}^{\SPX}, T_{1}^\VIX$ and $T_{2}^\SPX$, whereas for the next slices we have considered the previously calibrated parameters as starting point of the optimization.
\begin{figure}[H]
    \centering
    \includegraphics[scale=0.12]{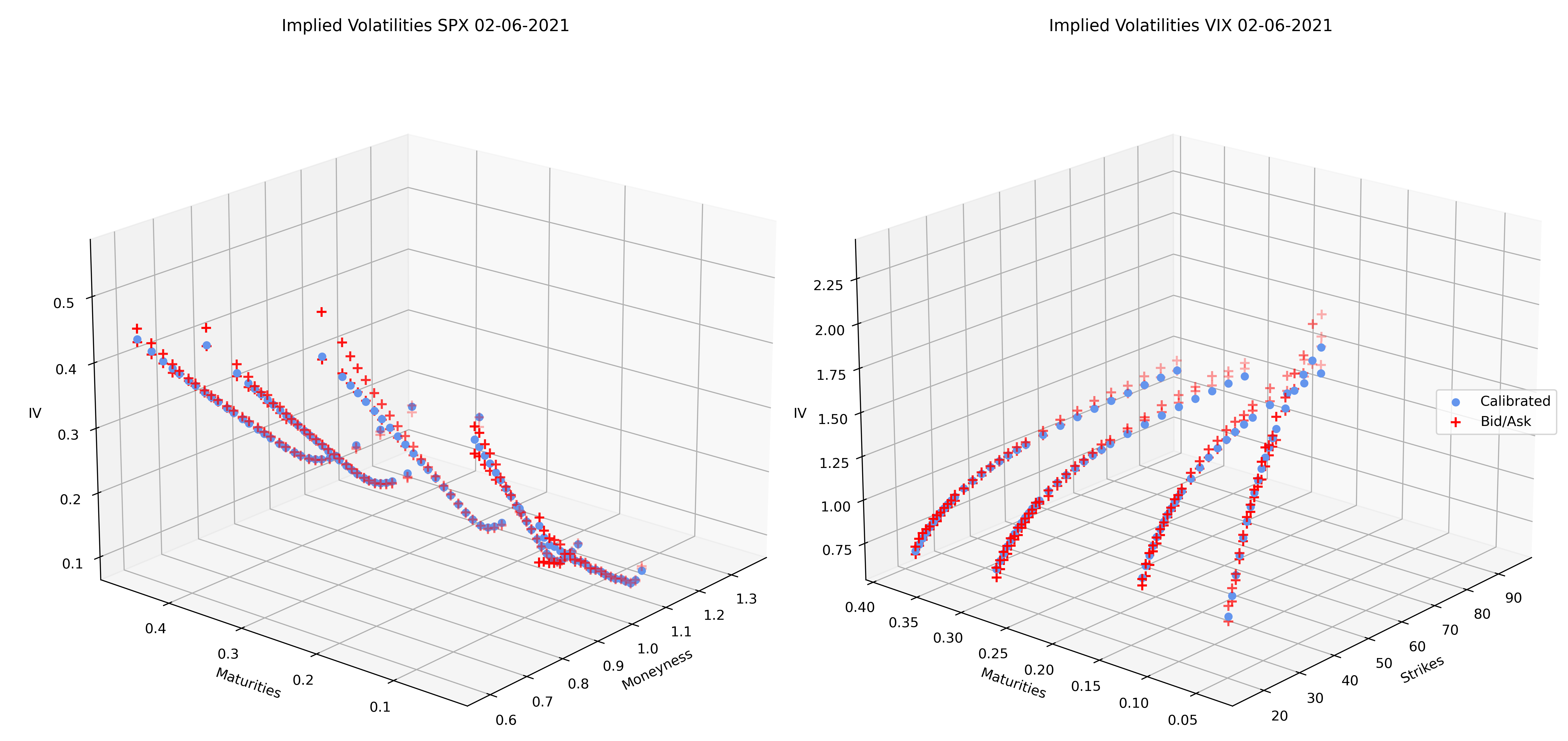}
    \caption{On the left-hand side: SPX smiles, in blue the calibrated implied volatilities and in red the bid-ask spreads. On the right-hand side: VIX smiles, in blue the calibrated implied volatilities and in red the bid-ask spreads.}
    \label{fig:rolling_surfaces}
\end{figure}
Finally we report the absolute relative error on the VIX futures' prices:

\vspace{0.2cm}
\begin{center}
\begin{tabular}{||c | c| c| c ||} 
 \hline
 $T_{1}^{\VIX}=0.0383$ & $T_{2}^{\VIX}=0.1205$ & $T_{3}^{\VIX}=0.1588$ & $T_{4}^{\VIX}=0.2108$  \\
 \hline\hline
 $\varepsilon_{T_{1}^\VIX}=1.2\cdot 10^{-11}$ &  $\varepsilon_{T_{2}^\VIX}=2.3\cdot 10^{-5}$ & $\varepsilon_{T_{3}^\VIX}=3.9\cdot 10^{-5}$& $\varepsilon_{T_{4}^\VIX}=2.9\cdot 10^{-6}$\\
 \hline
\end{tabular}
\end{center}
\vspace{0.2cm}

\subsubsection{Second approach}
Let us now consider the second approach described at the beginning of Section~\ref{numerical_result_joint}. Specifically, we consider a unique set of maturities for both $\SPX$ and $\VIX$ on the trading day of 02/06/2021. For this study, we do not consider time-varying parameters.
In the following table  we report the moneyness ranges for $\SPX$ options in the second row and on the last row the ones for $\VIX$ options:
\begin{center}
\begin{tabular}{||c |c| c| c||} 
 \hline
 $T_{1}=0.0192$ & $T_{2}=0.0383$ & $T_{3}=0.0575$ & $T_{4}=0.0767$\\
 \hline\hline
 (97$\%$,104$\%$) & (92$\%$,105$\%$) & (90$\%$,110$\%$) & (85$\%$,110$\%$)\\
  \hline\hline
 (90$\%$,200$\%$) & (90$\%$,220$\%$) & (90$\%$,290$\%$) & (90$\%$,290$\%$)\\
 \hline
\end{tabular}
\end{center}
We consider $\lambda=0.5$ and as loss function $\mathcal{L}$ we employ \eqref{greek_loss} with $\beta=1$ for VIX options and the same (without futures) for SPX options. 

\begin{figure}[H]
    \centering
    \includegraphics[scale=0.12]{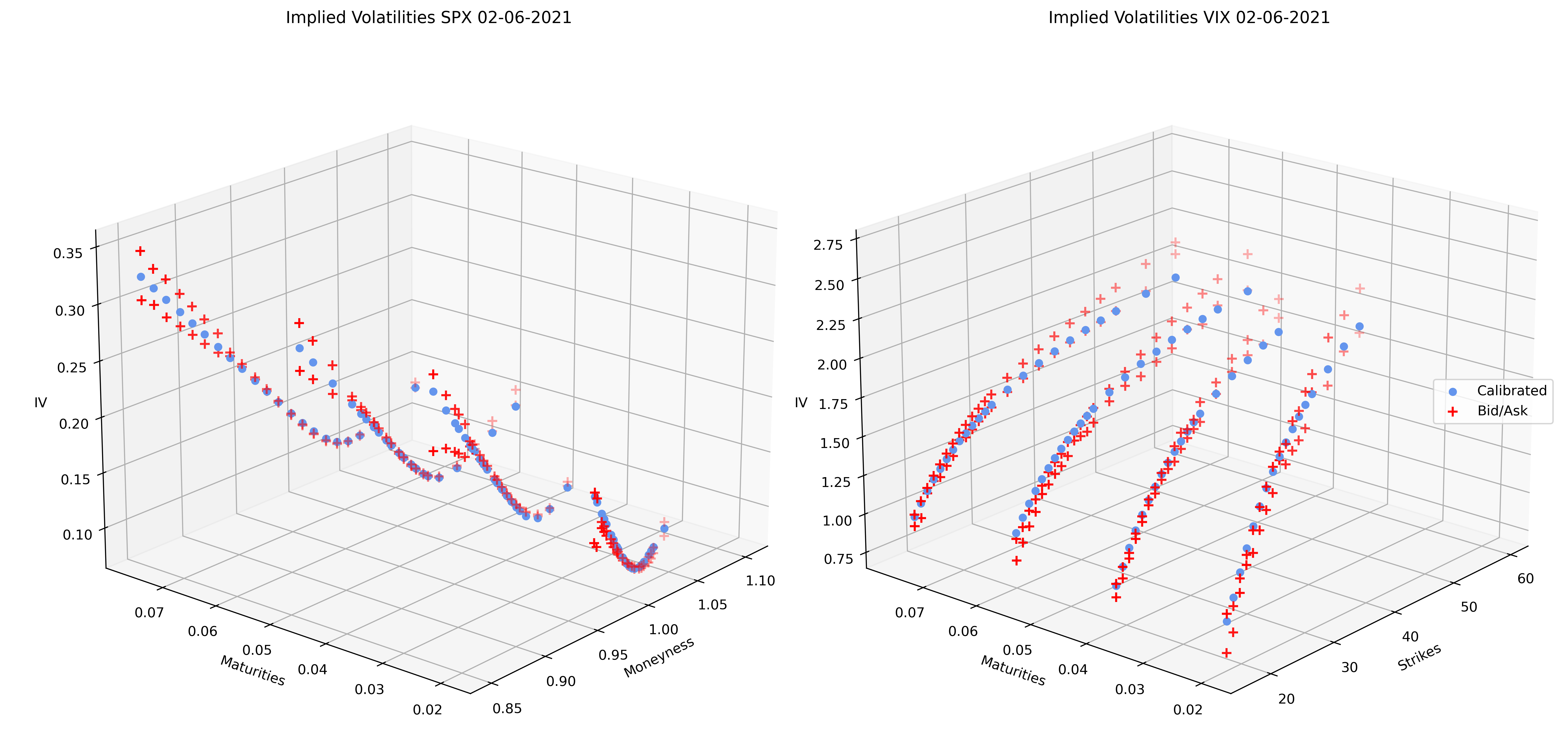}
    \caption{On the left-hand side: SPX smiles, in blue the calibrated implied volatilities and in red the bid-ask spreads. On the right-hand side: VIX smiles, in blue the calibrated implied volatilities and in red the bid-ask spreads.}
    \label{fig:gatheral_surfaces}
\end{figure}
We additionally report the relative error of the calibrated VIX futures:

\vspace{0.2cm}
\begin{center}
\begin{tabular}{||c | c| c| c ||} 
 \hline
 $T_{1}^{\VIX}=0.0192$ & $T_{2}^{\VIX}=0.0383$ & $T_{3}^{\VIX}=0.0575$ & $T_{4}^{\VIX}=0.0767$  \\
 \hline\hline
 $\varepsilon_{T_{1}^\VIX}=6.2\cdot 10^{-3}$ &  $\varepsilon_{T_{2}^\VIX}=1.2\cdot 10^{-5}$ & $\varepsilon_{T_{3}^\VIX}=1.3\cdot 10^{-2}$& $\varepsilon_{T_{4}^\VIX}=1.6\cdot 10^{-3}$\\
 \hline
\end{tabular}
\end{center}
\vspace{0.2cm}

\appendix
\section{Numerical results for the Brownian motion case}\label{appendix}

This appendix is dedicated to the calibration to VIX options only, similarly as in Section~\ref{numerical_result_vix}, however with the primary process $(X_t)_{t\geq 0}$ being simply correlated Brownian motions  (similarly as in \cite{CGS:23}) instead of OU-processes.
 
To be precise, we here model given by \eqref{model}-\eqref{sig-vol2}, where $(X_t)_{t\geq 0}$ is $2$-dimensional Brownian motion. The correlation matrix of $Z=(X,B)$ is specified, as in Section~\ref{numerical_result_vix}, namely by
\begin{equation*}
     \rho=	\begin{pmatrix}
	1 & -0.577 & 0.3 \\
	\cdot & 1 & -0.6\\
	\cdot & \cdot & 1 \\
	\end{pmatrix}.
\end{equation*}
 For the other parameters we consider a truncation's level $n=3$, we sample $N_{MC}=80000$ trajectories for Monte Carlo pricing, and we minimize the loss function \eqref{greek_loss} with $\beta=1$ to fit the same data-set as in Section~\ref{numerical_result_vix}.
\begin{center}
 \begin{figure}[H]
        \centering
		\includegraphics[scale=0.13]{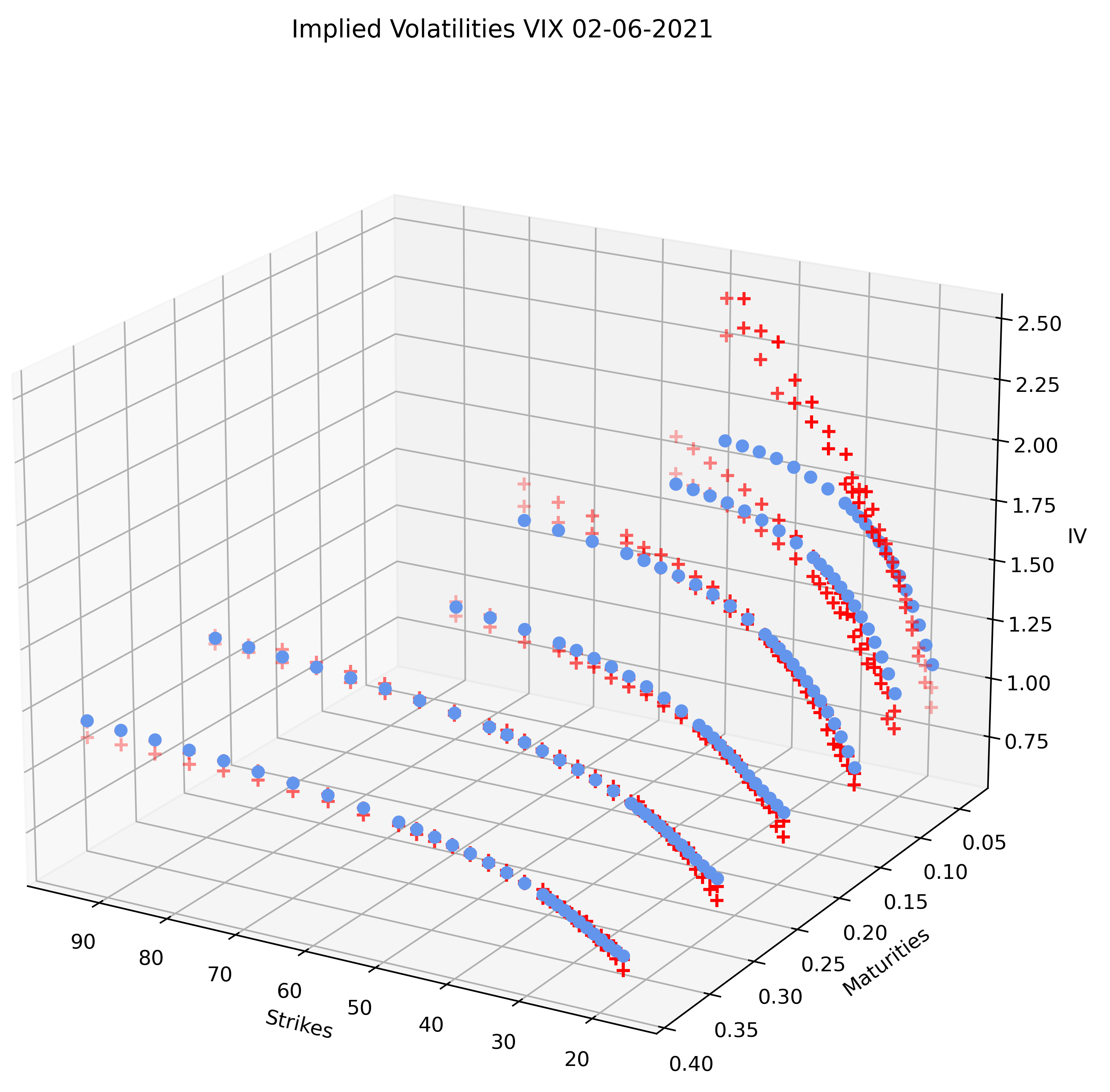}
		\caption{The red crosses denote the bid-ask spreads (of the implied volatilities) for each maturity, while the azure dots denote the calibrated implied volatilities of the model.}
		\label{fig:iv_calib_only_vix_0}
\end{figure}
\end{center}
 We observe that with this specification the model is neither able to calibrate to all future market prices (see Figure~\ref{fig:FTS_only_vix_0} below) nor to fit the market implied volatilites accurately.
 One can indeed see that the model implied volatilities often do not lie within the bid-ask spreads, in particular for high strikes and short maturities.
\begin{center}
\begin{tabular}{||c |c|c ||} 
 \hline
 $T_{1}=0.0383$ & $T_{2}=0.0767$ & $T_{3}=0.1342$\\ 
 \hline\hline
 $\varepsilon_{T_{1}}=2.1 \times 10^{-5}$ & $\varepsilon_{T_{2}}=2.7 \times 10^{-2}$ & $\varepsilon_{T_{3}}=2.1 \times 10^{-7}$\\ 
 \hline
\end{tabular}
\end{center}
\begin{center}
\begin{tabular}{||c |c|c||} 
 \hline
 $T_{4}=0.2108$ & $T_{5}=0.2875$ & $T_{6}=0.3833$\\ 
 \hline\hline
  $\varepsilon_{T_{4}}=1.6  \times 10^{-6}$ & $\varepsilon_{T_{5}}=6.8  \times 10^{-3}$ & $\varepsilon_{T_{6}}=2.1  \times 10^{-3}$ \\ 
 \hline
\end{tabular}
\end{center}

\begin{center}
 \begin{figure}[H]
        \centering
		\includegraphics[scale=0.45]{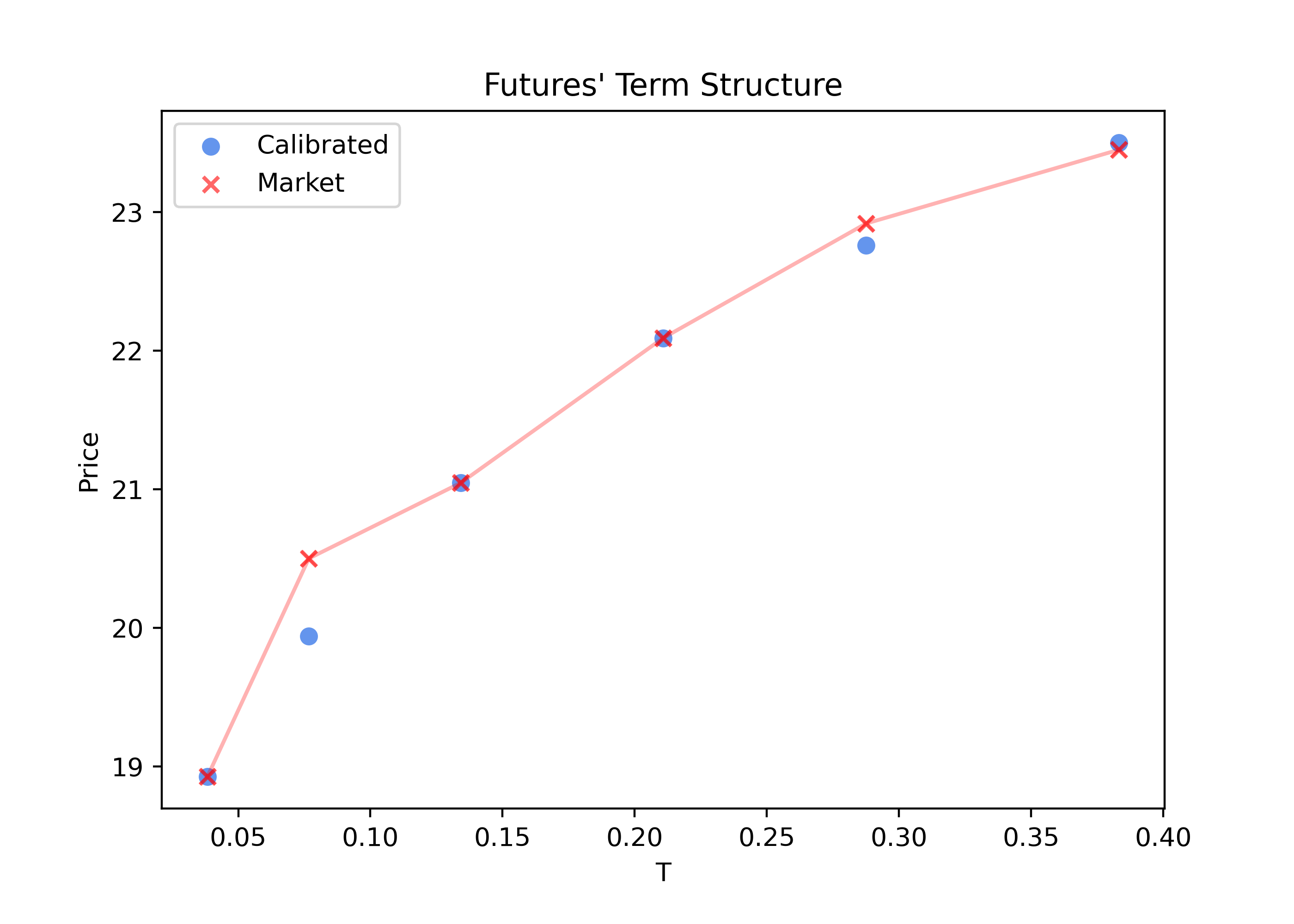}
		
		\caption{The blue circles denote the calibrated futures prices and the red crosses the market futures prices, in between a linear interpolation is applied.}
		\label{fig:FTS_only_vix_0}
\end{figure}
\end{center}

{ 
\section{On the stability of the calibrated parameters} \label{App:B}
We report here an analysis of the stability of the calibrated parameters, which is  for example essential for hedging since unstable parameters can lead to oscillating hedge ratios and high transaction cost. 

 For this purpose we have considered call options with the same time to maturity every trading week in the month of June 2021, i.e.~we take the following five dates: June 2, 2021; June 9, 2021;  June 16, 2021;  June 23, 2021 and  June 30, 2021. The times to maturity are on one hand $T_{1}^{\SPX}=$14 days and $T_{2}^{\SPX}=$44 days for SPX options and on the other hand $T_{1}^{\VIX}=$14 days for VIX (weekly) options. We employ the same primary process as of Section 7.1.1 for all 5 trading days.
 The goodness of fit of the respective implied volatilities is reported in Figures \ref{fig:June_09}--\ref{fig:June_30}, omitting June 2, 2021 as it is already presented in Section~\ref{numerical_result_joint}.  The obtained parameters are reported in Figures \ref{fig:calibrated_parameters1}--\ref{fig:calibrated_parameters3}. The labels of the $x$-axis refer to the signature's index $I$ and on the $y$-axis we have the corresponding coefficients $\ell_{I}$.
 
 Although no regularization on the model parameters has been enforced during the calibration, we observe from Figures \ref{fig:calibrated_parameters1}--\ref{fig:calibrated_parameters3} that most of the parameters which are close to zero in the first trading day, are kept in a neighbourhood of zero in the subsequent trading days. Likewise the more relevant parameters are mostly stable over the trading days or just slightly and continuously adjusted to fit the corresponding smiles, leading to an inherently stable calibration procedure.
 This experiment also indicates that -- even though we consider on purpose an overparametrized model -- it does not suffer from overfitting (which would be the case if the parameters were highly oscillatory).

\begin{figure}[H]
    \centering
    \includegraphics[width=0.8\textwidth]{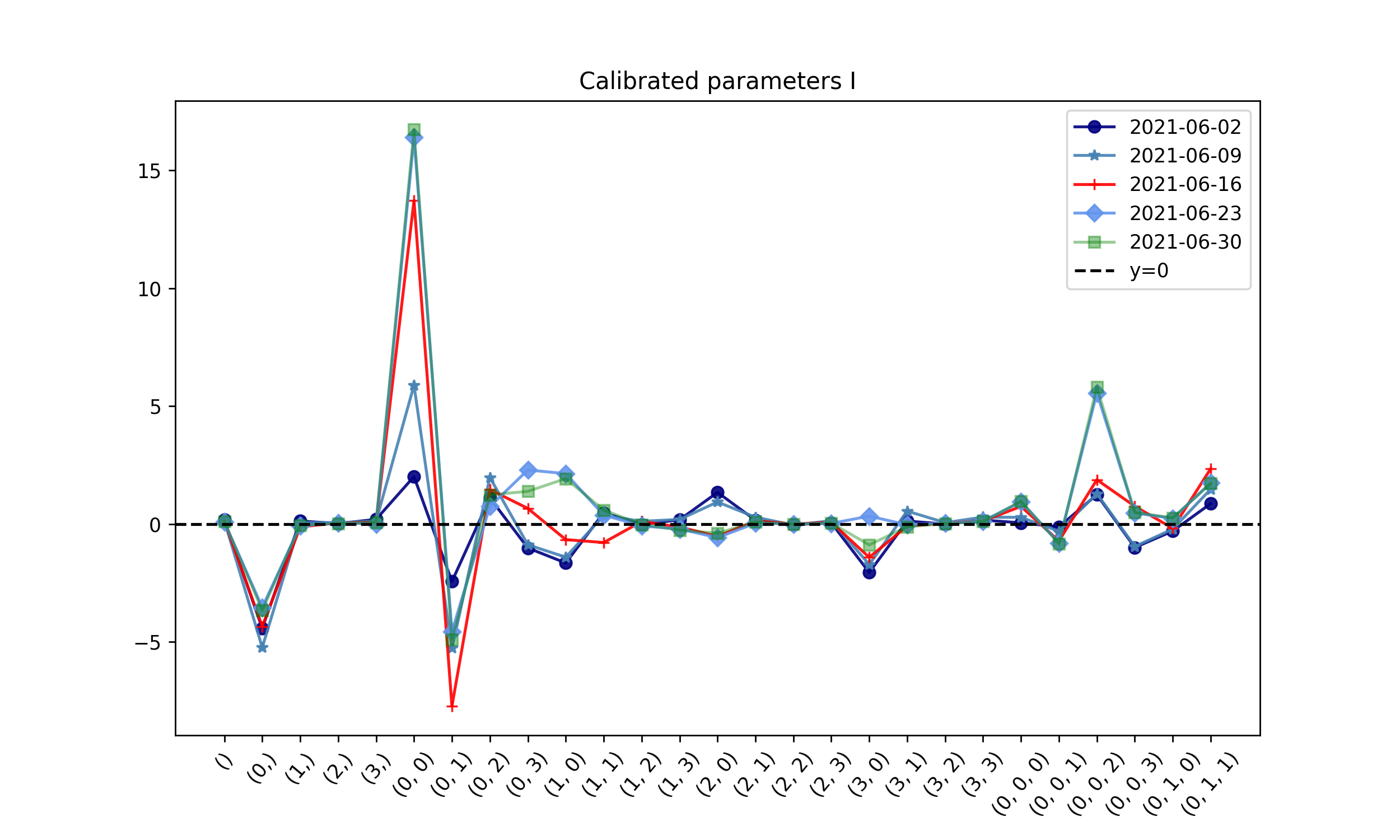}
    \caption{First subset of calibrated parameters $\ell_I$.}
    \label{fig:calibrated_parameters1}
\end{figure}
\begin{figure}[H]
    \centering
    \includegraphics[width=0.8\textwidth]{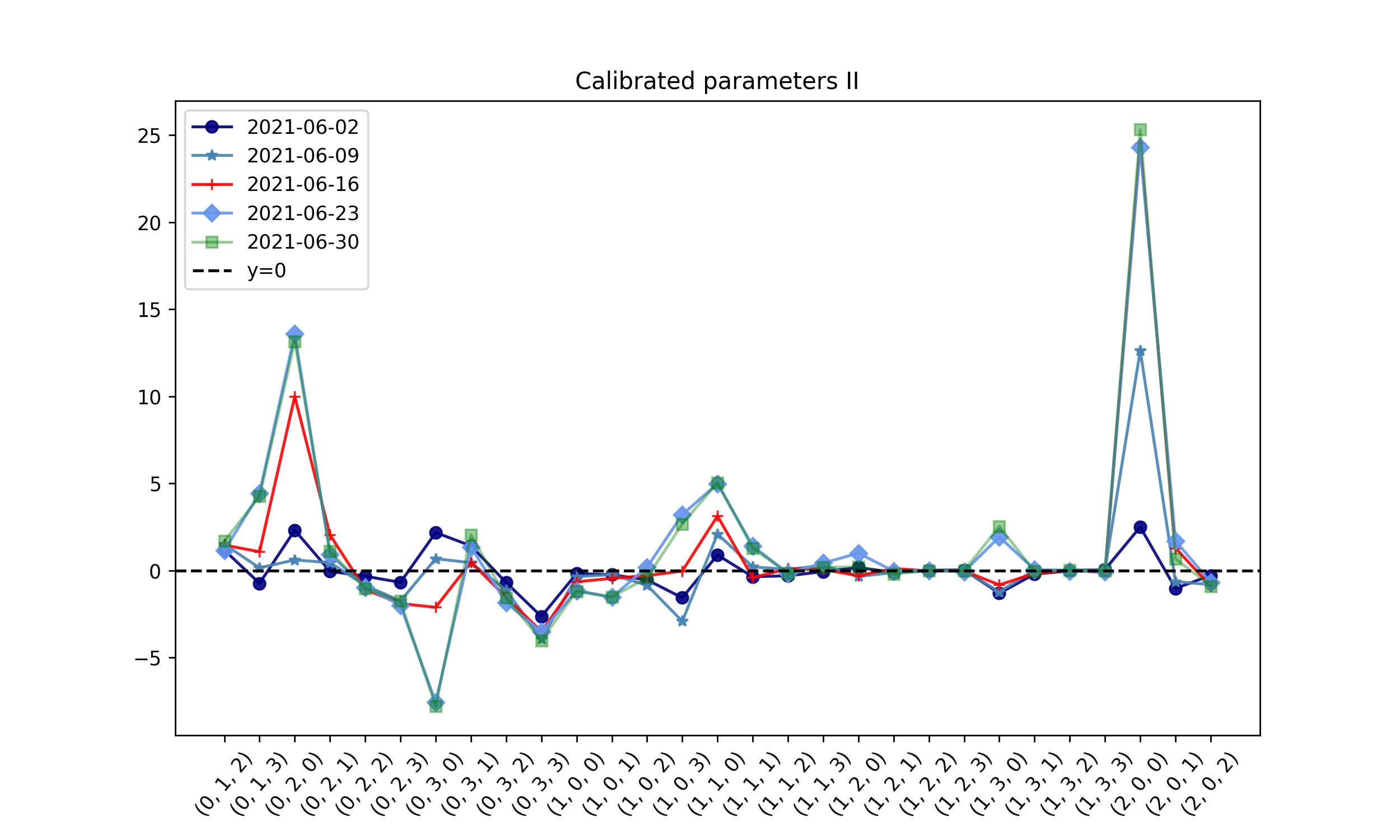}
    \caption{Second subset of calibrated parameters $\ell_I$.}
    \label{fig:calibrated_parameters2}
\end{figure}
\begin{figure}[H]
    \centering
    \includegraphics[width=0.8\textwidth]{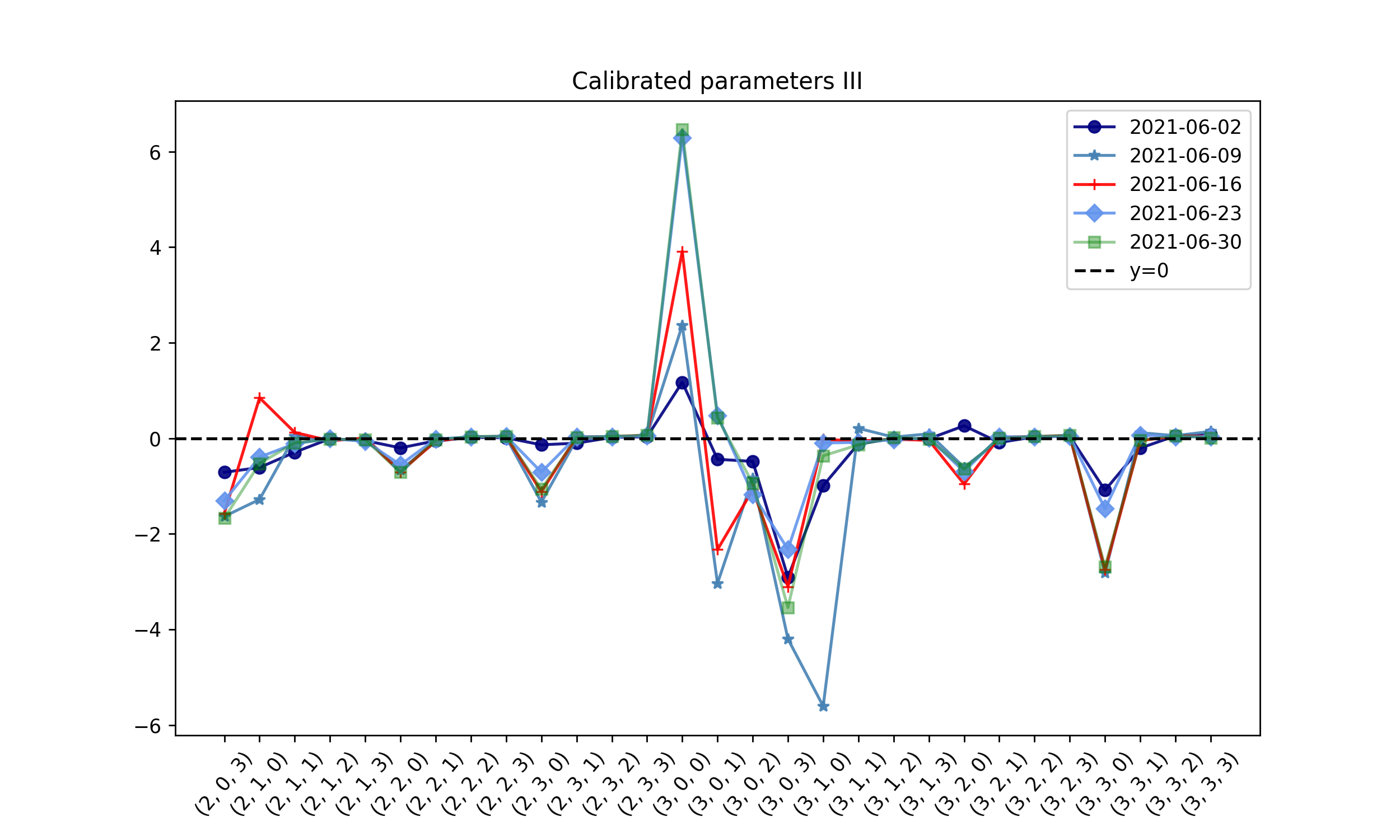}
    \caption{Third subset of calibrated parameters $\ell_I$.}
    \label{fig:calibrated_parameters3}
\end{figure}

\begin{figure}[H]
     \centering
     \begin{subfigure}[H]{0.35\textwidth}
         \centering
         \includegraphics[width=\textwidth]{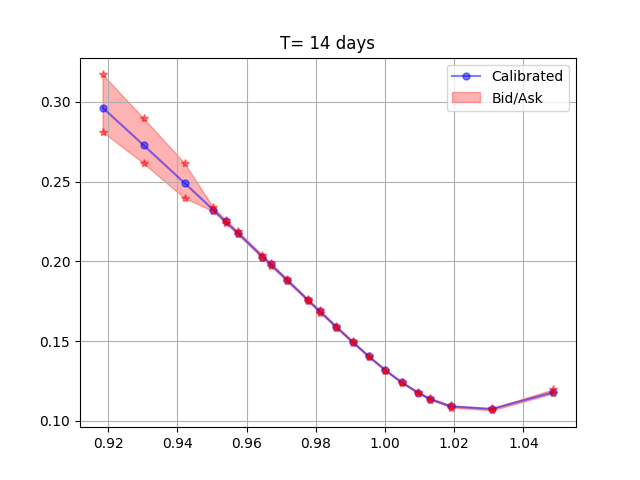}
     \end{subfigure}
     \begin{subfigure}[H]{0.35\textwidth}
         \centering
         \includegraphics[width=\textwidth]{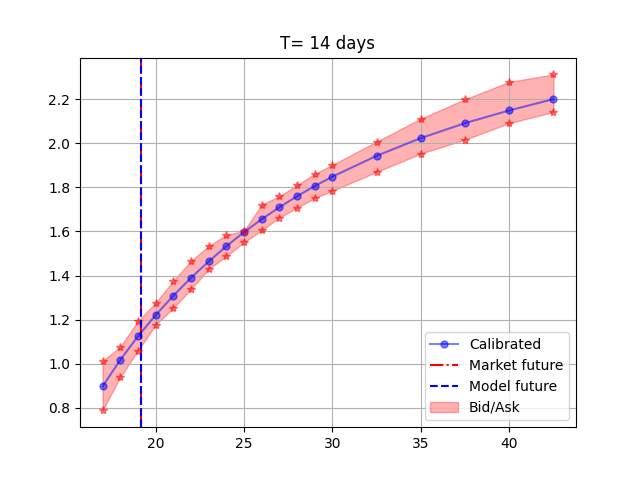}
     \end{subfigure}
     \begin{subfigure}[H]{0.35\textwidth}
         \centering
         \includegraphics[width=\textwidth]{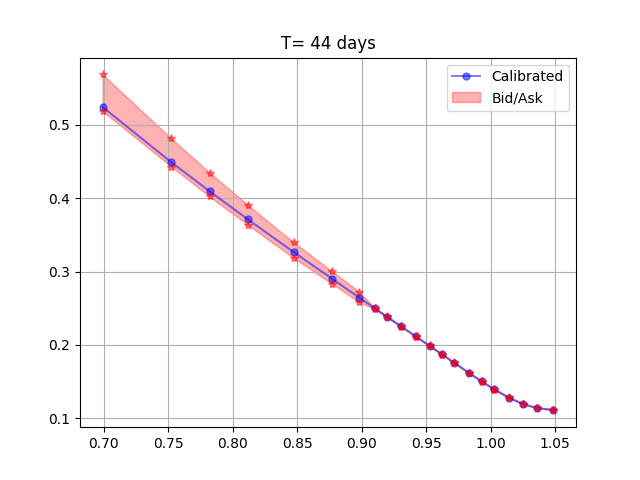}
     \end{subfigure}
        \caption{Implied volatilities as of June 9, 2021. In blue the calibrated implied volatility smiles from top-left at maturities $T_{1}^{\SPX}, T_{1}^{\VIX}, T_{2}^\SPX$. In red the corresponding bid-ask spreads. In the graphs of the VIX smile the red dashed line indicates the market future price at maturity and the blue dashed line the calibrated one.}
        \label{fig:June_09}
\end{figure}

\begin{figure}[H]
     \centering
     \begin{subfigure}[H]{0.35\textwidth}
         \centering
         \includegraphics[width=\textwidth]{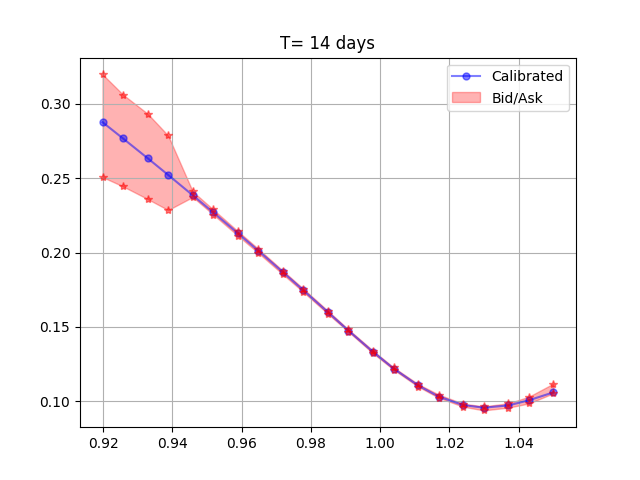}
     \end{subfigure}
     \begin{subfigure}[H]{0.35\textwidth}
         \centering
         \includegraphics[width=\textwidth]{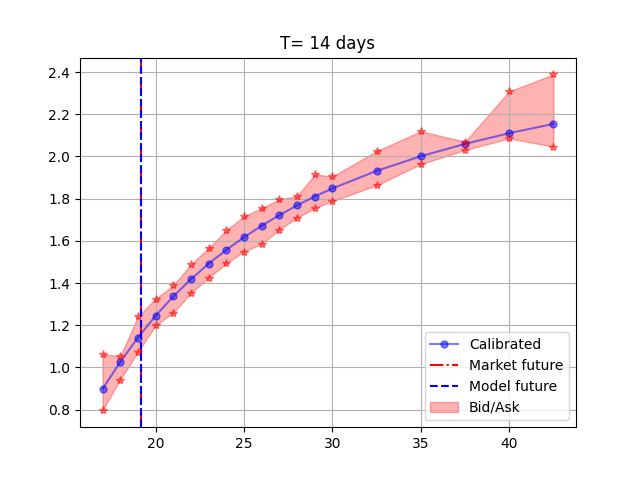}
     \end{subfigure}
     \begin{subfigure}[H]{0.35\textwidth}
         \centering
         \includegraphics[width=\textwidth]{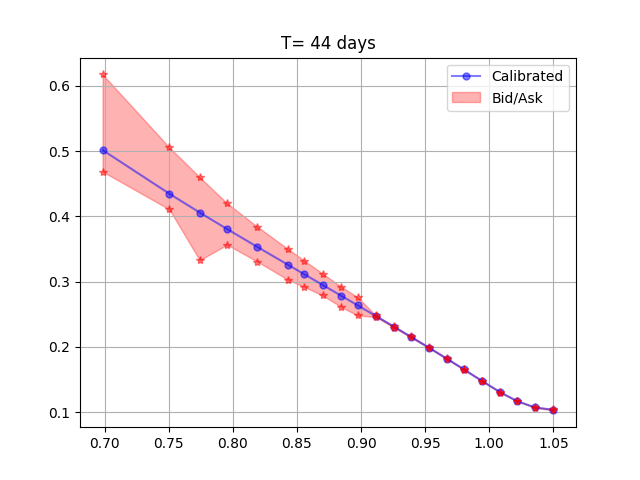}
     \end{subfigure}
        \caption{Implied volatilities as of June 16, 2021}
        \label{fig:June_16}
\end{figure}

\begin{figure}[H]
     \centering
     \begin{subfigure}[H]{0.35\textwidth}
         \centering
         \includegraphics[width=\textwidth]{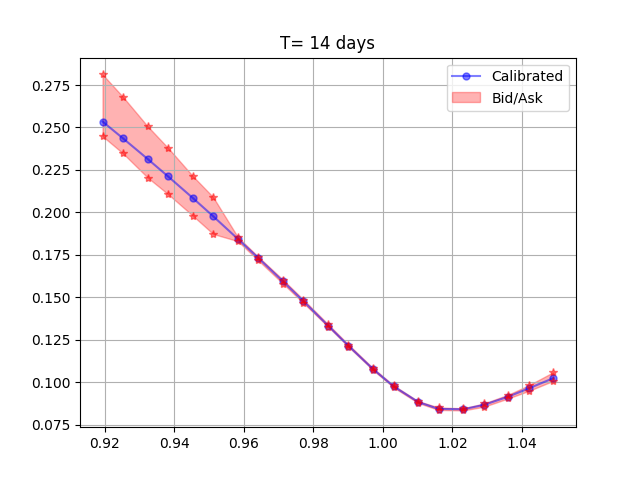}
     \end{subfigure}
     \begin{subfigure}[H]{0.35\textwidth}
         \centering
         \includegraphics[width=\textwidth]{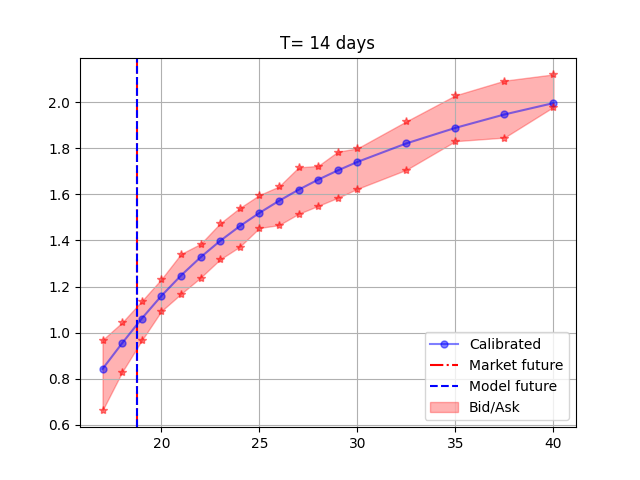}
     \end{subfigure}
     \begin{subfigure}[H]{0.35\textwidth}
         \centering
         \includegraphics[width=\textwidth]{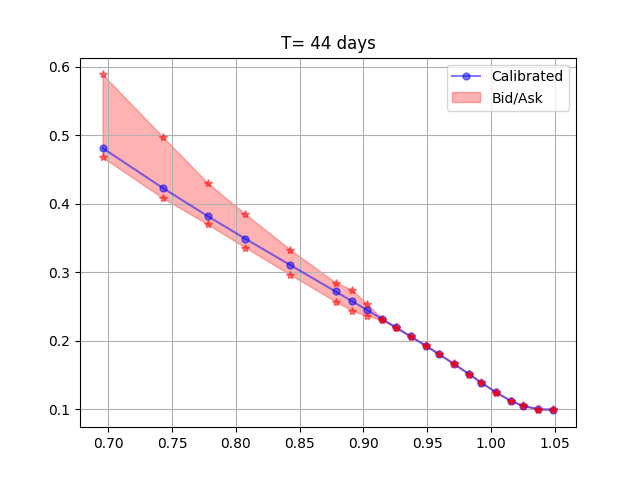}
     \end{subfigure}
        \caption{Implied volatilities as of June 23, 2021}
        \label{fig:June_23}
\end{figure}

\begin{figure}[H]
     \centering
     \begin{subfigure}[H]{0.35\textwidth}
         \centering
         \includegraphics[width=\textwidth]{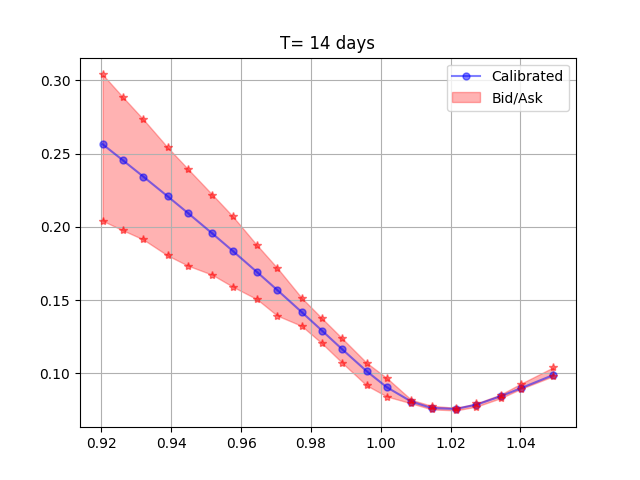}
     \end{subfigure}
     \begin{subfigure}[H]{0.35\textwidth}
         \centering
         \includegraphics[width=\textwidth]{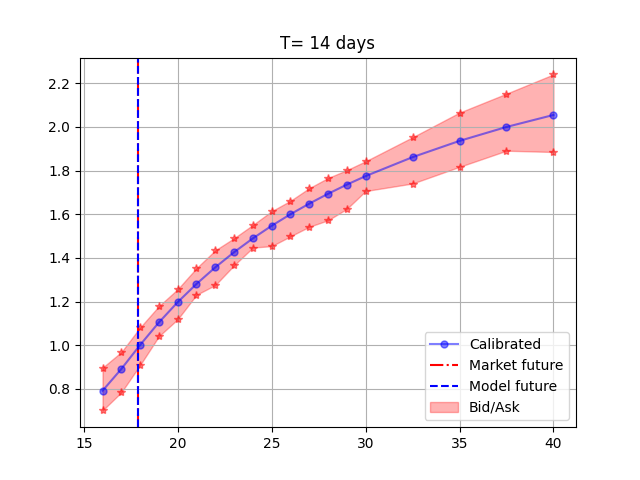}
     \end{subfigure}
     \begin{subfigure}[H]{0.35\textwidth}
         \centering
         \includegraphics[width=\textwidth]{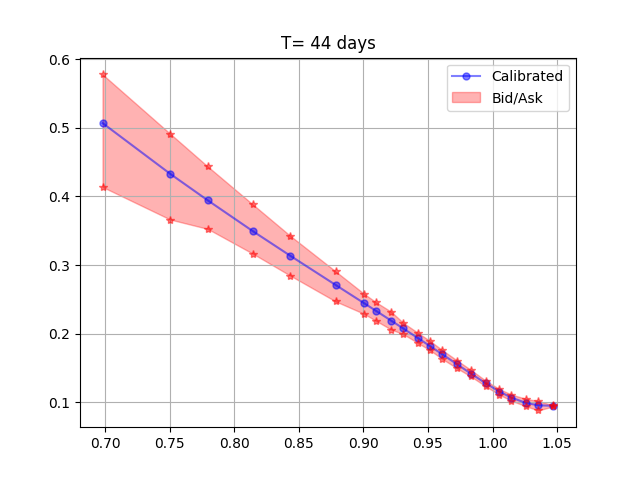}
     \end{subfigure}
        \caption{Implied volatilities as of June 30, 2021}
         \label{fig:June_30}
\end{figure}

}

\end{document}